\documentclass[11pt]{article}
\usepackage[margin=1in]{geometry}
\usepackage{graphicx}
\usepackage{caption}
\usepackage{subcaption}
\usepackage{bm}
\usepackage{mathrsfs}
\usepackage[table,xcdraw,dvipsnames]{xcolor}
\usepackage{amsmath, amsthm, amssymb, amstext}
\usepackage[colorlinks=true,linkcolor=Blue,citecolor=OliveGreen,urlcolor=blue,hypertexnames=false,linktocpage=true]{hyperref}
\usepackage[nameinlink,capitalize]{cleveref}
\usepackage[maxbibnames=99,backend=bibtex,sorting=nyt,style=alphabetic,citestyle=alphabetic, url=false,giveninits=true]{biblatex}
\usepackage{mathtools}
\usepackage{dsfont}
\usepackage[most]{tcolorbox}
\usepackage{changepage}
\usepackage{enumitem}

\newcommand{\veps}{\varepsilon}

\newcommand{\norm}[1]{\left\lVert #1 \right\rVert}
\DeclareMathOperator{\Tr}{Tr}
\newcommand{\ket}[1]{| #1 \rangle}
\newcommand{\bra}[1]{\langle #1 |}
\newcommand{\outerprod}[2]{| #1 \rangle \! \langle #2 |}

\newcommand{\calH}{\mathcal{H}}

\newcommand{\calM}{\mathcal{M}}

\newcommand{\calV}{\mathcal{V}}

\newcommand{\calA}{\mathcal{A}}

\newcommand{\calP}{\mathcal{P}}
\newcommand{\calQ}{\mathcal{Q}}

\newcommand{\ee}{\mathrm{e}}
\newcommand{\ii}{\mathrm{i}}

\newcommand{\rmD}{\mathrm{D}}
\newcommand{\rmU}{\mathrm{U}}
\newcommand{\rmP}{\mathrm{P}}
\newcommand{\calS}{\mathcal{S}}
\newcommand{\rmd}{\mathrm{d}}
\newcommand{\rmL}{\mathrm{L}}
\newcommand{\floor}[1]{\lfloor #1 \rfloor}
\newcommand{\ceil}[1]{\lceil #1 \rceil}
\newcommand{\Pacc}{P_{\textnormal{accept}}}
\newcommand{\Qacc}{Q_{\textnormal{accept}}}
\newcommand{\SG}{\mathfrak{S}}
\newcommand{\PSym}{\Pi_{\textnormal{Sym}}}
\newcommand{\acc}{\textnormal{accept}}
\newcommand{\dTr}{\mathrm{d}_{\textnormal{Tr}}}

\def\ba#1\ea{\begin{align}#1\end{align}}

\DeclarePairedDelimiterX{\infdivx}[2]{(}{)}{%
  #1\;\delimsize\|\;#2%
}

\DeclareMathOperator{\im}{{\textnormal{im}}}
\DeclareMathOperator{\Proj}{\textnormal{proj}}
\DeclareMathOperator{\Span}{\textnormal{span}}
\DeclareMathOperator{\SR}{\textnormal{SR}}
\DeclareMathOperator{\td}{\mathrm{d}_{\textnormal{Tr}}}
\DeclareMathOperator{\diag}{\textnormal{diag}}
\DeclareMathOperator{\LDS}{\textnormal{LDS}}

\newcommand*\encircle[1]{%
  \begin{tikzpicture}[baseline=(C.base)]
    \node[draw,circle,inner sep=1pt](C) {#1};
  \end{tikzpicture}}

\DeclareMathOperator{\expct}{{\mathbb E}}
\DeclareMathOperator{\rank}{rank}

\newtheorem{theorem}{Theorem}[section]
\newtheorem{lemma}[theorem]{Lemma}

\newtheorem{fact}[theorem]{Fact}
\newtheorem{corollary}[theorem]{Corollary}
\newtheorem{proposition}[theorem]{Proposition}

\newtheorem{claim}[theorem]{Claim}
\theoremstyle{definition}
\newtheorem{definition}[theorem]{Definition}

\newtheorem{remark}[theorem]{Remark}

\usepackage{algorithm}
\usepackage{algpseudocode}
\usepackage{microtype}
\usepackage{xargs}
\usepackage{tikz}
\usetikzlibrary{calc,fadings,decorations.pathreplacing,decorations.pathmorphing,shapes,arrows,positioning,cd}

\usepackage{graphicx} 
\usepackage{ytableau}
\ytableausetup{centertableaux}
\renewbibmacro{in:}{%
  \ifentrytype{article}{}{\printtext{\bibstring{in}\intitlepunct}}}
\bibliography{main}
\AtEveryBibitem{\clearfield{month}}
\AtEveryCitekey{\clearfield{month}}

\title{Nearly tight bounds for testing tree tensor network states}
\author{Benjamin Lovitz~\thanks{Department of Computer Science and Software Engineering, Concordia University, Montreal, QC, Canada}\and Angus Lowe~\thanks{Center for Theoretical Physics -- a Leinweber Institute, MIT, Cambridge, MA, USA}}
\date{\today}

\begin{document}

\maketitle

\begin{abstract}
Tree tensor network states (TTNS) generalize the notion of having low Schmidt-rank to multipartite quantum states, through a parameter known as the bond dimension. This leads to succinct representations of quantum many-body systems with a tree-like entanglement structure. In this work, we study the task of testing whether an unknown pure state is a TTNS on $n$ qudits with bond dimension at most $r$, or is far in trace distance from any such state. We first establish that, independent of the dimension of the state, $O(nr^2)$ copies suffice to accomplish this task with one-sided error. We then prove that $\Omega(n r^2/\log n)$ copies are necessary for any test with one-sided error whenever $r\geq 2 + \log n$. In particular, this closes a roughly quadratic gap in the previous bounds for testing matrix product states in this setting. On the other hand, when $r=2$ we show that $\Theta(\sqrt{n})$ copies are both necessary and sufficient for the related task of testing whether a state is a product of $n$ bipartite states having Schmidt-rank at most $r$, for some choice of the qudit dimensions. We also study the performance of tests using measurements performed on a small number of copies at a time. Here, we obtain new bounds for testing rank, Schmidt-rank, and TTNS when the tester is restricted to making measurements on $r+1$ copies of the state.
\end{abstract}

\section{Introduction}
A common task in quantum information science is to test whether an unknown state $\psi$ has a given property. In particular, the problem of testing whether $\psi$ is entangled according to some measure of interest arises in many areas including non-local games \cite{clauser1969proposed,natarajan2017quantumlinearity}, quantum cryptography~\cite{vazirani2019fully, coladangelo2019verifier}, and quantum computational complexity~\cite{harrow2013testing,jeronimo2024dimensionindependentdisentanglersunentanglement}. Moreover, a recent line of work suggests that the hardness of learning and testing states could itself provide a meaningful measure of complexity in quantum systems~\cite{Anshu2023}. Such widespread utility motivates a systematic study of the resources required to perform these tests. In the framework of \textit{quantum property testing}~\cite{Montanaro2016}, one performs a measurement on many identical copies of $\psi$, with the goal to determine whether $\psi$ has the property or is far from any state with the property, typically in trace distance. The minimum number of copies required to distinguish these two cases is called the \textit{copy complexity} of the task.

Similar to property testing of distributions in classical computer science~\cite{CIT-114}, the intention in the quantum case is to design tests which use far fewer copies than would otherwise be required for a complete reconstruction. For testing properties related to the entanglement of pure states, prior work has demonstrated that the copy complexity is often independent of the underlying dimension. This is desirable since for many quantum systems of interest --- e.g., interacting qubits --- the dimension grows exponentially in the size of the system. For example, the product test of Harrow and Montanaro~\cite{harrow2013testing} uses just two copies of an unknown multipartite pure state to determine whether it is a product state or constantly far from the set of all product states. The existence of this test in particular has a number of important implications, including that the complexity class $\mathsf{QMA}(2)$ is equal to $\mathsf{QMA}(k)$ for any $k\geq 2$. Another example comes from the rank test of O'Donnell and Wright~\cite{o2015quantum}, which determines whether a mixed state has rank at most $r$ or is constantly far away from any rank-$r$ state using $O(r^2)$ copies.
This test can be used to determine whether a bipartite pure state is entangled, as measured by its Schmidt-rank, and this strategy was recently shown to be optimal~\cite{chen2024local}. Already for this example, however, there is a gap in our understanding of the copy complexity. Among tests with \textit{one-sided error} --- that is, tests which accept with certainty given a state in the property, AKA perfect completeness --- the copy complexity of rank testing is $\Omega(r^2)$, matching the upper bound. However, for two-sided error the best known lower bound is $\Omega(r)$~\cite{o2015quantum}.

Venturing beyond bipartite entanglement, one may also consider testing multipartite quantum states with a limited, but non-zero amount of entanglement, as suggested in~\cite{Montanaro2016}. Such properties have been considered previously in the literature. For example, Ref.~\cite{Harrow_2017} showed that $O(n)$ copies suffice to test whether a pure $n$-partite state is \textit{not} genuinely multipartite entangled, and recently Ref.~\cite{jones2024testingmultipartiteproductnesseasier} established that this is tight up to log factors. Of most direct relevance to the present work is~\cite{soleimanifar2022testingmps} which studies the task of testing matrix product states (with open boundary conditions). In a matrix product state (MPS) representation, amplitudes are computed by multiplying matrices whose dimension is bounded by a parameter known as the \textit{bond dimension}, denoted by $r$ here and throughout. Such representations are an indispensable tool for studying entanglement in physically-motivated, 1D quantum systems~\cite{cirac2021mps}. Testing whether a state is an MPS of bond dimension $r$ on just two sites recovers the task of Schmidt-rank testing, while an MPS on $n$ sites with bond dimension equal to one is a product state.

For general $n$ and $r$, the MPS test proposed in \cite{soleimanifar2022testingmps} succeeds with one-sided error using at most $O(nr^2)$ copies. However, the authors conjectured that this bound is not tight, as the lower bound shown in their work leaves a gap, stating only that $\Omega(\sqrt{n})$ copies are necessary for $r\geq 2$. Despite follow-up work~\cite{aasonson2024peudoentanglement,chen2024local}, the best-known lower bound as a function of $n$ for one- or two-sided error has remained the same prior to this work. Separately, Ref.~\cite{soleimanifar2022testingmps} proposed studying the copy complexity of testing other examples of tensor network states, which generalize matrix product states to higher-order tensors.

\subsection{Summary of results}\label{sec:results}

In this work we study the task of testing whether an unknown state admits a description as a tree tensor network state (TTNS) with bounded bond dimension, generalizing the MPS testing task considered in Ref.~\cite{soleimanifar2022testingmps}. TTNS are an important extension of MPS which have been employed as representations of ground states on 2D lattices~\cite{tagliacozzo2009simulation}, ansatzes in quantum machine learning~\cite{Huggins2019towards}, as well as in simulations of molecules (e.g., dendrimers) in quantum chemistry~\cite{martin2002dendrimers, Nakatani2013efficient} and of other quantum many-body systems~\cite{verstraete2010simulatingwithttns}, including certain measurement-based quantum computations~\cite{shi2006classicalsimulationttns}. We give new bounds in two settings. In the first, the test may perform arbitrary measurements to accomplish the task, which is the setting considered in prior work. In the second, we consider restricting the test to multiple rounds of measurements performed on only a few copies of the state.

\paragraph{Copy complexity of testing trees}
Given a tree graph $G=(V,E)$ on $n=|V|$ vertices, a bond dimension parameter $r$, and a collection of Hilbert spaces $(\calH_v:v\in V)$, a state $\ket{\psi}\in \bigotimes_{v\in V}\calH_v$ is a \textit{tree tensor network state} (TTNS) if the Schmidt-rank of $\ket{\psi}$ with respect to any bipartite cut determined\footnote{Here, we mean the Schmidt decomposition with respect to the two subsystems defined in the following manner. Since the graph $G$ is a tree, removing any edge $e\in E$ from the graph results in two disjoint connected components. Collecting all the vertices in one of these components produces one subsystem, and the complement of this set of vertices corresponds to the other.} by an edge $e\in E$ is bounded from above by $r$. The resulting set of states can be equivalently defined in terms of contracting tensors placed on the vertices of a tree graph~\cite[Prop.~2.2]{Barthel2022}, with the contractions taking place along \textit{bond indices}, while \textit{physical indices} remain uncontracted. We refer to $\dim(\calH_v)$ for $v \in V$ as the \textit{physical dimensions} of the TTNS. This set is denoted here and throughout by
\begin{align}\label{eq:ttns_def}
    \mathsf{TTNS}(G,r):= \{\outerprod{\psi}{\psi}: \ket{\psi}\in \bigotimes\nolimits_{v\in V}\calH_v,\ \SR_e(\ket{\psi})\leq r\ \forall e\in E\}
\end{align}
where, for each $e\in E$, the notation $\SR_e(\ket{\psi})$ denotes the Schmidt-rank of $\ket{\psi}$ with respect to the bipartite cut determined by $e$.

Our first result is a tight (up to log factors) characterization of the copy complexity of testing trees with one-sided error, which includes MPS testing with one-sided error as a special case, when the bond dimension grows logarithmically in $n$. In the following theorem, we let $\veps$ denote the distance parameter, measured by the standard trace distance, within the property testing framework. (See \Cref{sec:property_testing_defs} for a formal definition of the task we consider.)
\begin{theorem}[Testing trees with one-sided error]\label{thm:main_thm}
    Fix a sequence $(G_n)_n$ of tree graphs on $n\geq 2$ vertices. There exists an algorithm for testing whether $\psi \in \mathsf{TTNS}(G_n,r)$ with one-sided error using $O(nr^2/\veps^2)$ copies of $\psi$. Moreover, in the regime where $\veps\in (0,1/\sqrt{6}]$ and $r\geq 2+\log n$, any algorithm with one-sided error requires $\Omega(nr^2/(\veps^2 \log n))$ copies.
\end{theorem}
The upper bound is proven in \Cref{sec:ttns} and the lower bound is proven in \Cref{sec:lower_bounds}. Ignoring the dependence on $\veps$, the best previous lower bound for MPS testing with one-sided error, due to~\cite{chen2024local}, was $\Omega(\sqrt{nr} + r^2)$. For two-sided error, lower bounds were considered in Refs.~\cite{soleimanifar2022testingmps,aasonson2024peudoentanglement,chen2024local}, with the current best lower bound of $\Omega((\sqrt{nr} + r)/\veps^2)$ proven in \cite{chen2024local}. For high enough bond dimension, our result puts the status of the copy complexity of MPS testing (or TTNS testing) on the same footing as that for rank testing: in either case, there is a (roughly) quadratic gap between the best-known lower bounds for testing with one- versus two-sided error. Up to log factors and for high enough bond dimension, this theorem also falsifies the conjecture from Ref.~\cite{soleimanifar2022testingmps} that the analysis from their work is loose. For constant bond dimension, or two-sided error, however, the copy complexity of MPS or TTNS testing remains open.
\begin{figure}
    \begin{subfigure}[t]{0.5\textwidth}
        \centering
        \begin{tikzpicture}
    
    \draw[line width=1.5pt] (-1,-1) -- (0,0);
    \draw[line width=1.5pt] (-1, 1) -- (0,0);
    \draw[line width=1.5pt] (2, 1)  -- (1,0);
    \draw[line width=1.5pt] (2,-1)  -- (1,0);
    \draw[line width=1.5pt] (0, 0)  -- (1,0);

    \draw[line width=0.5pt] (0,0) -- +(0,0.6);
    \draw[line width=0.5pt] (1,0) -- +(0,0.6);
    \draw[line width=0.5pt] (-1,-1) -- +(0,0.6);
    \draw[line width=0.5pt] (-1, 1) -- +(0,0.6);
    \draw[line width=0.5pt] (2,1) -- +(0,0.6);
    \draw[line width=0.5pt] (2,-1) -- +(0,0.6);

    \draw[fill=black] (0,0) circle (4pt);
    \draw[fill=black] (1,0) circle (4pt);
    \draw[fill=white] (-1,-1) circle (4pt);
    \draw[fill=white] (-1,1) circle (4pt);
    \draw[fill=white] (2,1) circle (4pt);
    \draw[fill=white] (2,-1) circle (4pt);
\end{tikzpicture}
        \caption{\label{fig:tree_like_a}}
    \end{subfigure}
    \begin{subfigure}[t]{0.5\textwidth}
        \centering
        \begin{tikzpicture}
    
    \draw[line width=1.5pt] (-1,-1) -- (0,-0.25);
    \draw[line width=1.5pt] (-1, 1) -- (-0.216,0.125);
    \draw[line width=1.5pt] (2.5, 1)  -- (1.716,0.125);
    \draw[line width=1.5pt] (2.5,-1)  -- (1.5,-0.25);
    \draw[line width=1.5pt] (0.216, 0.125)  -- (1.284,0.125);

    \draw[line width=0.5pt] (0,-0.25) -- +(0,0.6);
    \draw[line width=0.5pt] (-0.216,0.125)  -- +(0,0.6);
    \draw[line width=0.5pt] (0.216,0.125)  -- +(0,0.6);
    \draw[line width=0.5pt] (-1,-1) -- +(0,0.6);
    \draw[line width=0.5pt] (-1, 1) -- +(0,0.6);
    \draw[line width=0.5pt] (1.5,-0.25) -- +(0,0.6);
    \draw[line width=0.5pt] (1.284,0.125) -- +(0,0.6);
    \draw[line width=0.5pt] (1.716,0.125) -- +(0,0.6);
    \draw[line width=0.5pt] (2.5,1) -- +(0,0.6);
    \draw[line width=0.5pt] (2.5,-1) -- +(0,0.6);

    \draw[dashed] (0,0) circle (13pt);
    \draw[dashed] (1.5,0) circle (13pt);

    \draw[fill=white] (-1,1) circle (4pt);

    \draw[fill=white] (-1,-1) circle (4pt);

    \draw[fill=white] (0,-0.25) circle (4pt);
    \draw[fill=white] (-0.216,0.125) circle (4pt);
    \draw[fill=white] (0.216,0.125) circle (4pt);

    \draw[fill=white] (1.5,-0.25) circle (4pt);
    \draw[fill=white] (1.284,0.125) circle (4pt);
    \draw[fill=white] (1.716,0.125) circle (4pt);

    \draw[fill=white] (2.5,1) circle (4pt);

    \draw[fill=white] (2.5,-1) circle (4pt);

\end{tikzpicture}
        \caption{\label{fig:tree_like_b} }
    \end{subfigure}
    \caption{\label{fig:tree_like} (a) Depiction of a TTNS on a graph $G$ with $n=6$ vertices and bond dimension $r$. Physical indices are represented by thin lines and bond indices by thick lines. Here, each empty vertex has physical dimension $d$ and each filled vertex has physical dimension $d^3$. (b) A product of $5$ bipartite states where all the physical dimensions are equal to $d$. The dotted lines are a visual guide to identify Hilbert spaces. If each of the Schmidt-ranks is at most $r$ then this state is an element of $\mathsf{Prod}_2(10,r)\subseteq\mathsf{TTNS}(G,r)$ and the TTNS test accepts. On the other hand, our hard instance for the lower bound takes each of the $5$ bipartite states to be slightly far from Schmidt-rank $r$ (as in~\cref{eq:hard_phi}), and a test with one-sided error struggles to reject.}
\end{figure}

The restriction $r\geq 2+\log n$ for the lower bound in \Cref{thm:main_thm} is a mild one in the sense that the resulting class of states is still efficient to describe, and in practice bond dimensions scaling up to polynomially in $n$ are considered (e.g., for representing the ground states of 1D gapped Hamiltonians~\cite{Arad2017}). Nevertheless, we ask: Is it necessary? Our next result suggests that, if possible, removing this restriction while maintaining a bound linear in $n$ will require analyzing a different class of states than the one considered in this and prior work. Specifically, our lower bound as well as the $\Omega(\sqrt{n})$ lower bounds from Refs.~\cite{soleimanifar2022testingmps} and \cite{chen2024local} rely on analyzing a sub-class of states with the TTNS property: the set of $n/2$-wise tensor products of bipartite states of Schmidt-rank at most $r$. Letting $d$ denote the local dimensions of the sites, we write this set as
\begin{align*}
    \mathsf{Prod}_2(n,r):= \{\outerprod{\psi}{\psi}: \ket{\psi}=\bigotimes\nolimits_{i\in [n/2]}\ket{\psi_i},\ \ket{\psi_i}\in (\mathbb{C}^d)^{\otimes 2}\ \text{and}\ \SR(\ket{\psi_i})\leq r\ \forall i\in[n/2]\},
\end{align*}
and it is contained in $\mathsf{TTNS}(G,r)$ under a suitable choice for the local Hilbert spaces in \cref{eq:ttns_def}, as depicted in \Cref{fig:tree_like}. In \Cref{sec:improved_ub} we prove that, at least for Schmidt-rank $r=2$ and local dimensions equal to $3$, an $\Omega(\sqrt{n})$ lower bound is optimal when considering this class of states.
\begin{corollary}
   There is a universal constant $\delta_0 > 0$ such that, for any fixed distance parameter at most $\delta_0$, the copy complexity of testing $\mathsf{Prod}_2(n,r=2)\subseteq \rmP\left((\mathbb{C}^3)^{\otimes n}\right)$ with one- or two-sided error is $\Theta(\sqrt{n})$.
\end{corollary}
Here, $\rmP((\mathbb{C}^3)^{\otimes n})$ denotes the set of pure states on $n$ qutrits. This is an immediate consequence of \Cref{thm:improved_ub_qutrits} in \Cref{sec:improved_ub}. In the statement above we have suppressed a dependence on the distance parameter $\veps$ for the test. (In fact, as a function of $\veps$, our analysis leaves a gap in the copy complexity of this task.) We also note that the above bound stands in contrast to that obtained using a naive ``test-by-learning" approach. Here, even provided the guarantee that the input state is a product state with respect to the $n/2$ subsystems, learning a description of it to within small trace distance error $\veps$ would seem to require linear in $n$ copies. (See the beginning of \Cref{sec:improved_ub} for an explanation of this point.)

\paragraph{Tests using few-copy measurements}
For practical applications, a drawback of the test used to prove the upper bound in \Cref{thm:main_thm} is that it requires highly entangled measurements (based on the Schur transform) on many copies of the input state. In \Cref{sec:few_copy} we consider testing properties using measurements performed on only a few copies of the input state at a time, once again in the setting of one-sided error. We first show in \Cref{sec:adaptivity} that when performing multiple rounds of measurements, adaptivity offers no advantage in many cases.
\begin{theorem}\label{thm:adaptive}
	Let $\calP$ be a set of pure states that forms an irreducible variety, and let $P_{\textnormal{accept}}$ be a PC-optimal $\ell$-copy test for $\calP$. Then among all adaptive $k$-round, $\ell$-copy tests for $\calP$, the test $(P_{\textnormal{accept}})^{\otimes k}$ is PC-optimal.
\end{theorem}
Here, \textit{PC-optimal} refers to a test which is optimal among all tests with one-sided error (i.e. perfect completeness). See \Cref{sec:property_testing_defs} for a formal definition of optimality, and \Cref{sec:optimal_measurements} for a description of a canonical PC-optimal test for properties of pure states. A \textit{variety} is the common zero locus of a set of homogeneous polynomials, and an \textit{irreducible variety} is a variety that is not the union of two proper subvarieties. Many properties of interest form irreducible varieties, including TTNS. Notable properties that do not form irreducible varieties include the set of pure biseparable states~\cite{Harrow_2017} and any finite set of two or more states, such as the set of stabilizer states~\cite{gross2021schur}.
Determining tasks for which adaptivity may help in learning or testing quantum states has been a subject of recent activity (see, e.g., \cite{chen2022exponential, chen2023doesadaptivityhelpquantum, anshu2022distributed,Flammia2024quantumchisquared}), and this result adds another example to the list of tasks for which it does not.

In \Cref{sec:few_copy_compelxity} we use this result to bound the copy complexity of testing $\mathsf{TTNS}(G,r)$ using $(r+1)$-copy measurements (the fewest number of copies possible for a non-trivial rank test in the one-sided error setting).
\begin{theorem}[Testing trees with few-copy measurements]\label{thm:trees_few_copies}
Let $G=(V,E)$ be a tree on $n$ vertices, $r \geq 2$ be a positive integer, and $\veps \in (0, \frac{1}{r+2})$. There exists an algorithm which tests whether $\psi \in \mathsf{TTNS}(G,r)$ with one-sided error using ${O((n-1)^r (r+1)!/\veps^{2r})}$ copies and measurements performed on just $r+1$ copies at a time. Furthermore, any algorithm with one-sided error which measures $r+1$ copies at a time, possibly adaptively, requires ${\Omega((n-1)^{r-1} (r+1)! / \veps^{2r})}$ copies.
\end{theorem}
When $n=2$, the equivalence between Schmidt-rank and rank testing (see \Cref{sec:sr_rank_testing_equivalence}) implies that the copy complexity of rank testing using adaptive $(r+1)$-copy measurements is precisely $\Theta((r+1)!/\veps^r)$. Despite being substantially higher than the copy complexity in the general setting, we note that the copy complexity using $(r+1)$-copy measurements is still independent of the local dimensions. Hence, such tests may be of practical interest in cases where one wishes to verify that an unknown state can be described by an MPS with a small (constant) bond dimension. For example, if the property of interest is that of being an MPS on $n$ sites with bond dimension at most two, (for a fixed distance parameter) the theorem implies the existence of a test with one-sided error using $O(n^2)$ total copies and measurements.

We prove \Cref{thm:trees_few_copies} using~\Cref{thm:adaptive} and a closed-form expression for the error probability of the PC-optimal $(r+1)$-copy rank test as a function of $\veps$, given in \Cref{sec:closed_form}. Another application of this closed-form expression is to the open problem of finding an improved upper bound on the error probability of the product test, and determining whether this probability tends to $1/2$ as $\veps\to 1$~\cite{harrow2013testing, Montanaro2016, soleimanifar2022testingmps}. In the special case where $r=1$ the closed-form expression applies to the product test on bipartite states, i.e., the SWAP test on one subsystem.

\subsection{Overview of techniques}\label{sec:techniques}
Similar to other works in quantum property testing~\cite{childs2007weakfourierschur,o2015quantum,soleimanifar2022testingmps}, we use techniques from the representation theory of the symmetric and unitary groups to argue that, due to symmetries in the property being considered, it suffices to restrict one's attention to a single, optimal measurement. Let us first describe the basic framework in slightly more detail. Suppose we are interested in a property of bipartite pure states which depends only on the Schmidt coefficients. Then, provided as input the state $\ket{\psi}^{\otimes N}$, the test should not be affected by local unitary transformations of the form $U_A^{\otimes N}\otimes U_B^{\otimes N}$, nor by permutations of the $N$ identical copies of $\ket{\psi}$. Furthermore, since $\ket{\psi}^{\otimes N}$ is in the symmetric subspace, we can without loss of generality focus on the action of a test within this subspace. A standard averaging argument (see, e.g., \cite[Lemma 20]{Montanaro2016}) then leads to the conclusion that, without loss of generality, the optimal measurement for testing this property is a projective measurement of a certain form. Crucially, we may interpret this measurement as leading to perfectly correlated outcomes between measurements performed locally on the two subsystems $A$ and $B$, and so the property may be optimally tested by discarding the $N$ copies of, say the $B$ subsystem. But then this recovers the setting of mixed state testing under unitary symmetry (spectrum testing), for which the optimal measurements --- known as Weak Schur Sampling --- are well-understood. This argument was employed as far back as Ref.~\cite{hayashi2002universal} to study compression of quantum states, and was more recently used for property testing lower bounds in~\cite{chen2024local}. It is also similar to the reasoning behind Theorem~35 in~\cite{soleimanifar2022testingmps}.

The correspondence between spectrum testing and entanglement testing provides an avenue for deriving lower bounds for the properties we consider. Optimal tests for pure state properties have an intuitive description compared to those for mixed states, taking the form of a projection onto a subspace specified by the property (see \Cref{lem:pc_opt_pure_proj_lemma}). Our approach for the lower bound in \Cref{thm:main_thm} leverages this description at first, and uses the correspondence to switch to the mixed state description when it is more convenient for calculations. This allows us to proceed with a direct analysis of the acceptance probability of an optimal test for trees, rather than through an information-theoretic argument based on state discrimination, as in prior work~\cite{soleimanifar2022testingmps,chen2024local}. The task is thus reduced to proving that the rank test of O'Donnell and Wright~\cite{o2015quantum} rejects with too small a probability on a certain hard instance of a far-away state; small enough that amplifying $n$ times does not help, at least for $r$ growing logarithmically in $n$. This requires a fine-grained analysis of the acceptance probabilities of the rank test based on a combinatorial interpretation of Weak Schur Sampling, building on ideas from~\cite{o2015quantum}.

For our proof of~\Cref{thm:adaptive} we make use of an elementary algebraic-geometric characterization of PC-optimal tests for pure states. As we observe in~\Cref{sec:optimal_measurements}, the (essentially unique) PC-optimal $N$-copy test for a property of pure states $\calP$ is to project onto the space $\calP^{N}:=\text{span}(\nu_{N}(\calP))$, where $\nu_{N}(\calP):=\{\psi^{\otimes N} : \psi \in \calP\}$ is the $N$-th \textit{Veronese embedding} of $\calP$. Any variety $\calP$ is isomorphic to its $N$-th Veronese embedding. This allows us to carry over qualities of $\calP$ to qualities of the PC-optimal measurement. We remark that the space $\calP^N$ was recently used to develop semidefinite programming hierarchies for generalized notions of separability~\cite{DJLV24}.

Our other results utilize additional bounds on the the rank test acceptance probabilities under simplifying assumptions about the nature of the unknown state. For example, \Cref{thm:improved_ub_qutrits} relies on a computation enabled by taking the dimension of the state to be equal to $3$, while the closed-form solution for the rank test in \Cref{sec:few_copy} uses the fact that, when measuring $r+1$ copies at a time, the ``reject" measurement operator is equal to the projector onto the antisymmetric subspace.

\subsection{Open questions}
Our work leaves many questions open. Firstly, what is the copy complexity of testing MPS or TTNS with one-sided error and \textit{constant} bond dimension? Could $O(\sqrt{n})$ copies suffice? If instead the $O(n)$ upper bound is tight, \Cref{thm:improved_ub_qutrits} suggests that analyzing a different hard instance --- one which is \textit{not} a product of bipartite states --- may be required to prove it. It would also be interesting to see if the lower bound of $\Omega(\sqrt{n})$ for MPS and TTNS testing can be improved in the setting of two-sided error. Since our analysis relies on identifying a unique optimal measurement for one-sided error tests, we believe new techniques will be required to pursue this direction.

Can our techniques be extended to tensor networks on an arbitrary graph $G$? Intriguingly, our lower bound argument in \Cref{sec:lower_bounds} only makes use of the assumption that the graph $G$ is a tree in one step; that is \Cref{lem:overlap_multiplicative_lemma}, which shows the best approximation to a product of bipartite states by a tree is simply to take the product of the best individual approximations. Does a version of this lemma continue to hold for graphs with cycles? If so, this could imply that testing other tensor network states, such as PEPS or MPS with closed boundary conditions, has a copy complexity that scales at least linearly in the number of edges. Since little is known about the copy complexity of testing or learning more complicated tensor network states, this would be an interesting direction to pursue.


Finally, in the few-copy setting, it may be fruitful to consider the performance of tests which measure greater than $r+1$ copies at a time, interpolating between the two measurement scenarios that we have considered in this work.

\section{Preliminaries}
\subsection{Notation}\label{sec:notation}
\paragraph{Sets, ordering}
For any positive integer $a$ we let $[a]$ denote the set $\{1,\dots, a\}$. Let $\calH$, $\calH^\prime$ be finite-dimensional Hilbert spaces. We denote by $\rmL(\calH,\calH^\prime)$ the set of linear maps from $\calH$ to $\calH^\prime$, $\rmL(\calH)$ the set of linear maps from $\calH$ to $\calH$, $\rmD(\calH)\subset \rmL(\calH)$ the set of quantum states (positive semidefinite operators of unit trace) on $\calH$, and $\rmP(\calH)\subset \rmD(\calH)$ the subset of pure states (rank-1 quantum states) on $\calH$. We denote by $\rmU(\calH)$ the set of unitary operators on $\calH$. Also, we let $\vee^N(\calH)$ be the symmetric subspace of $\calH^{\otimes N}$. For any two operators $X,Y\in\rmL(\calH)$, the notation $X\preceq Y$ means $Y-X$ is positive semidefinite. Throughout, we let $\SG_N$ denote the symmetric group of order $N$.

\paragraph{Bra-ket notation, adjoints}
For a finite-dimensional Hilbert space $\calH$, we denote vectors (not necessarily normalized) in $\calH$ using ket notation $\ket{\psi}$. For a vector $\ket{\psi} \in \calH$, let $\bra{\psi} \in \calH^*$ be defined by $\bra{\psi}= (\ket{\psi},-)$. Here, $(\cdot, \cdot):\calH \times \calH \rightarrow \mathbb{C}$ denotes the inner product on $\calH$, which we take to be antilinear in the first term. In coordinates, $\bra{\psi}$ is the conjugate-transpose of $\ket{\psi}$.
For a unit vector $\ket{\psi}$, it holds that  $\outerprod{\psi}{\psi}\in\rmP(\calH)$ is a pure state. We therefore sometimes also refer to unit vectors as pure quantum states. We use the shorthand $\psi \equiv \outerprod{\psi}{\psi}$ when it has been established that $\ket{\psi}$ is a unit vector in $\calH$. We refer to the quantity $|\langle \phi | \psi\rangle|$ as the \textit{overlap} between the vectors $\ket{\phi}$ and $\ket{\psi}$. For a bipartite quantum state $\ket{\psi}$, we use the notation $\SR(\ket{\psi})$ to refer to the Schmidt-rank (number of non-zero terms in the Schmidt decomposition) of $\ket{\psi}$.

For an operator $A \in \rmL(\calH,\calH')$, let $A^{\dag} \in \rmL(\calH',\calH)$ be the adjoint map defined by $( \ket{\psi}, A \ket{\phi}) = (A^{\dag} \ket{\psi}, \ket{\phi})$ for all $\ket{\psi} \in \calH$, $\ket{\phi}\in \calH'$.

\paragraph{Random variables, distributions}
We denote random variables using bold font, e.g., $\bm{x}$, $\bm{y}$, etc. If $\bm{x}$ is a real-valued random variable we write $\bm{x}\in\mathbb{R}$ and similarly for other sets. If $\bm{x}$ is a discrete random variable over a finite alphabet $[d]$, we identify its distribution with a probability vector $p\in [0,1]^d$ such that $p_i\geq 0$, $\sum_i p_i = 1$. We then write $\bm{x}\sim p$, and $(\bm{x}_1,\bm{x}_2,\dots,\bm{x}_N)\sim p^{\otimes N}$ to mean $\bm{x}_1,\bm{x}_2,\dots,\bm{x}_N$ are $N$ i.i.d.\ copies of $\bm{x}$. We let $\textnormal{Unif}(\calS)$ denote the uniform distribution over the set $\calS$.

\subsection{Property testing of quantum states}\label{sec:property_testing_defs}
In the setting of property testing of quantum states, one considers a set of quantum states $\calS$, with $\calS=\rmD(\calH)$ for mixed state testing and $\calS=\rmP(\calH)$ for pure state testing, as well as a distance measure $\rmd:\calS\times\calS\to \mathbb{R}$. For the most part we consider pure state testing. In this paper, the distance measure $\rmd$ is the standard trace distance between states, $\rmd=\rmd_{\textnormal{Tr}}: \rho\times \sigma\mapsto \frac{1}{2}\norm{\rho-\sigma}_1$. A \textit{property} $\calP$ is a subset of $\calS$, and a state $\rho\in \calS$ is said to be \textit{$\veps$-far} from $\calP$ if $\min_{\sigma\in \calP}\rmd(\rho,\sigma)\geq \veps$. An \textit{$N$-copy test} for $\calP$ is a measurement operator $\Pacc\in\rmL(\calH^{\otimes N})$,  $0\preceq \Pacc\preceq \mathds{1}$ such that, given a distance $\veps\in (0,1]$,
\begin{align*}
      \inf_{\rho\in \calP}\Tr(P_{\textnormal{accept}}\rho^{\otimes N})= a\quad \textnormal{and}\quad
     \sup_{\rho \ \veps\textnormal{-far from}\ \calP} \Tr(P_{\textnormal{accept}}\rho^{\otimes N})= b
\end{align*}
for some $1\geq a \geq b \geq 0$. In such a case we say the test has completeness-soundness (CS) parameters $(a,b)$. Note that only $b$ depends on the distance $\veps$. We say the test is \textit{successful} for the distance $\veps$ if $a-b \geq 1/3$, and has \textit{one-sided error} or \textit{perfect completeness} if $a=1$. In the setting of perfect completeness, there is a natural notion of optimality.
\begin{itemize}
    \item A test is \textit{($N$-copy) PC-optimal} for a property $\calP$ if it has perfect completeness and for every $\veps\in (0,1]$ its soundness parameter is at most that of any other test with perfect completeness.
    \item A test $\Pacc$ is \textit{($N$-copy) strongly PC-optimal} for a property $\calP$ if it is has perfect completeness and for any other $N$-copy test $Q$ with perfect completeness it holds that $\Tr(P_{\textnormal{accept}}\rho^{\otimes N})\leq  \Tr(Q \rho^{\otimes N})$ for any $\rho\in \calS$ which is not in $\calP$.
\end{itemize}
Every strongly PC-optimal test is also PC-optimal. Finally, given a POVM $\calM=\{M_1,\dots,M_k\}\subset \rmL(\calH^{\otimes N})$, we say that $\calM$ is an \textit{optimal measurement} for $\calP$ if for any $\veps\in (0,1]$ the following holds: for any $N$-copy test with CS parameters $(a,b)$ there is an $N$-copy test $\Pacc$ of $\calP$ with CS parameters $(a^\prime,b^\prime)$ such that  $a^\prime\geq a$, $b^\prime\leq b$, and $P_{\textnormal{accept}}=\sum_{j=1}^k\alpha_j M_j$ for some $\alpha_j\in [0,1]$. An operational interpretation of this definition is that $\calM$ is optimal if the performance of any test can be replicated by performing this measurement and classically post-processing the outcomes.

\subsection{Optimal measurements of pure states from symmetries}\label{sec:optimal_measurements}
In this section we describe optimal measurements for testing properties of quantum states under various assumptions. In \Cref{sec:pc_optimal_pure_state} we characterize the PC-optimal tests for pure state properties. In \Cref{sec:schur_weyl} we determine a restricted form for the optimal measurements when the property of interest is invariant under local unitary operations. Finally, in \Cref{sec:sr_rank_testing_equivalence} we describe how Schmidt-rank testing and rank testing are equivalent tasks, which is a key step in the proof of the lower bound for testing trees in \Cref{sec:lower_bounds}.

\subsubsection{Pure state tests with perfect completeness}\label{sec:pc_optimal_pure_state}
In this section we characterize the PC-optimal tests for pure state properties. Let $\calH$ be a finite-dimensional Hilbert space and $N$ be a positive integer. Consider an $N$-copy test $\Pacc\in\rmL(\calH^{\otimes N})$ for some pure state property $\calP\subseteq \rmP(\calH)$. Upon receiving $N$ copies of an unknown pure state $\psi\in\rmP(\calH)$, the test accepts with probability equal to $\Tr(\Pacc \psi^{\otimes N})=\Tr(\Pi_{\textnormal{Sym}}\Pacc \Pi_{\textnormal{Sym}} \psi^{\otimes N})$, where $\Pi_{\textnormal{Sym}}:\calH^{\otimes N}\to \vee^N(\calH)$ denotes the projector onto the symmetric subspace. This motivates defining the following equivalence relation: we say that two tests $P, Q\in\rmL(\calH^{\otimes N})$ are \textit{equivalent within the symmetric subspace} if $\Pi_{\textnormal{Sym}}P \Pi_{\textnormal{Sym}} = \Pi_{\textnormal{Sym}}Q\Pi_{\textnormal{Sym}}$.
\begin{lemma}\label{lem:pc_opt_pure_proj_lemma}
    For any pure state property $\calP\subseteq \rmP(\calH)$, the $N$-copy test
    \begin{align*}
        \Pacc &:= \Proj (\Span \{\ket{\phi}^{\otimes N}: \outerprod{\phi}{\phi}\in\calP\})
    \end{align*}
    is strongly PC-optimal. Moreover, up to the equivalence defined above, it is the unique such test.
\end{lemma}
\begin{proof}
    Let $0\preceq Q_{\textnormal{accept}}\preceq\mathds{1}$ be another test for $\calP$ with one-sided error. This means that, for every $\phi\in \calP$, we have $\Tr(Q_{\acc}\phi^{\otimes N})=1$ and, since the maximum eigenvalue of $Q_{\acc}$ is $1$, it holds that $\ket{\phi}^{\otimes N}$ is an eigenvector of $Q_{\acc}$ with eigenvalue $1$. By linearity, any $\ket{\chi}\in\calV:=\Span \{\ket{\phi}^{\otimes N}: \outerprod{\phi}{\phi}\in\calP\}$ is also an eigenvector of $Q_{\acc}$ with eigenvalue $1$. Therefore, we can write $Q_{\acc}=P_{\acc} + M_{\calV^{\perp}}$ for some $M_{\calV^{\perp}}\in \rmL(\calV^{\perp})$ with eigenvalues between $0$ and $1$. This implies that $\Pacc$ is strongly PC-optimal. Furthermore,
    \begin{align*}
        \PSym\Qacc\PSym &= \Pacc + \PSym M_{\calV^\perp}\PSym.
    \end{align*}
    Suppose that $\Qacc$ is strongly PC-optimal, but $\PSym M_{\calV^\perp}\PSym\neq 0$. Then we can set $\ket{\xi}$ to be any state such that $\ket{\xi}^{\otimes N}$ is not in the kernel of $\PSym M_{\calV^\perp}\PSym$. That it is possible to pick such a state follows from the fact that the symmetric subspace is spanned by the set of all $N$-wise tensor product states. Then
    \begin{align*}
        \Tr(\Qacc\xi^{\otimes N}) &= \Tr(\Pacc\xi^{\otimes N}) + \Tr(M_{\calV^\perp}\xi^{\otimes N}) > \Tr(\Pacc\xi^{\otimes N}).
    \end{align*}
    This is a contradiction with the definition of strong PC-optimality and completes the proof.
\end{proof}
Tests of the same form as the operator $\Pacc$ in the above lemma have been used in several prior works~\cite{wang2011property,harrow2013testing,Montanaro2016}, and in particular were used to prove optimality of the product test in~\cite{harrow2013testing} for the case of one-sided error.

\subsubsection{Schur-Weyl duality applied to bipartite systems}\label{sec:schur_weyl}
This section reviews some of the facts from representation theory used to derive optimal measurements in property testing. See,
e.g.,~\cite{goodman2009symmetry,landsberg2012tensors,fulton2013representation}
for introductions to these topics. Let $N$ be a positive integer and $\calH$ be a finite-dimensional complex Hilbert space. Also, given a permutation $\pi\in\SG_N$ let $W_{\pi}\in\rmU(\calH^{\otimes N})$ denote the permutation operator corresponding to $\pi$ with the action
\begin{align*}
    W_{\pi} \ket{x_1}\otimes\ket{x_2}\otimes\dots\otimes\ket{x_N} = \ket{x_{\pi^{-1}(1)}}\otimes \ket{x_{\pi^{-1}(2)}}\otimes\dots\otimes\ket{x_{\pi^{-1}(N)}}
\end{align*}
for any $\ket{x_1},\ket{x_2},\dots,\ket{x_N}\in \calH$. The tensor product Hilbert space $\calH^{\otimes N}$ admits a unitary representation of the group $\SG_N\times \rmU(\calH)$ via the action $\ket{v}\mapsto U^{\otimes N}W_{\pi}\ket{v}$ for any $\pi\in \SG_N$ and $U\in\rmU(\calH)$. Since the representation is unitary it has a canonical decomposition into a direct sum of orthogonal isotypic components,
$
    \calH^{\otimes N} = \bigoplus_{\lambda} V_{\lambda}
$. Schur-Weyl duality states that the direct sum in this equality is over partitions $\lambda\vdash N$ whose length $\ell(\lambda)$ is at most $\dim(\calH):=d$, and that the corresponding isotypic component is actually irreducible and isomorphic to $\calP_\lambda\otimes \calQ_\lambda(\calH)$, where $\calP_\lambda$ and $\calQ_\lambda$ are irreducible representations (irreps) of $\SG_N$ and $\rmU(\calH)$, respectively. We can summarize most of this information with the statement
\begin{align}\label{eq:sw_duality}
    \calH^{\otimes N}\cong \bigoplus_{\substack{\lambda\vdash N\\ \ell(\lambda)\leq d}} \calP_\lambda \otimes \calQ_\lambda(\calH).
\end{align}
We often identify $V_{\lambda}$ with $\calP_\lambda \otimes \calQ_\lambda(\calH)$ without comment, and we let $P_\lambda$ denote the projection onto the $\lambda^{\text{th}}$ subspace in the direct sum. Performing the projective measurement $\{P_\lambda\}_{\lambda}$ on a quantum state $\rho^{\otimes N}$ with $\rho\in \rmD(\calH)$ is called \textit{Weak Schur Sampling} (WSS) and is known to be optimal for testing properties invariant under unitary operations, i.e., properties of the spectrum of $\rho$~\cite{Montanaro2016}. Notably, this argument depends on the decomposition in \cref{eq:sw_duality} being multiplicity-free. We record some essential facts about WSS in \Cref{sec:wss}.

Now consider the case where $\calH=\calH_A\otimes \calH_B$ for a pair of finite-dimensional Hilbert spaces $\calH_A$ and $\calH_B$, and we are interested in a property $\calP\subseteq \rmP(\calH)$ which is invariant under local unitary operations on the $A$ and $B$ subsystems. In other words, $\calP$ is such that, if $\psi\in\calP$ then
\begin{align}\label{eq:statement_of_symmetry}
    (U_A \otimes U_B) \psi (U_A^\dag\otimes U_B^\dag)\in \calP \quad \textnormal{for all}\ \ U_A\in\rmU(\calH_A),\ U_B\in\rmU(\calH_B).
\end{align}
To take advantage of this symmetry, one may consider applying the decomposition arising from Schur-Weyl duality to each of the subsystems. Indeed, the space $\calH_A^{\otimes N}\otimes \calH_B^{\otimes N}$ admits a representation of the product group $\rmU(\calH_A)\times \rmU(\calH_B)$ via the action $\ket{v}\mapsto (U_A^{\otimes N}\otimes U_B^{\otimes N})\ket{v}$ for each $U_A\in\rmU(\calH_A)$ and $U_B\in\rmU(\calH_B)$. Setting $d_A:=\dim(\calH_A)$ and $d_B:=\dim(\calH_B)$ and using the decomposition arising from Schur-Weyl duality twice, we see
\begin{align}
    \calH_A^{\otimes N}\otimes \calH_B^{\otimes N} &\cong \bigoplus_{\substack{\lambda_1\vdash N\\ \ell(\lambda_1)\leq d_A}} \bigoplus_{\substack{\lambda_2\vdash N\\ \ell(\lambda_2)\leq d_B}} \calP_{\lambda_1} \otimes \calP_{\lambda_2}\otimes \calQ_{\lambda_1}(\calH_A)\otimes \calQ_{\lambda_2}(\calH_B)\nonumber\\
    &\cong \bigoplus_{\mu\vdash N}\bigoplus_{\substack{\lambda_1\vdash N\\ \ell(\lambda_1)\leq d_A}} \bigoplus_{\substack{\lambda_2\vdash N\\ \ell(\lambda_2)\leq d_B}} \mathbb{C}^{k_{\mu}(\lambda_1,\lambda_2)}\otimes \calP_{\mu}\otimes \calQ_{\lambda_1}(\calH_A)\otimes \calQ_{\lambda_2}(\calH_B)\label{eq:direct_sum_with_mu}
\end{align}
where in the second line we decomposed the product of symmetric group irreps, which results in multiplicities $k_{\mu}(\lambda_1,\lambda_2)$ known as the Kronecker coefficients. Evidently, this decomposition of the vector space into irreps of $\rmU(\calH_A)\times \rmU(\calH_B)$ is not multiplicity-free in general. However, a useful observation, perhaps first employed within the quantum information literature in \cite{hayashi2002universal}, simplifies this decomposition greatly. Namely, for measurements performed on identical copies of pure states, the representation of interest is just the subrepresentation within the symmetric subspace, i.e., the terms in the direct sum in \cref{eq:direct_sum_with_mu} where the partition $\mu=(N)$. Restricting the sum to these terms and applying Schur's Lemma (see the proof of \Cref{lem:double_sw_projectors} in \Cref{sec:optimal_measurement_proofs} for further details) enables one to conclude that
\begin{align}\label{eq:schur_weyl_bipartite_symmetric}
    \vee^N(\calH_A\otimes \calH_B) \cong \bigoplus_\lambda \calQ_{\lambda}(\calH_A)\otimes \calQ_{\lambda}(\calH_B).
\end{align}
where the isomorphism is as a representation of the group $\rmU(\calH_A)\times \rmU(\calH_B)$ and the direct sum ranges over $\lambda\vdash N$ such that $\ell(\lambda)\leq \min\{d_A,d_B\}$. Alternatively, \cref{eq:schur_weyl_bipartite_symmetric} follows immediately from the fact that, within the symmetric subspace, the representations of $\rmU(\calH_A)$ and $\rmU(\calH_B)$ form a dual reductive pair (see, e.g., \cite[Prop.~9.2.1]{goodman2009symmetry} and \cite[Sec.~5.4]{harrow2005applications}). An averaging argument similar to that in~\cite[Lemma 20]{Montanaro2016} then leads to the following theorem concerning the optimal measurements for testing properties related to bipartite entanglement of pure states; that is, having the symmetry in \cref{eq:statement_of_symmetry}.
\begin{theorem}[Implied by Theorem 3.2 in \cite{chen2024local}]\label{thm:purification}
    Let $\calH_A$, $\calH_B$ be finite-dimensional Hilbert spaces and $\calP \subset \rmP(\calH_A\otimes \calH_B)$ be a property of pure states which is invariant under local unitary transformations, as in \cref{eq:statement_of_symmetry}. Weak Schur Sampling performed locally, either on $\calH_A^{\otimes N}$ or $\calH_B^{\otimes N}$, followed by classical post-processing, is an optimal measurement for $\calP$.
\end{theorem}
The claim is slightly weaker than Theorem~3.2 in \cite{chen2024local}, in the sense of having a stricter requirement on the symmetry\footnote{\cite{chen2024local} only requires invariance under local unitary transformations applied to one of the two subsystems. However, this point is not relevant for the case of testing properties related to entanglement.}. The statement may also have been implicit in earlier works, such as the proof of Lemma~36 in \cite{soleimanifar2022testingmps}. We provide a self-contained proof of this theorem in \Cref{sec:optimal_measurement_proofs}.


\subsubsection{Equivalence between Schmidt-rank and rank testing}\label{sec:sr_rank_testing_equivalence}
Let $\calH_A$, $\calH_B$ be 
finite-dimensional Hilbert spaces and assume without loss of generality 
that $d_A\leq d_B$. For any positive integer $r\leq d_A$ and pure state $\ket{\psi}\in \calH_A\otimes \calH_B$ with a Schmidt decomposition
\begin{align}\label{eq:psi_schmidt_decomp}
    \ket{\psi} = \sum_{j=1}^{d_A}\lambda_j \ket{u_j}_A\otimes \ket{v_j}_B
\end{align}
such that
$\lambda_1\geq\lambda_2\geq\dots\geq \lambda_{d_A}$, let
$\Delta_r(\ket{\psi}):= \lambda_{r+1}^2+ \lambda_{r+2}^2+\dots+\lambda_{d_A}^2$
and $\mathsf{SR}(r)\subseteq \rmP(\calH_A\otimes \calH_B)$ be the 
property of pure states corresponding to having Schmidt-rank at most 
$r$. Consider testing this property with the distance measure set to the standard trace distance. Then $\psi$ is $\veps$-far from $\mathsf{SR}(r)$ if and only if
$\Delta_r(\ket{\psi})\geq \veps^2$. To see this, note that
\begin{align*}
    \min_{\varphi\in\mathsf{SR}(r)}\rmd_{\textnormal{Tr}}(\psi,\varphi)^2 &= 1 - \max_{\varphi\in\mathsf{SR}(r)} |\langle \psi|\varphi\rangle|^2\\
    &= \sum_{j=1}^{d_A}\lambda_j^2 - \sum_{j=1}^r \lambda_j^2\\
    &= \sum_{j\geq r+1} \lambda_j^2
\end{align*}
where the first line makes use of the fact that $\rmd_{\textnormal{Tr}}(\phi,\phi^\prime)=\sqrt{1-|\langle\phi|\phi^\prime\rangle|^2}$ for any two pure states $\phi=\outerprod{\phi}{\phi}$ and $\phi^\prime=\outerprod{\phi^\prime}{\phi^\prime}$ and in the second line we used the following version of the Eckart-Young-Mirsky Theorem on low-rank approximation~\cite{Eckart1936,mirsky1960symmetric}. This version is also stated explicitly as Lemma 18 in~\cite{soleimanifar2022testingmps}, and we provide an elementary proof in \Cref{sec:eckart_young_mirsky} for completeness.
\begin{proposition}\label{prop:eckart_young_mirsky}
    Let $\ket{\psi}\in\calH_A\otimes \calH_B$ be as defined in \cref{eq:psi_schmidt_decomp}. It holds that
    \begin{align*}
        \max\{|\langle\psi|\varphi\rangle|^2:\ket{\varphi}\in\calH_A\otimes\calH_B,\ \varphi\in\mathsf{SR}(r)\} &= \sum_{j=1}^r \lambda_j^2
    \end{align*}
    and the maximum is attained by $(\sum_{k=1}^r\lambda_k^2)^{-1/2}\sum_{j=1}^r\lambda_j \ket{u_j}\otimes \ket{v_j}$.
\end{proposition}
One may alternatively be interested in the property of mixed states
corresponding to having rank at most $r$, which we denote by $\mathsf{Rank}(r)\subseteq\rmD(\calH_A)$. Let $P^{(j)}_\lambda\in\rmL(\calH_j^{\otimes N})$ denote the orthogonal projection onto the $\lambda^\text{th}$ irrep in $\calH_j^{\otimes N}$ as in \cref{eq:sw_duality}, for each $j\in \{A,B\}$. Here and throughout we set
\begin{align}\label{eq:rank_test_op}
    \Pi^{(A)}_{\leq r} := \sum_{\lambda\vdash N\colon \ell(\lambda)\leq r} P_\lambda^{(A)},\quad   \Pi^{(B)}_{\leq r} := \sum_{\lambda\vdash N\colon \ell(\lambda)\leq r} P^{(B)}_\lambda
\end{align}
which are orthogonal projection operators. When the vector space being acted on is clear from context we drop the superscripts and write $\Pi_{\leq r}$ to denote the same operation. Ref.~\cite{o2015quantum} proved that $\Pi_{\leq r}$ is a strongly PC-optimal test for the property $\mathsf{Rank}(r)$. Moreover, by \cite[{Prop.~2.2}]{o2015quantum} it holds that $\rho$ is $\veps$-far from $\mathsf{Rank}(r)$ with respect to trace distance if and only if $\alpha_{r+1}+\alpha_{r+2}+\dots+\alpha_{d_A}\geq \veps$, where $\alpha_1\geq \alpha_2\geq\dots\geq \alpha_{d_A}$ is the spectrum of $\rho$. Combining these observations with \Cref{thm:purification} leads to the following two key facts, the second of which is also proven in~\cite{chen2024local}.
\begin{theorem}[PC-optimal Schmidt-rank test]\label{thm:PC_optimal_sr_test}
    Let $\calH_A,\calH_B$ be finite-dimensional Hilbert spaces such that $\dim(\calH_j)=d_j$ for each $j\in \{A,B\}$, $N$ and $r\leq d_A\leq d_B$ be positive integers, and $\PSym\in\rmL(\calH_A^{\otimes N}\otimes \calH_B^{\otimes N})$ denote the projector onto the symmetric subspace $\vee^N(\calH_A\otimes \calH_B)$.
    \begin{enumerate}
        \item It holds that
            \begin{align*}
                (\mathds{1}_{A}^{\otimes N}\otimes \Pi_{\leq r}^{(B)})\PSym &= (\Pi_{\leq r}^{(A)}\otimes \mathds{1}_{B}^{\otimes N})\PSym = \Proj \Span \{\ket{\phi}^{\otimes N}: \SR(\ket{\phi})\leq r\}
            \end{align*}
        is a strongly PC-optimal test for $\mathsf{SR}(r)\subseteq\rmP(\calH_A\otimes\calH_B)$.
        \item The copy complexity of testing $\mathsf{SR}(r)$ with one-sided error is $\Theta(r^2/\veps^2)$.
    \end{enumerate}
\end{theorem}
We defer the proof of this theorem to \Cref{sec:optimal_measurement_proofs}.

\subsection{Weak Schur Sampling toolkit}\label{sec:wss}
As explained in \Cref{sec:schur_weyl}, the optimal measurement for testing properties related to bipartite entanglement is given by a collection of local projection operators, and in particular, projections onto irreducible $(\mathfrak{S}_N\times \rmU(\mathbb{C}^d))$-representations. Performing this measurement induces a distribution over partitions $\lambda\vdash N$ which is often referred to as \textit{Weak Schur Sampling} (WSS). We provide a brief overview of the relevant facts below for convenience, following closely the presentation in Ref.~\cite{o2015quantum}.

Recall that given a partition $\lambda\vdash N$, the \textit{Young diagram} associated with $\lambda$ is the collection of equally-sized boxes which has $\lambda_i$ boxes in the $i^\text{th}$ row, as depicted in \Cref{fig:young_diagrams_a}. A \textit{standard Young tableau} of shape $\lambda$ over alphabet $[d]$ is a filling of these boxes such that the integers appearing in the boxes are strictly increasing from left-to-right and from top-to-bottom. If the integers are only weakly increasing from left-to-right, but still strictly increasing from top-to-bottom (as in \Cref{fig:young_diagrams_b}) the Young tableau is said to be \textit{semistandard}.
\begin{definition}[Schur polynomials]\label{def:schur_polynomial}
    Let $N$ and $d$ be positive integers and fix a collection of independent variables $\alpha_1,\dots,\alpha_d$ as well as a partition $\lambda\vdash N$. The \emph{Schur polynomial} $s_{\lambda}(\alpha_1,\dots,\alpha_d)$ is equal to
    $
        \sum_T \alpha^T
    $
    where the sum is over all semistandard Young tableaus of shape $\lambda$ over $[d]$ and, for each $T$ we have
    $\alpha^T := \prod_{j=1}^d\alpha_j^{k_j}$ with $k_j$ equal to the number of occurrences of letter $j$ in $T$.
\end{definition}
\begin{figure}
    \centering
    \begin{subfigure}[t]{0.45\textwidth}
        \centering
        \begin{ytableau}
            {}  &   \\
            {}  &   \\
            {}    
        \end{ytableau}
     \caption{\label{fig:young_diagrams_a}A Young diagram of shape ${\lambda=(2,2,1)\vdash 5}$.}
    \end{subfigure}
    \hfill
    \begin{subfigure}[t]{0.45\textwidth}
        \centering
        \begin{ytableau}
            1  &  1 \\
            2  & 3  \\
            4    
        \end{ytableau}
        \caption{\label{fig:young_diagrams_b}A semistandard Young tableau of shape $\lambda$ over $[4]$.}   
    \end{subfigure}
    \caption{\label{fig:young_diagrams}}
\end{figure}
With the above in hand, we may now state an important lemma which characterizes the WSS distribution in terms of computable quantities. We state this result without proof. A proof can be found, e.g., in Section~2.6 of~\cite{o2015quantum}.
\begin{lemma}\label{lem:wss_probs}
    Let $\rho\in\rmD(\mathbb{C}^d)$ be a quantum state with spectrum $\alpha:=(\alpha_1,\alpha_2,\dots,\alpha_d)$ and $P_\lambda\in\rmL((\mathbb{C}^d)^{\otimes N})$ be the projection operator onto the $\lambda^{\textnormal{th}}$ irrep in \cref{eq:sw_duality}. It holds that
    \begin{align*}
        \Tr(P_{\lambda}\rho^{\otimes N})=\dim(\calP_{\lambda})\cdot s_{\lambda}(\alpha_1,\alpha_2,\dots,\alpha_d).
    \end{align*}
\end{lemma}
By the Hook-Length formula~\cite{Frame1954TheHG},
\begin{align}\label{eq:dim_of_sn_irrep}
    \dim(\calP_{\lambda}) &= \frac{N!}{\ell_1!\dots\ell_t!}\prod_{i<j}(\ell_i-\ell_j),\qquad \ell_i:= \lambda_i+t-i\quad \forall i\in [t].
\end{align}
Together with \cref{lem:wss_probs} this yields an explicit, albeit difficult-to-work-with, expression for the probabilities of the various outcomes $\lambda\vdash N$ when performing WSS on $N$ copies of a quantum state $\rho$. However, using the Robinson-Schensted-Knuth (RSK) correspondence~\cite{knuth1970permutations} between Young tableaus and generalized permutations, it is possible to derive a more useful expression for the relevant probabilities under the WSS distribution. Given any sequence of $N$ letters $x\in [d]^N$ over the alphabet $[d]$ we let $\LDS(x)$ denote the longest (strictly) decreasing subsequence (LDS) of $x$. Then the result below follows immediately upon combining \Cref{lem:wss_probs} with Prop.~2.16 in \cite{o2015quantum}. (See also~\cite[Eq.~(2.2)]{its2001random}.)
\begin{lemma}\label{lem:lds_characterization}
    Let $\rho\in\rmD(\mathbb{C}^d)$ be a quantum state with spectrum $\alpha:=(\alpha_1,\alpha_2,\dots,\alpha_d)$ sorted in nonincreasing order and $\Pi_{\leq r}\in\rmL((\mathbb{C}^d)^{\otimes N})$ be defined as in \cref{eq:rank_test_op}. It holds that
    \begin{align}
        \Tr(\Pi_{\leq r}\rho^{\otimes N})&=\Pr_{\bm{x}\sim \alpha^{\otimes m}}\left[\LDS(\bm{x})\leq r \right].\label{eq:rsk_lds_pr}
    \end{align}
\end{lemma}

We make extensive use of this characterization in our proofs, both for upper and lower bounds. When the quantum state is completely mixed, the acceptance probability is straightforward to characterize combinatorially as in the following lemma due to~\cite{o2015quantum}.
\begin{lemma}[Prop.~2.31 in \cite{o2015quantum}]\label{lem:lds_conc}
    Let $N$, $d$, and $r\leq d$ be positive integers. It holds that
    \begin{align*}
        \Pr_{\bm{x}\sim \textnormal{Unif}([d]^N)}\left[\LDS(\bm{x})\geq r \right]&\leq\left(\frac{(1+r/d)\ee^2 N}{r^2}\right)^{r}.
    \end{align*}
\end{lemma}

\section{Test for TTNS}\label{sec:ttns}
Let $G=(V,E)$ be a tree graph on $|V|=n$ vertices and $|E|=n-1$ edges, $r$ be a positive integer, $(\calH_v:v\in V)$ be a collection of finite-dimensional Hilbert spaces, and $\mathsf{TTNS}(G,r)$ be defined as in \cref{eq:ttns_def}. Suppose one wishes to determine whether an unknown state is in this property or is $\veps$-far from it in trace distance, as described in \Cref{sec:property_testing_defs}. In this section we prove the following theorem, which essentially follows from combining the analysis in Section 4 of Ref.~\cite{soleimanifar2022testingmps} with Theorem~11.58 in Ref.~\cite{Hackbusch2019}.
\begin{theorem}\label{thm:ttns_testing}
    There is a test with one-sided error for the property $\mathsf{TTNS}(G,r)$ which is successful given $O(nr^2/\veps^2)$ copies of the unknown state.
\end{theorem}
To prove this result, we make use of the following theorem which bounds the error in approximating an arbitrary state in $\bigotimes_{v\in V} \calH_v$ by one contained in $\mathsf{TTNS}(G,r)$.
\begin{lemma}\label{thm:ttns_approximation}
   Let $\ket{\psi}\in \bigotimes_{v\in V} \calH_v$ be an arbitrary pure
   state on $n=|V|$ sites. For each $e\in E$, let $L^{(e)}, R^{(e)}\subseteq [n]$ be
   the bipartition of the vertices induced by removing the edge
   $e$. Suppose further that $\ket{\psi}$ has the Schmidt decomposition
   \begin{align*}
       \ket{\psi} = \sum_j {\lambda_j^{(e)}}\
       \ket{a^{(e)}_j}_{L^{(e)}}\otimes \ket{b^{(e)}_j}_{R^{(e)}}
   \end{align*}
   with Schmidt coefficients $\lambda_1^{(e)}\geq
   \lambda_2^{(e)}\geq \dots$ for every $e\in E$, and that
   \begin{align}\label{eq:truncated_schmidt_ineq}
       0\leq \sum_{e\in E}\sum_{j\geq r+1}(\lambda_j^{(e)})^2\leq 1.
   \end{align}
   Then there exists a state
   $\ket{\phi}\in\mathsf{TTNS}(G,r)$ for which
   \begin{align*}
       |\langle \phi | \psi \rangle|\geq 1-\sum_{e\in E}\sum_{j\geq r+1}(\lambda_j^{(e)})^2.
   \end{align*}
\end{lemma}
This lemma is implied by Theorem~11.58 in Ref.~\cite{Hackbusch2019}, which is in turn similar to the faithful MPS approximation result of Ref.~\cite{verstraeteMPSFaithfully}. We provide a proof in the language of this paper in \Cref{sec:ttns_approx_proof} for completeness. The TTNS tester then proceeds in a similar manner to that for MPS testing~\cite{soleimanifar2022testingmps}, via a reduction to Schmidt-rank testing, which is enabled by the following corollary. Here, ``$\veps$-far'' is with respect to trace distance between pure states, and ``with respect to edge $e$'' is shorthand for ``with respect to the bipartition determined by removing the edge $e$ from $E$''.
\begin{corollary}\label{cor:ttns_far}
   Let $G=(V,E)$ be a tree and $\psi$ be a pure state on $n=|V|$ sites that is $\veps$-far from $\mathsf{TTNS}(G,r)$ for some $\veps\in (0,1]$. There exists an edge $e\in E$ such that $\psi$ is $\veps/\sqrt{2(n-1)}$-far from the pure state property $\mathsf{SR}(r)$ comprising all states of Schmidt-rank at most $r$ with respect to this edge.
\end{corollary}
\begin{proof}
   For each $e\in E$ let $\lambda_1^{(e)},\lambda_2^{(e)},\dots$ be the Schmidt coefficients of $\ket{\psi}$, in non-increasing order, with respect to $e$ and set $\delta:=\sum_{e\in E}\sum_{j\geq r+1}(\lambda_j^{(e)})^2$. We may assume without loss of generality that $\delta\in [0,1]$, since if $\delta > 1$ we trivially have
   \begin{align}\label{eq:eq_32}
       \veps^2/2\leq \delta \leq (n-1)\max_{e\in E} \sum_{j\geq r+1}(\lambda_j^{(e)})^2
   \end{align}
   which proves the claim since, by \Cref{prop:eckart_young_mirsky}, $\sum_{j\geq r+1}(\lambda_j^{(e)})^2\geq \veps^2/2(n-1)$ if and only if the state $\psi$ is $\veps/\sqrt{2(n-1)}$-far from $\mathsf{SR}(r)$, defined with respect to $e$.  If $\psi$ is $\veps$-far from $\mathsf{TTNS}(G,r)$ then the greatest possible overlap between $\psi$ and a state $\phi\in\mathsf{TTNS}(G,r)$ is at most $\sqrt{1-\veps^2}$. By \Cref{thm:ttns_approximation} we then have
   \begin{align*}
       1-\veps^2/2\geq \sqrt{1-\veps^2}\geq |\langle \phi |\psi\rangle|\geq 1-\delta
   \end{align*}
   for some $\phi \in\mathsf{TTNS}(G,r)$. This also leads to \cref{eq:eq_32}.
\end{proof}

We now have the tools to prove~\Cref{thm:ttns_testing}.

\begin{proof}[Proof of \Cref{thm:ttns_testing}]
    Let $N$ be a positive integer to be specified shortly. Define the TTNS test $\Pi_{\mathsf{TTNS}}\in \rmL(\calH^{\otimes N})$ as
\begin{align}\label{eq:ttns_test_def}
    \Pi_{\mathsf{TTNS}}:=\Pi_{\leq r}^{(e_1)}\Pi_{\leq r}^{(e_2)}\dots\Pi^{(e_{n-1})}_{\leq r}
\end{align}
where $e_1,e_2,\dots,e_{n-1}\in E$ are the edges in the tree and $\Pi^{(e)}_{\leq r} := \Proj (\Span \{\ket{\phi}^{\otimes N}:\SR_e(\ket{\phi}\leq r)\})$. We claim the projections commute and therefore $\Pi_{\mathsf{TTNS}}$ is also a projection operator. To see this, note that the projections in the product each have an image which is contained in the symmetric subspace. Then consider a pair of distinct edges $e$ and $e^\prime$ such that, upon removing these edges we obtain bipartitions $L^{(e)},R^{(e)}\subseteq V$ and $L^{(e^\prime)},R^{(e^\prime)}\subseteq V$, respectively. Since $G$ is a tree, it is always possible to pick one subset from each of these bipartitions such that the resulting two subsets of vertices are disjoint\footnote{In more detail, for the edge $e$ we can pick the subset of vertices in the subgraph which does not contain $e^\prime$, and vice versa for $e^\prime$.}. Let us therefore assume without loss of generality that $L^{(e^\prime)}$ and $R^{(e)}$ are disjoint sets. By a judicious choice of the equivalent expressions for $\Pi_{\leq r}^{(e)}$ and $\Pi_{\leq r}^{(e^\prime)}$ established in the first item of \Cref{thm:PC_optimal_sr_test} we have
\begin{align*}
    \Pi^{(e)}_{\leq r} \Pi^{(e^\prime)}_{\leq r} &= (\Pi_{\leq r}\otimes \mathds{1}^{\otimes N}_{R^{(e)}})\PSym\Pi^{(e^\prime)}_{\leq r}\\
    &= (\Pi_{\leq r}\otimes \mathds{1}^{\otimes N}_{R^{(e)}})\Pi^{(e^\prime)}_{\leq r}\\
    &= (\Pi_{\leq r}\otimes \mathds{1}^{\otimes N}_{R^{(e)}})(\mathds{1}^{\otimes N}_{L^{(e^\prime)}}\otimes \Pi_{\leq r})\PSym\\
    &= (\mathds{1}^{\otimes N}_{L^{(e^\prime)}}\otimes \Pi_{\leq r})(\Pi_{\leq r}\otimes \mathds{1}^{\otimes N}_{R^{(e)}})\PSym\\
    &= (\mathds{1}^{\otimes N}_{L^{(e^\prime)}}\otimes \Pi_{\leq r})\PSym(\Pi_{\leq r}\otimes \mathds{1}^{\otimes N}_{R^{(e)}})\PSym\\
    &= \Pi^{(e^\prime)}_{\leq r} \Pi^{(e)}_{\leq r},
\end{align*}
as claimed. Now by \Cref{cor:ttns_far}, if the state $\psi$ being measured is $\veps$-far from $\mathsf{TTNS}(G,r)$ then there exists an edge with respect to which $\psi$ is $\veps/\sqrt{2(n-1)}$-far from the property $\mathsf{SR}(r)$. Since the projection operators in the right-hand side of \cref{eq:ttns_test_def} commute, we may without loss of generality assume that $e_{n-1}$ is this edge. The test has the operational interpretation of accepting if and only if each measurement in the sequence of measurements $\{\Pi^{(e_{n-1})}_{\leq r},\mathds{1}-\Pi^{(e_{n-1})}_{\leq r}\},\dots,\{\Pi^{(e_1)}_{\leq r},\mathds{1}-\Pi^{(e_1)}_{\leq r}\}$ accepts. This occurs with probability at most $\Tr(\Pi_{\leq r}^{(e_{n-1})}\psi^{\otimes N})$ which, by the second item in \Cref{thm:PC_optimal_sr_test}, is at most, say, $1/3$ for some $N=O(nr^2/\veps^2)$. On the other hand, we clearly have that $\Pi_{\mathsf{TTNS}}$ always accepts whenever $\psi\in\mathsf{TTNS}(G,r)$ for any choice of $N$.
\end{proof}

\section{Lower bound for testing TTNS}\label{sec:lower_bounds}
In this section we prove that for sufficiently high bond dimension, the copy complexity of testing $\mathsf{TTNS}(G,r)$ with one-sided error is at least $\Omega(n r^2 /\log n)$.

\subsection{Hard instance for testing TTNS}\label{sec:tree_like}
We first describe a family of states that are $\veps$-far from $\mathsf{TTNS}(G,r)$, which we will use as a hard instance for testing this property. This is the construction depicted schematically in \Cref{fig:tree_like_b}.
Let $d>0$ be a positive integer and $G=(V,E)$ be a tree graph with $n=|V|$ vertices and $(\calH_v)$ be a list of local Hilbert spaces such that for each $v\in V$ we have $\calH_v=(\mathbb{C}^d)^{\otimes D_v}$, where $D_v$ is the degree of $v$. By fixing a choice of the maximum bond dimension $r$ we get a property $\mathsf{TTNS}(G,r)\subset \rmP(\calH)$, with $\calH:=\bigotimes_{v\in V}\calH_v$. One example of a state in $\rmP(\calH)$ but \textit{not} in this property is the state which we denote $\ket{\phi_G}$, constructed in the following way. Let $\ket{\phi}\in\calH_A\otimes \calH_B$ be a bipartite state with Schmidt-rank at least $r+1$, and where $\calH_A=\calH_B=\mathbb{C}^d$. We may then define $\ket{\phi_G}\in\cal H$ as a tensor product of $n-1$ copies of $\ket{\phi}$ where, for each $e=(u,v)\in E$, we identify $\calH_A$ with one of the $d$-dimensional subsystems in $\calH_u$ and $\calH_B$ with one of the $d$-dimensional subsystems in $\calH_v$. Then $\phi_G=\outerprod{\phi_G}{\phi_G}\in \rmP(\calH)$ but $\phi_G\notin \mathsf{TTNS}(G,r)$.

How far away from the property is this state? The following lemma makes precise the intuition that the optimal $\mathsf{TTNS}(G,r)$ approximation to the state $\phi_G$ is given by simply using the optimal low Schmidt-rank approximation for each state in the product comprising $\ket{\phi_G}$. Hence, for a family of trees $G$ with $n$ vertices and a fixed choice of $\ket{\phi}$ the overlap of the optimal approximation to $\phi_G$ decays exponentially in $n$. This is similar to Prop.~3.3 in~\cite{soleimanifar2022testingmps}, which applies to matrix product states.
\begin{lemma}\label{lem:overlap_multiplicative_lemma}
    Let $\ket{\phi}$ and $\ket{\phi_G}$ be as defined above. If
    \begin{align}\label{eq:sr_opt_statement}
        \max_{\ket{\psi}\in \mathsf{SR}(r)} |\langle \psi | \phi\rangle|^2 = F
    \end{align}
    then
    \begin{align}\label{eq:ttns_maximization_with_far_state}
        \max_{\ket{\psi}\in \mathsf{TTNS}(G,r)} |\langle \psi | \phi_G\rangle|^2 = F^{n-1}.
    \end{align}
\end{lemma}
\begin{proof}
   The proof is by induction on $n$. The base case $n=2$ is trivial. Assume the claim holds for all connected trees on $n-1$ vertices and suppose $G$ has $n$ vertices. Consider a vertex $u$ which is the parent of a leaf vertex $v$. By definition, $\calH_u=(\mathbb{C}^d)^{\otimes D_u}$ where $D_u$ is the degree of $u$. Define the Hilbert space $\calH_A$ as the tensor product of $\calH_v=\mathbb{C}^d$ and the last copy of $\mathbb{C}^d$ in the product comprising $\calH_u$, and set $\calH_B$ to be the tensor product of the remaining copies of $\mathbb{C}^d$ in the product comprising $\calH$. Then $\calH=\calH_A\otimes \calH_B$ and any state $\ket{\psi}\in\mathsf{TTNS}(G,r)$ has a Schmidt decomposition with respect to this tensor product structure which may be written as
    \begin{align*}
        \ket{\psi} &= \sum_i\lambda_i \ket{a_i}_A\otimes \ket{b_i}_B.
    \end{align*}
    Furthermore, for each index $i$ it holds that $\ket{a_i}\in\mathsf{SR}(r)\subseteq\mathbb{C}^d\otimes \mathbb{C}^d$ while $\ket{b_i}\in \mathsf{TTNS}(G^\prime,r)\subset \rmP(\bigotimes_{v^\prime\in V^\prime}\calH_{v^\prime})$ where $G^\prime=(V^\prime, E^\prime)$ is the tree on $n-1$ vertices formed by removing the leaf $v$ from $G$ and the local Hilbert space $\calH_u$ has been redefined as $(\mathbb{C}^d)^{\otimes D_u-1}$. To see this, assume for contradiction that there is an edge $e\in E^\prime$ for which $\SR_e(\ket{b_i})\geq r+1$ for some $i$ and consider the partition of $V$ into disjoint subsets $S$ and $V\backslash S$ induced by removing $e$ from the original graph $G$. Also, without loss of generality assume $V\backslash S$ is the set which contains both $u$ and $v$. Then the reduced state of $\psi$ on the sites corresponding to the vertices in $S$ is
    \begin{align*}
        \Tr_{V \backslash S}(\psi) = \sum_i \lambda_i^2 \Tr_{V^\prime\backslash S}(\outerprod{b_i}{b_i})
    \end{align*}
    which has rank at least $r+1$, contradicting the membership of $\psi$ in $\mathsf{TTNS}(G,r)$. A similar reasoning shows $\ket{a_i}\in\mathsf{SR}(r)$ for each $i$. Next, we compute
    \begin{align}
        |\langle \phi_{G} | \psi \rangle|^2 &= \left\lvert\bra{\phi}\left(\sum_i \lambda_i \langle \phi_{G^\prime} | b_i\rangle \ket{a_i}\right)\right\rvert^2\nonumber\\
        &=: \left(\sum_i \lambda_i^2|\langle \phi_{G^\prime} | b_i\rangle|^2\right) | \langle\phi | \chi\rangle|^2\label{eq:eq_18}
    \end{align}
    where $\ket{\chi}$ is the unit vector proportional to $\sum_{i}\lambda_i\langle \phi_{G^\prime} | b_i\rangle\ket{a_i}=(\bra{\phi_{G^\prime}}\otimes\mathds{1})\ket{\psi}$. By the induction hypothesis, $|\langle \phi_{G^\prime} | b_i\rangle|^2$ is at most $F^{n-2}$ for each $i$. Furthermore, $\SR(\ket{\chi})\leq r$, which is straightforward to see since $\ket{\psi}$ has Schmidt-rank at most $r$ with respect to the edge between $u$ and $v$, and applying a local operator cannot increase the Schmidt-rank of a vector. Hence, $|\langle \phi | \chi\rangle|^2\leq F$ and the right-hand side of \cref{eq:eq_18} is at most $F^{n-1}$, using the fact that $\sum_{i}\lambda_i^2=1$.
    
    That $F^{n-1}$ is also a lower bound on the left-hand side of \cref{eq:ttns_maximization_with_far_state} follows immediately by taking the tensor product of $n-1$ copies of the optimal solution to the optimization in the left-hand side of \cref{eq:sr_opt_statement}.
\end{proof}
Using this result, we describe a family of states which are $\veps$-far from $\mathsf{TTNS}(G,r)$ for some $\veps\in (0,\frac{1}{\sqrt{6}}]$, but require many copies to reject with constant probability. Let $r \geq 2$, $m=n-1$ be the number of edges in the tree, and $d\in\mathbb{Z}_+$ be such that $d\geq 2r-1$. Also, let $\ket{\phi}\in\mathbb{C}^d\otimes \mathbb{C}^d$ be a state of Schmidt-rank $d$ defined as
\begin{align}\label{eq:hard_phi}
    \ket{\phi} = \sqrt{1-\frac{4\veps^2}{m}}\ket{1}\otimes \ket{1} + \sqrt{\frac{4\veps^2}{m(d-1)}}\sum_{j=2}^d \ket{j}\otimes\ket{j}.
\end{align}
Then, by \Cref{prop:eckart_young_mirsky} and our assumption $\veps \in (0,\frac{1}{\sqrt{6}})$,
\begin{align*}
    \max_{\ket{\psi}\in\mathsf{SR}(r)} |\langle \psi | \phi\rangle|^2 &= 1-\frac{4\veps^2}{m} + \left(\frac{r-1}{d-1}\right)\frac{4\veps^2}{m}\leq 1-\frac{2\veps^2}{m}.
\end{align*}
Thus, by \Cref{lem:overlap_multiplicative_lemma}, it holds that
\begin{align}
    \min_{\psi\in\mathsf{TTNS}(G,r)}\td(\phi_G, \psi) &= \sqrt{1-\max_{\psi\in \mathsf{TTNS}(G,r)}|\langle \psi | \phi_G\rangle|^2}\nonumber\\
    &\geq \sqrt{1-\left(1-\frac{2\veps^2}{m}\right)^m}\nonumber\\
    &\geq \veps\label{eq:ttns_eps_far_ineq}
\end{align}
where in the final line we used the inequality $(1-x)^q\leq 1-qx/2$ for all $x\leq 1/q$, which follows from the fact that the first order Taylor polynomial for $(1-x)^q$ has Lagrange remainder bounded by $q(q-1)x^2/2$. In other words, with the above choice of Schmidt coefficients for $\ket{\phi}$, the resulting pure state $\phi_G$ is $\veps$-far from the property $\mathsf{TTNS}(G,r)$.

\subsection{Proof of the lower bound}
In order to show that the state $\phi_G$ requires many copies to reject for any test with one-sided error, we relate the performance of a specific test which is easier to analyze to that of the strongly PC-optimal test. For every Hilbert space $\calH_v$, $v\in V$ we assign each of the $D_v$ copies of $\mathbb{C}^d$ in the tensor product comprising $\calH_v$ to a unique edge incident on $v$. In this way, we can write $\calH=\bigotimes_{e\in E}\calH_e$ with $\calH_e = \mathbb{C}^d\otimes \mathbb{C}^d$ for every $e\in E$. Then, for each $e\in E$ let $\Pi^{(e)}_{\leq r}\in\rmL(\calH_e^{\otimes N})$ be the $N$-copy strongly PC-optimal test for the property $\mathsf{SR}(r)\subset \rmP(\calH_e)$ from \Cref{thm:PC_optimal_sr_test}. Also, let $\PSym$ be the projector onto the symmetric subspace of $\calH^{\otimes N}$.
\begin{lemma}\label{lem:pc_optimal_test_lb_operator}
    Let $\Pacc\in\rmL(\calH^{\otimes N})$ be an $N$-copy strongly PC-optimal test for the property $\mathsf{TTNS}(G,r)\subset \rmP(\calH)$. It holds that
    $
        \bigotimes_{e\in E}\Pi_{\leq r}^{(e)} \preceq \PSym\Pacc\PSym.
    $
\end{lemma}
\begin{proof}
    By \Cref{lem:pc_opt_pure_proj_lemma}, we can take $\Pacc=\Proj\Span\{\ket{\phi}^{\otimes N}:\outerprod{\phi}{\phi}\in\mathsf{TTNS}(G,r)\}$ without loss of generality. Now consider the property $\mathsf{Prod}_2(G,r)$ which is defined to be the set of those states in $\mathsf{TTNS}(G,r)$ which are tensor products of bipartite states ``along the edges'' of $G$, i.e.,
    \begin{align*}
         \mathsf{Prod}_2(G,r)= \left\{\outerprod{\psi}{\psi}:\ket{\psi}=\bigotimes_{e\in E}\ket{\phi_e},\ \ket{\phi_e}\in\calH_e\ \textnormal{and}\ \SR(\ket{\phi_e})\leq r\ \forall  e\in E\right\}.
    \end{align*}
    Evidently,
    \begin{align}\label{eq:image_contains_products}
        \im(\Pacc)&\supseteq \Span\left\{\ket{\psi}^{\otimes N}:\outerprod{\psi}{\psi}\in\mathsf{Prod}_2(G,r)\right\}
    \end{align}
    so it suffices to demonstrate that the right-hand side is precisely the image of the projector $\bigotimes_{e\in E}\Pi_{{\leq r}}^{(e)}$. By definition, in the set on the right-hand side over which the span is being taken the elements $\ket{\psi}^{\otimes N}$ can each be expressed as $\bigotimes_{e\in E}\ket{\phi_e}^{\otimes N}$ for some $\ket{\phi_e}\in\calH_e$ satisfying $\SR(\ket{\phi_e})\leq r$, up to an irrelevant convention for the order in which the tensor product is written. Hence, the right-hand side of \cref{eq:image_contains_products} is equal to
    \begin{align*}
        &\Span \left\{\bigotimes_{e\in E}\ket{\phi_e}^{\otimes N}: \ket{\phi_e}\in\calH_e\ \textnormal{and}\ \SR(\ket{\phi_e})\leq r\ \forall e\in E\right\} \\
        &\qquad\qquad\qquad\qquad\qquad = \bigotimes_{e\in E}\Span\left\{\ket{\phi_e}^{\otimes N}: \ket{\phi_e}\in\calH_e\ \textnormal{and}\ \SR(\ket{\phi_e})\leq r\right\}\\
        &\qquad\qquad\qquad\qquad\qquad = \bigotimes_{e\in E}\im(\Pi^{(e)}_{\leq r})\\
        &\qquad\qquad\qquad\qquad\qquad = \im\left(\bigotimes_{e\in E}\Pi_{\leq r}^{(e)}\right),
    \end{align*}
    which concludes the proof.
\end{proof}
By \Cref{lem:pc_optimal_test_lb_operator} along with the definition of strong PC-optimality from \Cref{sec:property_testing_defs} we have that, for any test $\Pacc$ with one-sided error and pure state $\chi\in\rmP(\calH)$ which is $\veps$-far from $\mathsf{TTNS}(G,r)$, it holds that $\Tr(\Pacc \chi^{\otimes N})\geq \Tr\left((\bigotimes_{e\in E}\Pi^{(e)}_{\leq r}) \chi^{\otimes N}\right)$. Taking $\chi = \phi_G$ as in \Cref{sec:tree_like} and noting that $\phi_G = \bigotimes_{e\in E}\phi_e$ where $\ket{\phi_e}\in\calH_e=\mathbb{C}^d\otimes \mathbb{C}^d$ is as in the right-hand side of \cref{eq:hard_phi} for each $e\in E$, we have
\begin{align}\label{eq:accept_pr_is_prod}
   \Tr(\Pacc \phi_G^{\otimes N})&\geq \prod_{e\in E}\Tr\left(\Pi_{\leq r}^{(e)}\phi_e^{\otimes N}\right) = \Tr\left(\Pi_{\leq r} \rho^{\otimes N}\right)^m
\end{align}
where $m=n-1$ is the number of edges, $\rho\in\rmD(\mathbb{C}^d)$ is the reduced density matrix of any $\ket{\phi_e}\in\mathbb{C}^d\otimes \mathbb{C}^d$ on the first subsystem, i.e.,
$\rho = \diag(1-4\veps^2/m,\frac{4\veps^2}{m(d-1)},\dots,\frac{4\veps^2}{m(d-1)})$, and $\Pi_{\leq r}\in\rmL((\mathbb{C}^d)^{\otimes N})$ is the optimal rank test defined in \cref{eq:rank_test_op}. We can now state the theorem which leads to the lower bound in \Cref{thm:main_thm}.
\begin{theorem}\label{thm:main_lb_thm}
    Let $G$ be a tree graph on $n\geq 2$ vertices, $r\geq \max\{50, 1+\alpha(n)\}$ be an integer where $\alpha(n):=\log(4(n-1))$, and $\veps\in (0,1/\sqrt{6}]$. For any $N$-copy test $\Pacc\in\rmL(\calH^{\otimes N})$ for the property $\mathsf{TTNS}(G,r)\subset \rmP(\calH)$ with one-sided error, it holds that
    $
        \Tr(\Pacc \phi_G^{\otimes N})\geq \frac{1}{2}
    $
    so long as
    \begin{align}\label{eq:m_constraint}
        N \leq \frac{(n-1) r^2}{400 \alpha(n) \veps^2}.
    \end{align}
\end{theorem}
The lower bound in \Cref{thm:main_thm} follows from this theorem together with the fact established above through \cref{eq:ttns_eps_far_ineq} that $\phi_G$ is $\veps$-far from $\mathsf{TTNS}(G,r)$. Note that the restriction $r\geq 50$ is irrelevant for the statement of the copy complexity lower bound in terms of its limiting behaviour since, for any sequence of tuples of the relevant parameters either i) $r$ is upper bounded by $50$ or ii) $r$ tends to a value greater than $50$. In the first case $n$ is also upper bounded by some constant due to the constraint $r\geq 1+\alpha(n)$, and the trivial\footnote{It takes $\Omega(1/\veps^2)$ copies just to discriminate two states which are $\veps$-far from each other in trace distance.} lower bound $\Omega(1/\veps^2)$ still applies. In the second case, the hypothesis of \Cref{thm:main_lb_thm} is satisfied for all elements in the sequence whose index is greater than some positive integer.

We remind the reader that our construction of the pure states $\phi_G$ required us to identify $\calH_v$ with $(\mathbb{C}^d)^{\otimes D_v}$, where $d \geq 2r-1$ and $D_v$ is the degree of the vertex $v \in V$. So the statement of~\Cref{thm:main_lb_thm} implicitly requires $\dim(\calH_v)\geq (2r-1)^{D_v}$. We remark that similar implicit constraints are required for the lower bounds proven in~\cite{soleimanifar2022testingmps} and \cite{chen2024local} for MPS.
\begin{proof}[Proof of \Cref{thm:main_lb_thm}]
    Throughout, we let $m=n-1$ be the number of edges in $G$ and $\alpha=\alpha(n)$, suppressing the dependence on $n$. The proof is based on a hard instance of rank testing where the spectrum of the state is concentrated on one letter, similar to the proof of Lemma 6.2 in \cite{o2015quantum}. By \cref{eq:accept_pr_is_prod} it suffices to establish the bound $\Tr(\Pi_{\leq r}\rho^{\otimes N}) \geq 2^{-1/m}$ for $N$ satisfying the bound in \cref{eq:m_constraint}. By \Cref{lem:lds_characterization} we have
    \begin{align}\label{eq:lds_eq_in_lb_proof}
        \Tr(\Pi_{\leq r} \rho^{\otimes N}) &= \Pr_{\bm{x}\sim p^{\otimes N}}\left[\LDS(\bm{x})\leq r\right]
    \end{align}
    where $\bm{x}\in [d]^N$ is a random string over the alphabet $[d]$ and $p:=(1-4\veps^2/m,\frac{4\veps^2}{m(d-1)},\dots,\frac{4\veps^2}{m(d-1)})$ is the distribution over $[d]$ equal to the sorted spectrum of $\rho$. Consider a modified random variable $\widetilde{\bm{x}}$ obtained by removing all occurrences of the letter ``1" from $\bm{x}\sim p^{\otimes N}$. Since this process results in a random string which has an LDS that is either equal to $\LDS(\bm{x})$ or $\LDS(\bm{x})-1$, the right-hand side of \cref{eq:lds_eq_in_lb_proof} is at least
    \begin{align}
        \Pr_{\bm{x}\sim p^{\otimes N}}\left[\LDS(\widetilde{\bm{x}})\leq r-1\right] &\geq \sum_{L=0}^{\lfloor \mu(1+t)\rfloor} \Pr_{\bm{y}\sim \textnormal{Unif}([d-1]^L)}\left[\LDS(\bm{y})\leq r-1\right]\cdot \Pr_{\bm{x}\sim p^{\otimes N}}\left[|\widetilde{\bm{x}}|=L\right]\nonumber\\
        &\geq \left(\min_{L\in\mathbb{Z}_+: L\leq \mu(1+t)} \Pr_{\bm{y}\sim\textnormal{Unif}([d-1]^L)}\left[\LDS(\bm{y})\leq r-1\right]\right)\Pr_{\bm{x}\sim p^{\otimes N}}\left[|\widetilde{\bm{x}}|\leq \mu(1+t)\right] \label{eq:prod_of_2_probs}
    \end{align}
    where $\mu, t > 0$ are two real-valued parameters to be specified shortly and in the first line we have made use of the fact that, conditioned on the event $|\widetilde{\bm{x}}|=L$ the random variable $\widetilde{\bm{x}}$ is uniform on $\{2,3,\dots,d\}^L$. We will now specify parameters $\mu$ and $t$ which allow us to bound the right-hand side of \cref{eq:prod_of_2_probs} from below. Set $\mu=\expct_{\bm{x}\sim p^{\otimes N}} |\widetilde{\bm{x}}| = 4\veps^2N/m$ where the second equality is due to the fact that the random variable $|\widetilde{\bm{x}}|$ follows a Binomial distribution for $N$ trials with success probability $4\veps^2/m$. By a Chernoff bound (see \Cref{lem:chernoff_bound}) it holds that
    \begin{align*}
        \Pr_{\bm{x}\sim p^{\otimes N}}\left[|\widetilde{\bm{x}}|> \mu(1+t)\right]\leq \ee^{\frac{-\mu t^2}{2 + t}}
    \end{align*}
    for any $t\geq 0$. With some foresight, set
    $
        t = \frac{1}{2\mu}\left(\alpha + \sqrt{\alpha^2 + 8\alpha\mu}\right).
    $
    Then $\mu t^2 = (2+t)\alpha$ so that, by Bernoulli's inequality,
    \begin{align*}
       \left(1-\ee^{\frac{-\mu t^2}{2 + t}}\right)^{2m}&\geq 1-2m\ee^{\frac{-\mu t^2}{2 + t}} = \frac{1}{2},
    \end{align*}
    which implies that
    \begin{align}\label{eq:eq_23}
         \Pr_{\bm{x}\sim p^{\otimes N}}\left[|\widetilde{\bm{x}}|\leq \mu(1+t)\right]\geq 2^{-1/(2m)}.
    \end{align}
    Next, we prove a lower bound of $2^{-1/(2m)}$ for the term in parentheses in \cref{eq:prod_of_2_probs}. Since the probability in this expression is monotonically non-increasing in $L$, it suffices to establish a lower bound on this quantity when $L:=\lfloor \mu(1+t)\rfloor$. By \Cref{lem:lds_conc}, it holds that
    \begin{align*}
        \Pr_{\bm{y}\sim \textnormal{Unif}([d-1]^L)}\left[\LDS(\bm{y})\leq r-1\right]&\geq 1- \left(\frac{(1+\frac{r-1}{d-1})\ee^2 L}{(r-1)^2}\right)^{r-1}\geq 1-\left(\frac{3\ee^2 L}{2(r-1)^2}\right)^{r-1}
    \end{align*}
    where the second inequality follows from $d-1\geq 2(r-1)$. Assume for now that $3\ee^2L/(2(r-1)^2)\leq 1$. We will validate this helpful assumption momentarily. Then by Bernoulli's inequality once again, we have
    \begin{align}\label{eq:eq_109}
        \left(1-\left(\frac{3\ee^2 L}{2(r-1)^2}\right)^{r-1}\right)^{2m}&\geq 1-2m\left(\frac{3\ee^2 L}{2(r-1)^2}\right)^{r-1}
    \end{align}
    so in order to establish a lower bound of $2^{-1/2m}$ on the term in parentheses in \cref{eq:prod_of_2_probs} it suffices to show that that the right-hand side of \cref{eq:eq_109} is at least $1/2$, or equivalently
    \begin{align}\label{eq:rank_ratio_suffices}
        \left(\frac{2(r-1)^2}{3\ee^2 L}\right)^{r-1}\geq 4m.
    \end{align}
    The assumption that $r\geq 1 + \alpha$ is equivalent to $\log(4m)\leq r-1$. Hence, one sees by taking the logarithm of both sides that \cref{eq:rank_ratio_suffices} is implied by the inequality
    \begin{align}\label{eq:eq_66}
       \frac{2(r-1)^2}{3\ee^2L}\geq \ee,
    \end{align}
    which, if true, also validates the aforementioned helpful assumption. Using the definitions of $L$, $\mu$, and $t$ above we find \cref{eq:eq_66} is implied by
    \begin{align}\label{eq:eq_29}
        \frac{4\veps^2N}{m} + \frac{\alpha}{2} + \frac{1}{2}\sqrt{\alpha^2+\frac{32\alpha\veps^2 N}{m}} \leq  \frac{2}{3\ee^3}(r-1)^2.
    \end{align}
   Now since $\log(4)\leq \alpha \leq r-1$ and we have assumed $N \leq mr^2/(400 \alpha\veps^2)$, the left-hand side of \cref{eq:eq_29} is at most
    \begin{align*}
        \frac{r^2}{100 \log(4)} + \frac{r-1}{2} + \frac{1}{2}\sqrt{(r-1)^2 + \frac{2}{25}r^2} &\leq  \frac{r^2}{100} + \left(\frac{1}{2}+\frac{3\sqrt{3}}{10}\right)r -\frac{1}{2}
    \end{align*}
    which is indeed at most $2(r-1)^2/(3\ee^3)$ for every $r\geq 50$.
    This proves \cref{eq:eq_29} and therefore the term in parentheses in \cref{eq:prod_of_2_probs} is at least $2^{-1/(2m)}$. Combining with \cref{eq:eq_23} we have
    \begin{align*}
        \Pr_{\bm{x}\sim p^{\otimes N}}\left[\LDS(\widetilde{\bm{x}})\leq r-1\right]\geq 2^{-1/m}
    \end{align*}
    which proves the claim.
\end{proof}

\section{Testing products of bipartite states with constant Schmidt-rank}
\label{sec:improved_ub}
Our lower bound in \Cref{sec:lower_bounds}, as well as the $\Omega(\sqrt{n})$ lower bounds from Refs.~\cite{soleimanifar2022testingmps} and \cite{chen2024local} rely on analyzing states in the property $\mathsf{Prod}_2(n,r)$, which forms a subset of states with the TTNS property. Recall that $\mathsf{Prod}_2(n,r)$ is defined to be the set of all states which are $n/2$-wise products of bipartite states on the subsystems $(1,2), (3,4),\dots, (n-1,n)$ each with Schmidt-rank at most $r$. In this section, we prove \Cref{thm:improved_ub_qutrits}, which shows that, unlike the case where $r\geq 2+\log n$, a number of copies on the order of $\sqrt{n}$ suffices to test this class of states with one-sided error when $r=2$. As mentioned in \Cref{sec:results}, this improves upon the copy complexity obtained using a ``test-by-learning" approach. Here, a number of copies on the order of $3^n$ would be required to learn a full description of the state, assuming the subsystems are qutrits. Even provided the guarantee that the input state is product with respect to the $n/2$ subsystems, learning a description of it to within small trace distance error $\veps$ would seem to require linear in $n$ copies. This is because learning each bipartite state to within $\veps$ in trace distance uses $O(1/\veps^2)$ copies, but results in a worst-case error of $\sqrt{1-(1-\veps^2)^{n/2}}$ in trace distance for the full state, which is close to one unless $\veps=O(1/\sqrt{n})$.

\subsection{Decreasing subsequences on three letters}
The proof of \Cref{thm:improved_ub_qutrits} relies on the following lemma.
\begin{lemma}\label{lem:3_letter_lds_lem}
    Let $\calA=\{a,b,c\}$ be an alphabet on three letters, $N\geq 3$ be a positive integer, $t\in [0,1/3]$, and $\bm{w}\in \calA^N$ be a random $N$-letter word composed of i.i.d.\ copies of the random letter whose weights are given by
    \begin{align}\label{eq:lds_constraint}
        p(a)=1-x-y,\quad p(b)=x,\quad p(c)=y,\quad t \leq y\leq x\leq \frac{1-y}{2}
    \end{align}
    for some choice of $x,y \in \mathbb{R}$ satisfying the constraints. There exist universal constants $c_1\in (0,1/2]$ and $c_2>0$ such that if $Nt\leq c_1$ then
    \begin{align}\label{eq:lds_prob}
        \Pr\left[cba\ \textnormal{is a subsequence in}\ \bm{w}\right]\geq \frac{N^2t^2}{4} - c_2 t.
    \end{align}
\end{lemma}
\begin{proof}
    We prove the statement when $N$ is even for notational convenience; the case when $N$ is odd follows easily. We identify the letters $1$, $2$, and $3$ with the letters $a$, $b$, and $c$, respectively. First, let us restrict ourselves to the case $x=y=t$ without loss of generality, for this choice minimizes the probability by \Cref{lem:lds_helper} below. We have
    \begin{align}\label{eq:eq_50}
        \Pr\left[321\ \textnormal{is a subsequence in}\ \bm{w}\right] = 1-\Pr\left[\LDS(\bm{w})\leq 2\right]
    \end{align}
    and, when $x=y=t$,
    \begin{align}\label{eq:eq_51}
        \Pr\left[\LDS(\bm{w})\leq 2\right] &= \sum_{j=0}^{N/2}d_j s_j(t)
    \end{align}
    where $d_j$ is the dimension of the irrep of $\SG_N$ corresponding to the partition $(N-j,j)$, and $s_j(t)$ is short for the Schur polynomial $s_{(N-j,j)}(1-2t,t,t)$. Recalling the combinatorial definition of these polynomials (\Cref{def:schur_polynomial}), each term is itself a sum of polynomials indexed by semistandard Young tableaus (SSYTs) which we can write as follows. If we are considering an SSYT for which $i$ boxes are filled with the letter ``1" then the relevant polynomial is $(1-2t)^i t^{N-i}$. For a fixed shape $(N-j,j)$ and assuming $i\leq j$, the number of such SSYTs can be counted, and is equal to $(i+1)(N-2j+1)$. Similarly, if $i>j$ then this there are $(j+1)(N-i-j+1)$ such SSYTs. Hence,
    \begin{align*}
        s_j(t) &= \sum_{T} (1-2t)^{\alpha_1(T)} t^{N-\alpha_1(T)}\\
        &= \sum_{i=0}^j (i+1)(N-2j+1)(1-2t)^i t^{N-i} + \sum_{i=j+1}^{N-j}(j+1)(N-i-j+1)(1-2t)^i t^{N-i}\\
        &\leq (j+1)(1-2t)^{N-j}t^j + \sum_{i=0}^{N-j-1} (N-i+1)^2 t^{N-i}\\
        &= (j+1)(1-2t)^{N-j}t^j+\sum_{k=j+1}^{N} (k+1)^2 t^k.
    \end{align*}
    where in the first line the sum ranges over all semistandard Young tableaus of shape $(N-j,j)$ and $\alpha_1(T)$ denotes the number of occurrences of the letter ``1" in $T$. Substituting into \cref{eq:eq_51} we have the bound
    \begin{align}\label{eq:eq_56}
        \Pr\left[\LDS(\bm{w})\leq 2\right]\leq \sum_{j=0}^{N/2}\binom{N}{j} (j+1)(1-2t)^{N-j}t^j + \sum_{j=0}^{N/2}\binom{N}{j}\sum_{k=j+1}^{N} (k+1)^2 t^k
    \end{align}
    where we have also made use of the bound $d_j\leq \binom{N}{j}$, which is a straightforward consequence of \cref{eq:dim_of_sn_irrep}. The second of these sums is at most
    \begin{align}
        \sum_{j=0}^{N/2} N^j \sum_{k=j+1}^{N}(k+1)^2t^k &\leq t\sum_{k=1}^N (k+1)^2 t^{k-1}\sum_{j=0}^{k-1}N^j\nonumber\\
        &\leq 2t\sum_{k=1}^{N} (k+1)^2 (Nt)^{k-1}\nonumber\\
        &\leq 2t\cdot \frac{(Nt)^2-3Nt+4}{(1-Nt)^3}\nonumber\\
        &\leq 44\cdot t\label{eq:bound_linear_in_t}
    \end{align}
 where in the first line we extended the sum over $j$ to go from zero to $N-1$, in the second line we evaluated the sum over $j$ and bounded it by $2N^{k-1}$, in the third line we changed the limits of the sum to go from one to infinity and used \Cref{fact:sum_1}, and the final line follows from the assumption that $Nt\leq \frac{1}{2}$. Now, the first sum in the right-hand side of \cref{eq:eq_56} is equal to
    \begin{align*}
        \underbrace{(1-2t)^N\vphantom{\sum_{j=2}^{N/2}\binom{N}{j}}}_{\encircle{A}} + \underbrace{2Nt(1-2t)^{N-1}\vphantom{\sum_{j=2}^{N/2}\binom{N}{j}}}_{\encircle{B}} + \underbrace{\frac{3N(N-1)}{2}(1-2t)^{N-2}t^2\vphantom{\sum_{j=2}^{N/2}\binom{N}{j}}}_{\encircle{C}} + \underbrace{\sum_{j=3}^{N/2}\binom{N}{j} (j+1)(1-2t)^{N-j}t^j}_{\encircle{D}}
    \end{align*}
    where
    \begin{align}
        \encircle{A} &\leq 1-2Nt +2N(N-1)t^2 + C_0 \cdot (Nt)^3,\label{eq:eq_a}\\
        \encircle{B} &\leq 2Nt-4N(N-1)t^2 + C_0^\prime \cdot (Nt)^3,\ \textnormal{and}\label{eq:eq_b}\\
        \encircle{C} &\leq \frac{3N(N-1)t^2}{2} + C_0^{\prime\prime} \cdot (Nt)^3\label{eq:eq_c}
    \end{align}
    for some universal constants $C_0,C_0^\prime,C_0^{\prime\prime} > 0$, using \Cref{fact:sum_3}. We also have
    \begin{align}
        \encircle{D}&\leq \sum_{j=3}^{\infty} (j+1)(Nt)^j
        = \frac{(Nt)^3(4-3Nt)}{(1-Nt)^2}
        \leq 10\cdot N^3t^3\label{eq:eq_d}
    \end{align}
    where the equality uses \Cref{fact:sum_2} and the last line follows from the assumption that $Nt\leq 1/2$. Substituting the bounds in \cref{eq:eq_a,eq:eq_b,eq:eq_c} as well as \cref{eq:eq_d,eq:bound_linear_in_t} into \cref{eq:eq_56}, we arrive at the bound
    \begin{align*}
        \Pr\left[\LDS(\bm{w})\leq 2\right]&\leq 1-\frac{N(N-1)t^2}{2} + (\underbrace{C_0+C_0^\prime+C_0^{\prime\prime} + 10}_{:=\kappa})(Nt)^3 + 44\cdot t\\
        &\leq 1 - \frac{N^2t^2}{4} + c_2t
    \end{align*}
    where $c_2:=44$ and the inequality follows upon making an appropriate choice for the universal constant $c_1\in (0,1/2]$ such that $-(Nt)^2/3 +\kappa (Nt)^3\leq -(Nt)^2/4$ for all $0 \leq Nt\leq c_1$, noting that $N(N-1)/2\geq N^2/3$ for every positive integer $N\geq 3$.
\end{proof}
In the proof of \Cref{lem:3_letter_lds_lem} we used the following intuitive result which is a consequence of a combinatorial fact about LDS proven in~\cite{ODonnell2016}. Here, for any two $d$-dimensional real-valued vectors $x$, $y$ we write $x\succ y$ to mean $x$ \textit{majorizes} $y$, i.e., $\sum_{i=1}^kx_i\geq \sum_{i=1}^k y_i$ for all $k\in [d]$.
\begin{lemma}\label{lem:lds_helper}
    Let $N$ be a positive integer and let $p,q$ be two distributions over the letters $\{a,b,c\}$ such that $p\succ q$. It holds that
    \begin{align*}
        \Pr_{\bm{w}\sim p^{\otimes N}}[cba\ \textnormal{is a subsequence in}\ \bm{w}]\leq \Pr_{\bm{w}\sim q^{\otimes N}}[cba\ \textnormal{is a subsequence in}\ \bm{w}].
    \end{align*}
\end{lemma}
\begin{proof}
    For any distribution $\beta\in\mathbb{R}^{\{a,b,c\}}$ let $f(\beta):=1-\Pr_{\bm{w}\sim \beta^{\otimes N}}[cba\ \textnormal{is a subsequence in}\ \bm{w}]$. Also, let $\bm{\lambda}_\beta\vdash N$ denote the random partition of $N$ such that
    \begin{align*}
        \Pr[\bm{\lambda}_\beta = \lambda] &= \dim(\calP_{\lambda})s_{\lambda}\big(\beta(a),\beta(b),\beta(c)\big) \quad \textnormal{for every}\ \lambda\vdash N. 
    \end{align*}
    By Theorem~1.11 in~\cite{ODonnell2016} there is a coupling $(\bm{\lambda}_p, \bm{\lambda}_q)$ such that $\bm{\lambda}_p\succ \bm{\lambda}_q$ with certainty. 
    Furthermore, for any $\lambda_1,\lambda_2\vdash N$, $\lambda_1\succ \lambda_2$ implies $\ell(\lambda_1)\leq \ell(\lambda_2)$. Following the discussion in \Cref{sec:wss} we have
    \begin{align*}
        f(p) &= \sum_{\lambda: \ell(\lambda)\leq 2}\dim(\calP_{\lambda})s_\lambda\big(p(a),p(b),p(c)\big)\\
        &= \Pr[\ell(\bm{\lambda}_p)\leq 2]\\
        &\geq \Pr[\ell(\bm{\lambda}_q)\leq 2]\\
        &= f(q).
    \end{align*}
    The inequality $f(p)\geq f(q)$ established above proves the claim.
\end{proof}
We now return to the setting in \Cref{lem:3_letter_lds_lem}. Since the distribution $(1-2t,t,t)$ with $0\leq t\leq 1/3$ majorizes all distributions of the form $(1-x-y,x,y)$ where $0\leq t\leq y\leq x\leq 1-x-y$, the probability in the left-hand side of \cref{eq:lds_prob} is minimized by the choice $x=y=t$.

\subsection{Proof of the upper bound}
Our proof of the upper bound in~\Cref{thm:improved_ub_qutrits} assumes the input is a product state which has trace distance precisely $\veps/2$ from the property of interest. This assumption about the trace distance is without loss of generality, by the following lemma.
Here, we let $\mathsf{Prod}_2:=\mathsf{Prod}_2(n,r=2)$ for brevity and set $\dTr(\psi,\mathsf{Prod}_2):=\min_{\varphi\in\mathsf{Prod}_2}\dTr(\psi,\varphi)$ for a pure state $\psi$. We also assume a fixed positive integer $N$ has been given, and let $P_{\leq 2}\in \rmL((\mathbb{C}^3\otimes \mathbb{C}^3)^{\otimes N})$ denote the PC-optimal $N$-copy Schmidt-rank test of the form in \Cref{thm:PC_optimal_sr_test}, when $r=2$.
\begin{lemma}\label{lem:worst_case_eps}
    Let $n$ be a positive even integer, $\veps\in (0,1]$, and $\calS_{\veps}\subset \rmP((\mathbb{C}^3)^{\otimes n})$ be the set of all states which are $\veps$-far from $\mathsf{Prod}_2$ and of the form $\bigotimes_{i=1}^{n/2} \psi_{2i-1,2i}$ where $\psi_{2i-1,2i}\in\rmP((\mathbb{C}^3)^{\otimes 2})$ is a bipartite state on subsystems $2i-1$ and $2i$, for each $i\in [n/2]$. There exists an optimal solution $\psi$ to the optimization problem
    \begin{align*}
        \max_{\psi\in\calS_{\veps}}\left\{\prod_{i=1}^{n/2}\Tr\left(P_{\leq 2}\psi_{2i-1,2i}^{\otimes N}\right):  \psi=\bigotimes_{i=1}^{n/2}\psi_{2i-1,2i}\right\}
    \end{align*}
    satisfying $\dTr(\psi,\mathsf{Prod}_2)=\veps$.
\end{lemma}
\begin{proof}
    We will show that for any solution, there exists another solution which has trace distance from $\mathsf{Prod}_2$ precisely equal to $\veps$, and for which the objective function (acceptance probability) is at least as high. Let $\ket{\psi}\in(\mathbb{C}^3)^{\otimes n}$ be a unit vector such that $\psi=\outerprod{\psi}{\psi}$ and assume $\dTr(\psi,\mathsf{Prod}_2)>\veps$. By \Cref{prop:eckart_young_mirsky} and the multiplicativity of fidelity, this is true if and only if $\prod_{i=1}^{n/2}(1-\alpha_3^{(i)})<1-\veps^2$, where $\alpha^{(i)}:=(\alpha_1^{(i)},\alpha_2^{(i)},\alpha_3^{(i)})$ are the \textit{squared} Schmidt coefficients of the $i^\text{th}$ state in the product $\ket{\psi}=\bigotimes_{i=1}^{n/2}\ket{\psi_{2i-1,2i}}$, sorted in non-increasing order, for each $i\in [n/2]$. First, note that the objective function in the maximization is equal to $\prod_{i=1}^{n/2}p_i$ where
    \begin{align}\label{eq:p_i}
        p_i := \Pr_{\bm{w}\sim (\alpha^{(i)})^{\otimes N}}\left[\LDS(\bm{w})\leq 2\right].
    \end{align}
    Hence, we can assume without loss of generality that the $i^\text{th}$ triple of squared Schmidt coefficients is equal to $(1-2\alpha_3^{(i)},\alpha_3^{(i)},\alpha_3^{(i)})$, for this majorizes every triple of the form $(\alpha_1^{(i)},\alpha_2^{(i)},\alpha_3^{(i)})$ and by \Cref{lem:lds_helper}, does not decrease the probability of acceptance. Let $f_i\coloneq 1-\alpha_3^{(i)}$ so that $f_1,\dots,f_{n/2}\in [2/3,1]$. Now, suppose there exists a sequence of reals $f^\prime_1,\dots,f^\prime_{n/2}\in [2/3,1]$ such that $\prod_{i=1}^{n/2}f_i^\prime=1-\veps^2$ and $f_i^\prime\geq f_i$ for all $i\in [n/2]$. Then we may set $\mu_3^{(i)}\coloneq 1-f_i^\prime$ and $\ket{\psi^\prime} \in (\mathbb{C}^3)^{\otimes n}$ to be any unit vector which is a product of $n/2$ bipartite states having the squared Schmidt coefficients $(1-2\mu_3^{(i)},\mu_3^{(i)},\mu_3^{(i)})$, for each $i\in [n/2]$. Such a unit vector results in a state $\psi^\prime = \outerprod{\psi^\prime}{\psi^\prime}$ for which $\dTr(\psi^\prime,\mathsf{Prod}_2)=\veps$, by construction. Moreover, since $\mu^{(i)}_3\leq \alpha_3^{(i)}$, the $i^\text{th}$ probability $p_i$ is nondecreasing upon replacing $\alpha^{(i)}$ with $\mu^{(i)}$ in \cref{eq:p_i}, once again by \Cref{lem:lds_helper}.

    The sequence $f^\prime_1,\dots,f^\prime_{n/2}\in [2/3,1]$ can be constructed explicitly as follows. If $\prod_{i=2}^{n/2}f_i\geq 1-\veps^2$ then set $f_1^\prime=(1-\veps^2)\prod_{i=2}^{n/2}f_i^{-1}$ and $f_j^\prime = f_j$ for each $j\geq 2$. Then $f_1^\prime \leq 1$ and
    \begin{align*}
        \frac{f_1^\prime}{f_1} &= (1-\veps^2)\prod_{i=1}^{n/2}f_i^{-1}> 1
    \end{align*}
    where the inequality follows from the assumption that $\prod_{i=1}^{n/2}f_i < 1-\veps^2$. This implies $f^\prime_1\geq f_1\geq 2/3$, as required. Otherwise, if $\prod_{i=2}^{n/2}f_i < 1-\veps^2$ then set $f_1^\prime=1$ and repeat the procedure with the sequence $f_2^\prime,\dots,f_{n/2}^\prime$. Either this procedure terminates before the $(n/2)^\text{th}$ iteration, or we can take $f_j^\prime=f_j$ for $j\in \{1,2,\dots, n/2-1\}$ and $f_{n/2}^\prime=1-\veps^2$.
\end{proof}
We can now prove the main result of this section.
\begin{theorem}\label{thm:improved_ub_qutrits}
    For $\veps \in (0,\delta_0]$ where $0<\delta_0\leq 1/\sqrt{2}$ is a universal constant there exists an algorithm for testing whether $\psi\in \mathsf{Prod}_2(n,r=2)\subseteq \rmP\left((\mathbb{C}^3)^{\otimes n}\right)$ with one-sided error using $O(\sqrt{n}/\veps^4)$ copies of $\psi$. Moreover, for $\veps\in (0,1/\sqrt{2}]$ any algorithm with one- or two-sided error for this task requires $\Omega(\sqrt{n}/\veps^2)$ copies.
\end{theorem}
\begin{proof}
We abbreviate the property in the theorem by $\mathsf{Prod}_2$ throughout. The lower bound for testing this property with two-sided error is proven in the same way as the lower bound for MPS testing in Section~5 in~\cite{soleimanifar2022testingmps}, and we record this proof with the corresponding changes in~\Cref{sec:sqrt_lb}. In the remainder of this section, we focus solely on the upper bound.

The test which proves the upper bound has two phases. In the first phase, we perform the product test on $N_0$ copies of the state, and in the second phase we use $N_1$ copies of the state to perform two iterations of the PC-optimal Schmidt-rank test on each of the $n/2$ subsystems. The test accepts if both of these phases accept, and rejects otherwise. Since the tests in both phases have perfect completeness, i.e., always accept states in $\mathsf{Prod}_2$, it suffices to bound the probability of acceptance on an $\veps$-far state.

We begin with an analysis of the first phase. Let $\veps_0\in (0,\veps/2]$ be a parameter to be specified shortly and note that for any $\psi\in\rmP((\mathbb{C}^3)^{\otimes n})$ which is $\veps$-far from $\mathsf{Prod}_2$, at least one of the following must hold:
\begin{enumerate}
\item[i)] $\psi$ is $\veps_0$-far from product with respect to the $n/2$ subsystems, or
\item[ii)] the closest product state to $\psi$ with respect to the $n/2$ subsystems, denoted by $\widetilde{\psi}$, is $\veps/2$-far from $\mathsf{Prod}_2$.
\end{enumerate}
Otherwise, for some $\varphi\in\mathsf{Prod}_2$ we would have
\begin{align*}
    \rmd_{\mathrm{Tr}}(\psi,\varphi) &\leq \rmd_{\mathrm{Tr}}(\psi,\widetilde{\psi}) + \rmd_{\mathrm{Tr}}(\widetilde{\psi},\varphi)< \veps
\end{align*}
which contradicts the assumption that $\psi$ is $\veps$-far from $\mathsf{Prod}_2$. Let $\delta \in (0,\veps/2]$ and set $\veps_0=\delta N_1^{-1/2}$. If Item i) is true then taking $N_0=O(1/\veps_0^2)=O(N_1/\veps^2)$ copies of the state, the product test\footnote{Note that, here, by ``product test" we are referring to the test which performs multiple iterations of the 2-copy product test of \cite{harrow2013testing} to amplify the probability of rejection.} will reject with probability at least $99/100$.

We now assume that Item i) is false and analyze the performance of the second phase of the test. In this case, we have $\rmd_{\mathrm{Tr}}(\psi,\widetilde{\psi})< \veps_0$, which in turn is true if and only if $|\langle \psi |\widetilde{\psi}\rangle|^{2N}> (1-\veps_0^2)^{N}$. Then $|\langle \psi |\widetilde{\psi}\rangle|^{2N}> 1-\delta^2$, or equivalently, $\rmd_{\mathrm{Tr}}(\psi^{\otimes N},\widetilde{\psi}^{\otimes N})< \delta$. Hence, we may analyze the performance of the second phase on $N_1$ copies of $\widetilde{\psi}$ instead of the actual input state $\psi$. This is because the maximum difference in the acceptance probability from doing so will be $\delta$ and this can be taken to be arbitrarily close to zero.

Since we are assuming Item i) above is false we have $\tau:=\min_{\varphi\in \mathsf{Prod}_2}\dTr(\varphi, \widetilde{\psi})\geq \veps/2$. Additionally, since our goal is to upper bound the acceptance probability of performing Schmidt-rank tests on each of the $n/2$ subsystems, by \Cref{lem:worst_case_eps} we can assume without loss of generality that $\tau=\veps/2$. Writing $\widetilde{\psi}=\outerprod{\widetilde{\psi}}{\widetilde{\psi}}$ with $\ket{\widetilde{\psi}}=\bigotimes_{i=1}^{n/2}\ket{\widetilde{\psi}_{2i-1,2i}}$ for some $\ket{\widetilde{\psi}_{2i-1,2i}}\in(\mathbb{C}^3)^{\otimes 2}$, we may define the quantities
\begin{align*}
\veps_i:=\sqrt{1-\max_{\varphi\in \mathsf{SR}(2)}|\langle \widetilde{\psi}_{2i-1,2i}| \varphi \rangle|^2}
\end{align*}
for each $i\in [n/2]$. The squared overlap between $\widetilde{\psi}$ and the closest state in $\mathsf{Prod}_2(n,r=2)$ is equal to $\prod_{i=1}^{n/2}(1-\veps_i^2) = 1-\veps^2/4$, and we have
\begin{align}\label{eq:eps_sum_ineq}
    2\veps^2/7\geq \sum_{i=1}^{n/2}\veps_i^2\geq \veps^2/4.
\end{align}
where the second inequality is due to the Weierstrass product inequality, and the first follows from
\begin{align*}
    \frac{\veps^2/4}{1-\veps^2/4}
    &\geq -\log(1-\veps^2/4)
    = -\sum_{i=1}^{n/2}\log(1-\veps_i^2)
    \geq \sum_{i=1}^{n/2}\veps_i^2
\end{align*}
and the assumption that $0< \veps\leq 1/\sqrt{2}$.
We next analyze the acceptance probabilities $p_i$ of the Schmidt-rank tests performed on (some number of copies of) the states $\widetilde{\psi}_{2i-1,2i}$, for each $i\in [n/2]$. There are two cases of interest. Here, $N_2$ and $N_3$ are fixed functions of $n$ and $\veps$ to be specified shortly.

\paragraph{Case i)} If $\veps_i^2> 0.05\cdot c_1\veps^2/\sqrt{n}$ for some $i\in [n/2]$ then the $i^\text{th}$ Schmidt-rank test rejects with probability at least $99/100$ using $N_{2}=O(\sqrt{n}/\veps^2)$ copies, by \Cref{thm:PC_optimal_sr_test}.

\paragraph{Case ii)} If $\max_{i\in [n/2]}\veps_i^2\leq 0.05 \cdot c_1\veps^2/\sqrt{n}$ then performing the $n/2$ Schmidt-rank tests on $N_{3}=\lfloor 20\sqrt{n}/\veps^2 \rfloor$ copies of $\widetilde{\psi}$ results in constant rejection probability, as we now show. We may characterize the acceptance probability of the $i^\text{th}$ Schmidt-rank test using \Cref{lem:3_letter_lds_lem} as follows. From the discussion in \Cref{sec:sr_rank_testing_equivalence}, the smallest Schmidt coefficient $\lambda_3$ of $\ket{\widetilde{\psi}_{2i-1,2i}}$ satisfies $\lambda_3\geq \veps_i$. Furthermore, the probability of acceptance of the Schmidt-rank test applied to this $\veps_i$-far state is precisely the probability of acceptance of the rank test when the spectrum is $(\lambda_1^2,\lambda_2^2,\lambda_3^2)$, for some $\lambda_1,\lambda_2\geq 0$ such that $\lambda_1\geq \lambda_2\geq \lambda_3\geq 0$. (See \Cref{sec:sr_rank_testing_equivalence} for further details on both these points.) By \Cref{lem:lds_characterization}, the acceptance probability is
\begin{align*}
    p_i &= \Pr[\LDS(\bm{w})\leq 2] = 1-\Pr[\LDS(\bm{w})=3] \leq 1 - \frac{N_3^2\veps_i^4}{4} + c_2 \veps_i^2
\end{align*}
where the random word $\bm{w}\in \{1,2,3\}^{N_3}$ has $N_3$ random letters with weights given by the spectrum $(\lambda_1^2,\lambda_2^2,\lambda_3^2)$, and the inequality follows from \Cref{lem:3_letter_lds_lem} whose hypothesis is satisfied since $N_3\veps_i^2\leq c_1$. This means so long as $\veps$ is at most, say, $\sqrt{10/c_2}$, the overall probability of acceptance $\prod_{i=1}^{n/2}p_i$ is at most
\begin{align}
    \prod_{i=1}^{n/2}\left(1-\frac{N_3^2\veps_i^4}{4}+c_2\veps_i^2\right)
    &\leq \exp\left(-\frac{N_3^2}{4}\sum_{i=1}^{n/2}\veps_i^4 + \frac{2c_2\veps^2}{7}\right)& (\textnormal{$1^{\text{st}}$ ineq. in \cref{eq:eps_sum_ineq}})\nonumber\\
    &\leq \exp\left(\frac{-N_3^2\left(\sum_{i=1}^{n/2}\veps_i^2\right)^2}{2n}+\frac{20}{7}\right)& (\textnormal{see below})\label{eq:see_below}\\
    &\leq \exp\left(\frac{-N_3^2\veps^4}{32n}+\frac{20}{7}\right)& (\textnormal{$2^{\text{nd}}$ ineq. in \cref{eq:eps_sum_ineq}})\nonumber\\
    &\leq \ee^{-15/56}& (N_3\geq 10\sqrt{n}/\veps^2)\nonumber
\end{align}
where in \cref{eq:see_below} we used the inequality between norms $\norm{v}_1\leq \sqrt{D}\norm{v}_2$ for a $D$-dimensional vector $v$ as well as the assumption that $\veps\leq \sqrt{10/c_2}$.

From the above analysis, we may conclude that if we first perform the $n/2$ Schmidt-rank tests using $N_2$ copies of $\widetilde{\psi}$ and then repeat the same procedure using $N_3$ copies of $\widetilde{\psi}$, an accept outcome will be obtained in both trials with probability at most $\ee^{-15/56}$. Also, this second phase uses $N_1=N_2+N_3$ copies. The overall probability of acceptance of the two phases of the test is thus at most $\ee^{-15/56}+\delta$, which is at most, say, $4/5$ by taking $\delta=0.01$. The total number of copies of $\psi$ used is equal to $N_0+N_1=O(\sqrt{n}/\veps^4)+O(\sqrt{n}/\veps^2)=O(\sqrt{n}/\veps^4)$. Hence, we see that the statement in the theorem holds with $\delta_0=\min\{\sqrt{10/c_2}, \ 1/\sqrt{2}\}$, for $c_2 > 0$ the universal constant appearing in \Cref{lem:3_letter_lds_lem}.
\end{proof}
In this proof, the need to perform the Schmidt-rank tests using two ``batches" is likely an artifact of the analysis, rather than a genuine requirement for this test.

\section{Few-copy tests with perfect completeness}\label{sec:few_copy}
In this section, we give upper bounds and matching (or nearly so) lower bounds for testing rank, Schmidt-rank, or trees with perfect completeness when the tester is restricted to performing measurements on few copies at a time, possibly adaptively. Along the way, we address an open question raised in~\cite{harrow2013testing,Montanaro2016,soleimanifar2022testingmps} by giving a closed-form solution to the error probability of the rank (or Schmidt-rank) test when performed on $r+1$ copies of the unknown state, where $r$ is the rank parameter.

\subsection{Adaptivity offers no advantage for irreducible properties of pure states}\label{sec:adaptivity}
In this section we prove that if $\calP \subseteq \rmP(\calH)$ is an irreducible variety, then adaptivity offers no advantage for multi-round tests with one-sided error. We first require some background on algebraic varieties.

For our purposes, a \textit{(projective) variety} is a subset $\calP \subseteq \rmP(\calH)$ for which
\begin{align*}
\calP=\{\psi : f_1(\ket{\psi})=\dots = f_p (\ket{\psi}) =0\}
\end{align*}
for some collection of homogeneous polynomials $f_1,\dots, f_p\in \mathbb{C}[\calH]$. We say that $f_1,\dots, f_p$ \textit{cut out $\calP$}. We say that a variety $\calP$ is \textit{irreducible} if it cannot be written as the union of two proper subvarieties, i.e. $\calP \neq \calQ \cup \mathcal{R}$ for any varieties $\calQ, \mathcal{R} \subsetneq \calP$. It is a standard fact that if $\calP$ is an irreducible variety, and $\calQ \subsetneq \calP$ is a proper subvariety, then $\calP \setminus \calQ$ is Euclidean dense in $\calP$~\cite[Theorem 2.33]{mumford1995algebraic}. If $\calP$ is a variety, then
\begin{align*}
\nu_{\ell}(\calP):= \{ \psi^{\otimes \ell} : \psi \in \calP\} \subseteq \rmP(\calH^{\otimes \ell})
\end{align*}
is a variety, known as the $\ell$-th \textit{Veronese embedding} of $\calP$. Concretely, $\nu_{\ell}(\calP)$ is cut out by the union of three sets of polynomials: the linear forms $\tilde{f}_{i,\ket{x}}:=f_i \circ (\mathds{1}_{\calH} \otimes \bra{x}) : \calH^{\otimes \ell} \rightarrow \mathbb{C}$
as $\ket{x}$ ranges over a basis of $\calH^{\otimes \ell-1}$ and $i$ ranges over $[p]$ (the evaluation of the linear form at a point $\ket{y} \in \calH^{\otimes \ell}$ is given by $\tilde{f}_{i,\ket{x}}(\ket{y})=f_i( (\mathds{1}_{\calH} \otimes \bra{x})\ket{y})$), the $2 \times 2$ minors cutting out the set of pure product states in $\rmP(\calH^{\otimes \ell})$, and the linear forms that define $\vee^{\ell}(\calH)\subseteq \calH^{\otimes \ell}$. If $\calP$ is an irreducible variety, then $\nu_{\ell}(\calP)$ is also an irreducible variety. One way to see this is to note that if $\mathcal{R} \subsetneq \nu_{\ell}(\calP)$ is a proper subvariety, then $\calQ:=\{\psi : \psi^{\otimes \ell} \in \mathcal{R}\}$ is a proper subvariety of $\calP$.

\theoremstyle{plain}
\newtheorem*{adaptive}{\Cref{thm:adaptive}}
\begin{adaptive}
	Let $\calP \subseteq \rmP(\calH)$ be an irreducible variety and $P_{\textnormal{accept}}$ be a PC-optimal $\ell$-copy test for $\calP$. Then among all adaptive $k$-round, $\ell$-copy tests for $\calP$, the test $(P_{\textnormal{accept}})^{\otimes k}$ is PC-optimal.
\end{adaptive}

\begin{remark}
If $\calP$ is not an irreducible variety, then adaptive measurements may help. For example, consider $\calP=\{\outerprod{1}{1},\outerprod{+}{+}\} \subseteq \rmP(\mathbb{C}^2)$, where we define $\ket{+}=\frac{1}{\sqrt{2}}(\ket{1}+\ket{2})$ and $\ket{-}=\frac{1}{\sqrt{2}}(\ket{1}-\ket{2})$. Then the PC-optimal 1-copy test is to simply output ``accept." However, consider the adaptive $3$-round, $1$-copy test which repeatedly performs $\{\outerprod{1}{1},\outerprod{2}{2}\}$ until outcome ``2'' is obtained, and then performs $\{\outerprod{+}{+},\outerprod{-}{-}\}$ in the remaining rounds, outputting ``reject'' if and only if outcome ``$-$" is observed at least once. This test has perfect completeness, and if $\psi$ is $\frac{1}{\sqrt{2}}$-far from $\calP$, then outcome ``reject" will occur with probability at least $\frac{1}{4}$.
\end{remark}

\begin{proof}[Proof of \Cref{thm:adaptive}.]
    Throughout, we let $\Pacc=\Proj\Span\{\ket{\psi}^{\otimes \ell}:\psi\in\calP\}$. Let $P_{x_1}\in \rmL(\calH^{\otimes \ell})$ denote the possible measurement operators (indexed by outcomes $x_1$) in the first round and for each $2\leq j\leq k-1$ let $P^{x_1,\dots,x_{j-1}}_{x_j}\in\rmL(\calH^{\otimes \ell})$ be the $x_j^{\text{th}}$ measurement operator of the $j^{\text{th}}$ round conditioned on outcomes $x_1,\dots,x_{j-1}$ being obtained in the previous rounds. Fix a set of outcomes $x_1,\dots, x_{k-1}$. Without loss of generality, the final measurement conditioned on obtaining these outcomes is a two-outcome POVM which we denote by $\{P_0^{x_1,\dots,x_{k-1}}, \mathds{1}-P_0^{x_1,\dots,x_{k-1}}\}$. Define
    \begin{align*}
        \nu_{\ell}(\calP):=\{\psi^{\otimes \ell} : \psi \in \calP \} \quad \textnormal{and}\quad K := \bigg\{ \phi : \ket{\phi} \in  \ker(P_{x_1})\cup \bigcup_{j=1}^{k-2}\ker(P^{x_1,\dots,x_{j}}_{x_{j+1}})\bigg\}.
    \end{align*}
    Suppose that $\nu_{\ell}(\calP)\subseteq K$. Then clearly an optimal test outputs ``reject" upon obtaining the outcomes $x_1,\dots, x_{k-1}$. If instead $\nu_{\ell}(\calP)$ is not contained in $K$ then we must have
    \begin{align}\label{eq:max_eig}
        \Tr(P^{x_1,\dots,x_{k-1}}_0 \phi) = 1\quad \textnormal{for all}\quad \phi \in \nu_{\ell}(\calP)\backslash K.
    \end{align}
    Indeed, assume for contradiction this were false, i.e., there exists a state $\phi=\psi^{\otimes \ell} \in \nu_{\ell}(\calP)\backslash K$ for which $\Tr(P^{x_1,\dots,x_{k-1}}_0 \phi) < 1$. Then for the input state $\psi$, the probability of obtaining the outcomes $x_1,\dots,x_{k-1}$ in the test is non-zero, and therefore the conditional probability that the test accepts given these outcomes is well-defined. This probability is in turn equal to $\Tr(P^{x_1,\dots,x_{k-1}}_0 \psi^{\otimes \ell}) < 1$. But then this contradicts the assumption of perfect completeness for the test. Now, since $\nu_{\ell}(\calP)\backslash K$ is dense in $\nu_{\ell}(\calP)$ (since $K\cap \nu_\ell(P)$ is a proper subvariety of $\nu_\ell(P)$), it follows that $\Tr(P^{x_1,\dots,x_{k-1}}_0 \phi) = 1$ for all $\phi \in \nu_{\ell}(\calP)$, so for every element $\phi=\outerprod{\phi}{\phi}$ of $\nu_{\ell}(\calP)$, the unit vector $\ket{\phi}$ is contained in the $+1$ eigenspace of $P^{x_1,\dots,x_{k-1}}_0$. It follows that $\Pacc\preceq P^{x_1,\dots,x_{k-1}}_0$. Thus, the measurement
	$P_{\textnormal{accept}}$
	performs at least as well as $P^{x_1,\dots,x_{k-1}}_0$ in the $k^\text{th}$ round.


    We have shown that an optimal measurement in the $k^{\text{th}}$ round is to simply output ``reject" if $\nu_{\ell}(\calP) \subseteq K$ holds, and to perform the measurement $P_{\textnormal{accept}}$ otherwise. Since there are only two possibilities, we can coarse-grain the preceding $(k-1)$-round, $\ell$ copy measurement into a two-outcome measurement $\{Q_{\textnormal{accept}},Q_{\textnormal{reject}}\}$ such that if the input state has the property then the outcome ``accept" always occurs (i.e. it has perfect completeness). Then on input state $\psi$, the probability to output ``accept" is given by
	\begin{align*}
		\Tr(Q_{\textnormal{accept}}\psi^{\otimes \ell(k-1)}) \Tr(P_{\textnormal{accept}}
		\psi^{\otimes \ell}).
	\end{align*}
	Thus, it is optimal to choose a $(k-1)$-round, $\ell$-copy measurement $Q_{\textnormal{accept}}$ which accepts every state in $\calP$ and has minimal probability of accepting a state not in $\calP$. By induction on $k$, an optimal such $(k-1)$-round, $\ell$-copy measurement is given by $Q_{\textnormal{accept}}=(P_{\textnormal{accept}})^{\otimes k-1}$ (the base case $k=1$ is trivial). Thus, an optimal $k$-round measurement is given by $(P_{\textnormal{accept}})^{\otimes k}.$ This completes the proof.
\end{proof}

\subsection{Closed-form expression for \texorpdfstring{$(r+1)$}{(r+1)}-copy tests}\label{sec:closed_form}
Let $\calH$ be a finite-dimensional Hilbert space of dimension $d$. In this section we consider testing the property 
\begin{align*}
    \mathsf{Rank}(r)=\{\rho \in
    \rmD(\calH) : \rank(\rho)\leq r\}
\end{align*}
using a measurement performed on just $r+1$ copies. Such a test will be ``unsuccessful" in general; nevertheless, we aim to characterize its error probability. Note that, by \Cref{thm:PC_optimal_sr_test} and the preceding discussion, this property is equivalent to testing the pure state property $\mathsf{SR}(r)\subseteq \rmP(\calH\otimes \calH^\prime)$ of having Schmidt-rank at most $r$.

Recall that the $N$-copy rank test defined in \cref{eq:rank_test_op} is $P_{\textnormal{accept}}=\sum_{\lambda:\ \ell(\lambda)\leq r} P_{\lambda}$ and is strongly PC-optimal~\cite[{Prop.~6.1}]{o2015quantum} for the property $\mathsf{Rank}(r)$. When $N=r+1$ we can equivalently write this projection operator as $P_{\textnormal{accept}}=\mathds{1}-P_{(1^{r+1})}$ where $P_{(1^{r+1})}$ is the projection onto the
antisymmetric subspace $\wedge^{r+1}(\calH)\subseteq
\calH^{\otimes {r+1}}$. The following theorem provides a closed-form expression for the error probability of the test in this case as a function of the distance parameter $\veps$.
\begin{theorem}\label{thm:error}
    For a real number $0<\veps\leq 1-\frac{r}{d}$, let
    \begin{align*}
    \beta(\veps):=\max_{\rho\ \veps\textnormal{-far from}\ \mathsf{Rank}(r)} \Tr(P_{\textnormal{accept}} \rho^{\otimes r+1})
    \end{align*}
    be the soundness parameter of the $(r+1)$-copy rank test, where ``$\veps$-far" is with respect to the trace distance. Then
	\begin{align}\label{eq:error}
    		\beta(\veps)=1-\min\left\{g_{{\floor{z}}}\left(\veps,
    		r,\frac{1-\veps}{r}\right),
    	g_{\ceil{z}}\left(\veps,r,\frac{\veps}{\ceil{z}}\right),
    	g_{d-r}\left(\veps,r,\frac{\veps}{d-r}\right)\right\},
	\end{align}
	where $z:={\frac{r\veps}{1-\veps}}$ and
	\begin{align}
		g_k(\veps,r,q)&:=\binom{r+k-1}{r+1}
		q^{r+1} +
		\binom{r+k-1}{r}(1-(r+k-1)q)q^r\nonumber \\
					&\quad\quad\quad\quad\quad\quad\quad
					+\binom{r+k-1}{r-1}(1-\veps-(r-1)q)(\veps-qk)q^{r-1}.\label{eq:gk_defn}
    \end{align}
    \end{theorem}
When $r=1$ this gives the following. We set $\omega=1-\veps$.
\begin{corollary}\label{cor:error}
For $r=1$ and $0 < \omega \leq \frac{1}{d}$, it holds that
\begin{align}\label{eq:betarank1}
\beta(1-\omega)=\begin{cases} 1-\omega+\omega^2 ,& \omega \geq 1/2\\
1+\frac{1}{2}\omega \lfloor \frac{1}{\omega}\rfloor (-2+\omega+ \omega \lfloor \frac{1}{\omega}\rfloor), & \omega \leq 1/2.
\end{cases}
\end{align}
\end{corollary}

\begin{remark}
By similar reasoning as in the proof of~\cite[Proposition 9]{soleimanifar2022testingmps}, the expression in~\cref{eq:betarank1} lower bounds the soundness parameter of the $n$-partite product test introduced in~\cite{harrow2013testing}. Our result shows that this bound is tight when $n=2$, and~\cite[Proposition 9]{soleimanifar2022testingmps} shows that this this bound is tight when $\omega \geq 1/2$. We ask if it is also tight when $n > 2$ and $\omega < 1/2$.
\end{remark}

\begin{figure}
\begin{center}
   \includegraphics[width=0.7\textwidth]{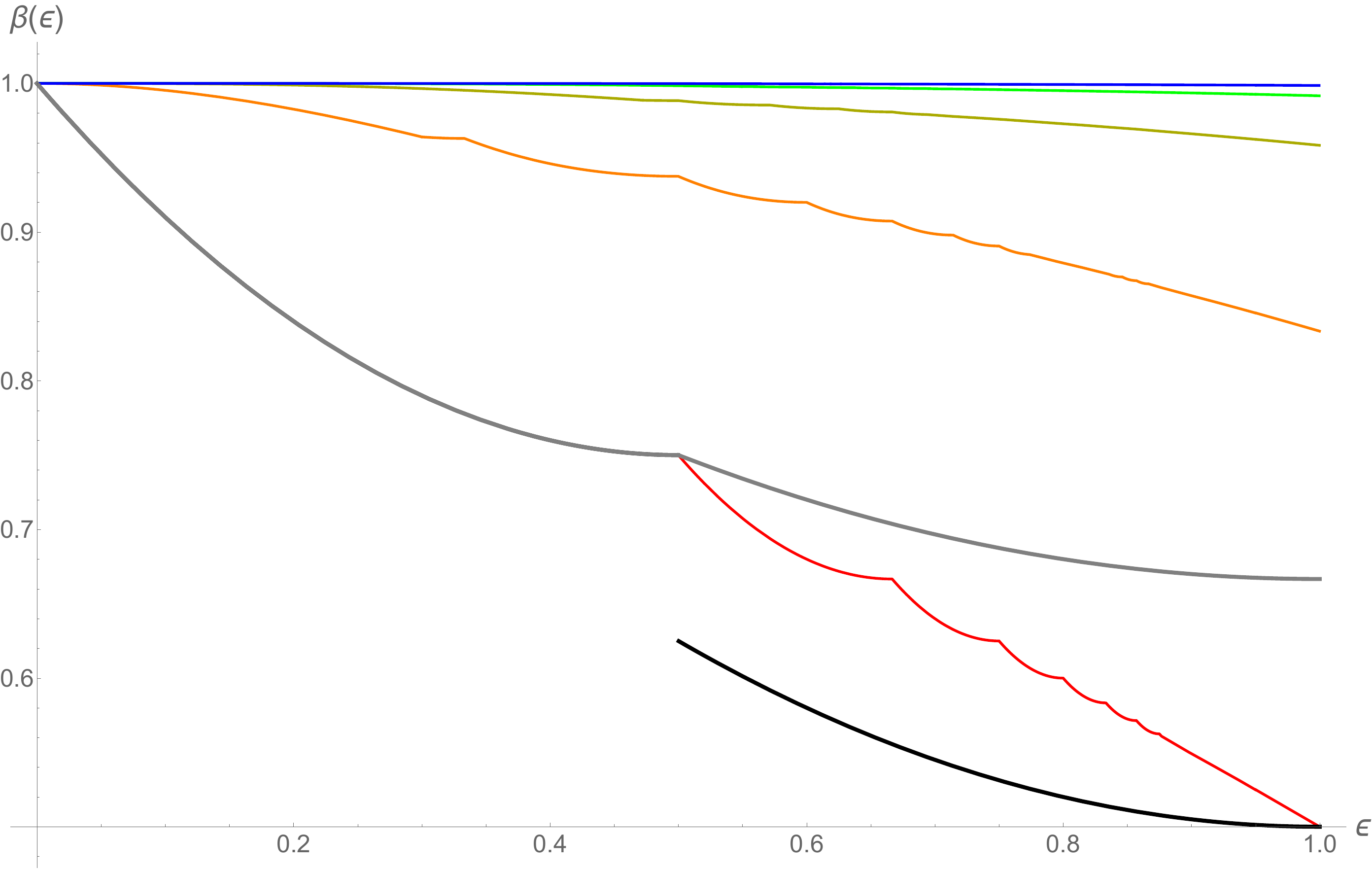}
   \end{center}
   \caption{The error probability $\beta(\veps)$ for $d$ large and $r=1$ (red), $r=2$ (orange), $r=3$ (yellow), $r=4$ (green), and $r=5$ (blue), along with the previous best known upper bound for $r=1$ (grey)~\cite[Theorem 8]{soleimanifar2022testingmps}, and the previous best known lower bound for $r=1$ (black)~\cite{harrow2013testing}. The red and grey plots are equal for $\veps\leq 1/2$.}\label{fig:plot}
\end{figure}
We prove this theorem in~\Cref{sec:closed_form_proof}. The constraint $0 < \veps \leq 1-\frac{r}{d}$ is necessary, as there does not exist any state $\rho$ that is more than $(1-\frac{r}{d})$-far from $\mathsf{Rank}(r)$ in trace distance. All three terms in the minimum are necessary; for example, when $r=2$ and $d$ is large, the third term is the minimum for $\veps$ between $0$ and $\sim 0.3$, the second term is the minimum for $\veps$ between $\sim 0.3$ and $\sim 0.33$, and the first term is the minimum for $\veps$ between $\sim 0.33$ and $1$. We plot $\beta(\veps)$ in~\Cref{fig:plot}. If $z={\frac{r\veps}{1-\veps}}$ is an integer, then the first two terms in the minimum in~\cref{eq:error} are equal:
 \begin{align}\label{eq:g1}
 g_{{\floor{z}}}\left(\veps,
		r,\frac{1-\veps}{r}\right)=
	g_{\ceil{z}}\left(\veps,r,\frac{\veps}{\ceil{z}}\right)=\binom{\frac{r}{1-\veps}}{r+1}
	\left(\frac{1-\veps}{r}\right)^{r+1}.
\end{align}
	As $d\rightarrow \infty$ the third term becomes
	\begin{align}\label{eq:g2}
	\lim_{d\rightarrow \infty} g_{d-r}\left(\veps,r,\frac{\veps}{d-r}\right)= \frac{\veps^r (r+1-\veps r)}{(r+1)!}.
	\end{align}
As $\veps \rightarrow 1$, both of these quantities approach $\frac{1}{(r+1)!}$. This gives:
\ba\label{eq:affirmative}
	\lim_{\veps \rightarrow
	1} \lim_{d\rightarrow \infty}
	\beta(\veps)=1-\frac{1}{(r+1)!}.
\ea
This value is approached when the test is applied to $(r+1)$ copies of the  
$d$-dimensional maximally mixed state. If $d \gg 0$ and $\veps< \frac{1}{r+2}$, then $z<1$ and we obtain
\begin{align}\label{eq:small_z}
\beta(\veps)&=1-\min\left\{\veps \left(\frac{1-\veps}{r}\right)^r, (1-r \veps) \veps^r, \frac{\veps^r (r+1-\veps r)}{(r+1)!}\right\}\\
&=1-\begin{cases} (1-\veps)\veps,& r=1\\
\frac{\veps^r (r+1-\veps r)}{(r+1)!},& r\geq 2.
\end{cases}
\end{align}

\subsection{Copy complexity of few-copy tests}\label{sec:few_copy_compelxity}
In this section, we use the closed-form expression derived above to conclude tight bounds on the copy complexity of testing trees with few-copy, possibly adaptive measurements. We first show the following.
\begin{theorem}\label{thm:l_copy}
    For $0<\veps < \frac{1}{r+2}$, the copy complexity of testing $\mathsf{SR}(r)$ with one-sided error using adaptive $(r+1)$-copy measurements is $\Theta((r+1)!/\veps^{2r})$.
\end{theorem}
We remark that the equivalence between rank testing and Schmidt rank testing observed in part one of~\Cref{thm:PC_optimal_sr_test} can be used to obtain the same copy complexity $\Theta((r+1)!/\veps^{r})$ for testing $\mathsf{Rank}(r)$. Also, since, for fixed $\veps$ and $r$, we have $\beta(\veps)$ is monotonically increasing in $d$, we can take $d\gg 1$ in the proofs below without loss of generality.
\begin{proof}
Since $\mathsf{SR}(r)$ is an irreducible variety (see e.g.~\cite[Proposition 12.2]{harris2013algebraic}), \Cref{thm:adaptive} applies and for any $\ell$ and $k$ the PC-optimal $k$-round, $\ell$-copy test for $\mathsf{SR}(r)$ is given by simply repeating the $\ell$-copy Schmidt-rank test. When $r=1$ the rank test has soundness $\beta(\veps)=1-(1-\veps)\veps$ and, by the discussion on trace distance from \Cref{sec:sr_rank_testing_equivalence}, the soundness of the $(r+1)$-copy Schmidt-rank test at distance $\veps$ is $1-(1-\veps^2)\veps^2$. Thus, $N=\Theta(1/\veps^2)$ copies are necessary and sufficient to ensure the soundness parameter is at most $1/3$. When $r\geq 2$ and $d\gg 1$, we have $\beta(\veps)= 1-\frac{\veps^r (r+1-\veps r)}{(r+1)!}$ and, by a similar reasoning, $N=\Theta((r+1)!/\veps^{2r})$ copies are necessary and sufficient to ensure the soundness parameter is at most $1/3$ for the Schmidt-rank test.
\end{proof}
The above result also implies bounds on the copy complexity of testing TTNS using $(r+1)$-copy measurements.
\theoremstyle{plain}
\newtheorem*{fewcopy}{\Cref{thm:trees_few_copies}}
\begin{fewcopy}
Let $G=(V,E)$ be a tree on $n$ vertices, $r \geq 2$ be a positive integer, and $\veps \in (0, \frac{1}{r+2})$. There exists an algorithm which tests whether $\psi \in \mathsf{TTNS}(G,r)$ with one-sided error using ${O((n-1)^r (r+1)!/\veps^{2r})}$ copies and measurements performed on just $r+1$ copies at a time. Furthermore, any algorithm with one-sided error which measures $r+1$ copies at a time, possibly adaptively, requires ${\Omega((n-1)^{r-1} (r+1)! / \veps^{2r})}$ copies.
\end{fewcopy}

\begin{proof}
The proof of the upper bound is similar to that of~\Cref{thm:ttns_testing}. By \Cref{cor:ttns_far}, for any $\veps$-far state there is an edge $e\in E$ with respect to which the state is $\veps/\sqrt{n-1}$-far from being in the property $\mathsf{SR}(r)$. The upper bound then follows from the upper bound on Schmidt-rank testing in \Cref{thm:l_copy}, in an identical manner to the proof of~\Cref{thm:ttns_testing}.

For the lower bound, set $m=n-1$ and note that, for any state $\ket{\phi} \in \mathbb{C}^d \otimes \mathbb{C}^d$ which is $({2\veps/\sqrt{m}})$-far from $\mathsf{SR}(r)$, the corresponding state $\ket{\phi_G}$, as defined in \Cref{sec:tree_like} is $\veps$-far from $\mathsf{TTNS}(G,r)$, by \Cref{lem:overlap_multiplicative_lemma}. Let $Q$ be the PC-optimal $(r+1)$-copy test for $\mathsf{SR}(r)$ presented in \Cref{thm:PC_optimal_sr_test}. Following~\Cref{thm:error}, we may then pick a state $\ket{\phi} \in \mathbb{C}^d \otimes \mathbb{C}^d$ which is $(2\veps / \sqrt{m})$-far from $\mathsf{SR}(r)$ and satisfies the property that $Q$ erroneously accepts $\ket{\phi}^{\otimes r+1}$ with probability at least $1-\frac{(4\veps^2/m)^r}{r!}.$ This is possible because the error probability is at least $1-g_{d-r}(4\veps^2/m, r, \frac{4\veps^2/m}{d-r})$, which approaches the strictly larger quantity
\begin{align*}
    1-\frac{(4\veps^2/m)^r (r+1-4\veps^2 r/m)}{(r+1)!}
\end{align*}
as $d \rightarrow \infty$. Let $P$ be the projection onto $\text{span}\{\ket{\psi}^{\otimes r+1} : \ket{\psi} \in \mathsf{TTNS}(G,r)\}$, which is the PC-optimal $(r+1)$-copy measurement for $\mathsf{TTNS}(G,r)$. Then $P \succeq Q^{\otimes m}$, by \Cref{lem:pc_optimal_test_lb_operator}.
Since $\mathsf{TTNS}(G,r)$ is an irreducible variety (see e.g.~\cite{bernardi2023dimension}),~\Cref{thm:adaptive} applies, and the optimal adaptive, $(r+1)$-copy measurement for $\mathsf{TTNS}(G,r)$ is $P^{\otimes \frac{N}{r+1}}$. To conclude the proof, we show that if the acceptance probability of $P^{\otimes \frac{N}{r+1}}$ on $\ket{\phi}^{\otimes Nm}$ is upper bounded by $1/3$, then $N=\Omega(m^{r-1} (r+1)! / \veps^{2r})$. First, it necessarily holds that
\begin{align*}
(1-(4\veps^2/m)^r/r!)^{\frac{N m}{r+1}}\leq 1/3,
\end{align*}
as the left-hand side lower bounds the acceptance probability of $Q^{\otimes \frac{mN}{r+1}}$ on $\ket{\psi}$. But this inequality can only hold if $N=\Omega(m^{r-1} (r+1)! / \veps^{2r})$.
\end{proof}

\subsection{Proof of Theorem~\ref{thm:error}}\label{sec:closed_form_proof}
In this section we prove~\Cref{thm:error}. The proof requires the following two facts. 
\begin{claim}\label{claim:majorization_claim}
    Every probability vector $\Vec{p} \in \mathbb{R}_+^{d}$ for which $p_1\geq \dots \geq p_d \geq 0$ and $p_{r+1}+\dots+ p_d \geq \veps$ is majorized by a probability vector of the form
    \begin{align*}
        (1-\veps- (r-1)q, q,\dots, q, \veps-q k, 0,\dots, 0)
    \end{align*}
    for some $\frac{\veps}{d-r}\leq q \leq \frac{1-\veps}{r}$, where $k:=\floor{\veps/q}$ and the real number $q$ is repeated $k+r-1$ times.  
\end{claim}
\begin{proof}
        Let $\tilde{\veps}:=p_{r+1}+\dots+ p_d \geq \veps$. Note that we necessarily have $\frac{\veps}{d-r}\leq p_r \leq \frac{1-\veps}{r}$. The probability vector $\Vec{p}$ is majorized by the probability vector
        \begin{align*}
    		(1-\tilde{\veps}-(r-1)p_r,p_r,\dots,
    		p_r,\tilde{\veps}-p_{r} \floor{\tilde{\veps}/p_r},0,\dots,0),
        \end{align*}
        where $p_r$ is repeated $r+\floor{\tilde{\veps}/p_r}-1$ times. In turn, this probability vector is majorized by
        \begin{align*}
        (1-\veps-(r-1)p_r,p_r,\dots, p_r,\veps-p_r \floor{\veps/p_r},0,\dots,0),
        \end{align*}
        which is of the desired form.
\end{proof}

In the following claim we consider a positive real number $\veps\in (0,1)$, a positive integer $r$ and a real number $k$. We also let
\begin{align*}
    g_{\veps,r}(k):=g_k(\veps,r,\frac{\veps}{k})
    			&=\binom{r+k-1}{r}\left(\frac{\veps}{k}\right)^r \left( 1- \frac{r \veps}{r+1}- \frac{r^2 \veps}{(r+1)k}\right),
\end{align*}
where we define $\binom{r+k-1}{r}:=(r+k-1)(r+k-2)\cdots (k)/r!$ for any $k \in \mathbb{R}$, and put $z=r\veps/(1-\veps)$.
\begin{claim}\label{claim:final}
    For any real $a$ and $b$ for which $ z \leq  a < b$, it holds that
\begin{align*}
    \min_{a\leq k \leq b} g_{\veps,r}(k)=\min\{g_{\veps,r}(a),g_{\veps,r}(b)\}.
\end{align*}
\end{claim}
\begin{proof}
Let $x=1/k$ and note that
\begin{align*}
g_{\veps,r}(k)=g_{\veps,r}(1/x)=\frac{\veps^r}{r!} \left( 1- \frac{r \veps}{r+1}- \frac{r^2 \veps x}{r+1}\right)\prod_{i=1}^{r-1} (1+j x).
\end{align*}
The term $\left( 1- \frac{r \veps}{r+1}- \frac{r^2 \veps x}{(r+1)}\right)$ and each of the terms $(1+jx)$ are positive in the range of interest $x \in (0,1/z]$. The function $h(x) = \log g_{\veps, r}(1/x)$ is therefore well-defined on this interval, and it is straightforward to show that it is concave: up to an additive constant it is equal to
\begin{align}
    \log(A-Bx)+\sum_{j=1}^{r-1}\log(1+jx)
\end{align}
where $A = 1-\frac{r\veps}{r+1}$, $B=\frac{r^2\veps}{r+1}$, and each summand has a negative second derivative on $(0,1/z]$.
In particular, $h$ is concave on $[1/b,1/a]$, so the minimum on this interval is attained at an endpoint. Moreover, since the exponential function is strictly increasing, this implies $g_{\veps,r}(1/x)$ attains its minimum at either $1/a$ or $1/b$. Equivalently, $g_{\veps,r}(k)$ attains its minimum at $a$ or $b$, which completes the proof.
\end{proof}

We may now prove the theorem.
\begin{proof}[Proof of \Cref{thm:error}]
It is known that if $\rho \in \rmD(\calH)$ has eigenvalues $p_1\geq p_2\geq\dots\geq p_d\geq 0$ then the projection of $\rho^{\otimes N}$ onto the Schur module indexed by $\lambda \vdash N$ has trace equal to $\dim(\calP_{\lambda}) s_{\lambda}(p_1,\dots, p_d),$ where $s_{\lambda}$ is the Schur polynomial corresponding to the partition $\lambda$, and $\dim(\calP_{\lambda})$ is the dimension of the irreducible representation of $\mathfrak{S}_N$ corresponding to $\lambda$~\cite{Montanaro2016}. (See also~\cite[{Sec.~3}]{haah2017sampleoptimal} or \cite[{Prop.~2.24}]{o2015quantum}.) Thus, for the $(r+1)$-copy rank test, we have
\begin{align}\label{eq:perr1}
    1-\beta(\veps)&=\min\{ \Tr(Q \rho^{\otimes r+1}): \rho\in \rmD(\calH): \rho\ \veps\text{-far from }\mathsf{Rank}(r)\}\\
    &=\min\left\{s_{(1^N)}(p_1,\dots,p_d): p_1\geq p_2\geq\dots\geq p_d\geq 0,\ \sum_{j=1}^d p_j=1,\  \sum_{j\geq r+1}p_j\geq \veps\right\}\\
    &=\min\left\{\sum_{\Vec{x}\in A_{r+1,d}}\prod_{j=1}^{r+1}p_{x_j}: p_1\geq p_2\geq\dots\geq p_d\geq 0,\ \sum_{j=1}^d p_j=1, \ \sum_{j\geq r+1}p_j\geq \veps\right\},\label{eq:schur_poly_explicit_opt}
\end{align}
where we define
\begin{align*}
    A_{r+1,d}:= \{(x_1,\dots,x_{r+1})\in [d]^{r+1}: x_1 < x_2<\dots<x_{r+1}\}.
\end{align*}
Using the Schur-Ostrowski criterion (see e.g.~\cite{peajcariaac1992convex}), it is straightforward to see that $s_{(1^{r+1})}$ is Schur-concave, i.e., that $s_{(1^{r+1})}(\Vec{p})\leq s_{(1^{r+1})}(\Vec{q})$ if $\Vec{q}\prec \Vec{p}$ ($\Vec{p}$ majorizes $\Vec{q}$) for any probability vectors $\Vec{p},\Vec{q}$ in the feasible set of the optimization in \cref{eq:schur_poly_explicit_opt}. Hence, the minimum value of the function is attained for inputs $\Vec{p}$ which are maximal with respect to the partial order $\prec$ on probability distributions, and satisfy the constraints of the optimization. For example, if $r=1$ and $\veps\leq 1/2$, it holds that $\Vec{p}=(1-\veps,\veps,0,\dots,0)$ majorizes every other distribution in the feasible set, and therefore $\beta(\veps)=1-\veps(1-\veps)$. The other cases are more involved.

\Cref{claim:majorization_claim} then implies the following expression for the error probability:
\begin{align*}
    1-\beta(\veps) 
                  &= \min \left\{g_k (\veps,r,q): \left\lfloor\frac{r\veps}{1-\veps}\right\rfloor\leq k\leq d-r,\ \frac{\veps}{k+1}<q\leq \min\left\{\frac{\veps}{k},\frac{1-\veps}{r}\right\}\right\}.
\end{align*}
The first line follows from restricting the optimization in \cref{eq:schur_poly_explicit_opt} to probability vectors of the form in \Cref{claim:majorization_claim}, while the second line uses the definition of $g_k(\cdot)$ in \cref{eq:gk_defn} as well as the observation that, for each integer $k$, the range of valid real numbers $q$ for which $k=\floor{\veps/q}$ is precisely $\frac{\veps}{k+1} < q \leq \min\{\frac{\veps}{k},\frac{1-\veps}{r}\}$.
Next, observe that
\begin{align*}
    \frac{1}{q^{r-2}}\frac{\partial g_k}{\partial q}&=\binom{r+k-1}{r+1}(r+1)q^2+\binom{r+k-1}{r}\left(r-(r+1)(r+k-1)q\right)q\nonumber\\
	&\quad \quad+\binom{r+k-1}{r-1}\Bigg(-(r-1)(\veps-q k)q\nonumber\\
	&\quad\quad\quad\quad\quad\quad\quad -(1-\veps-(r-1)q)qk+(r-1)(1-\veps-(r-1)q)(\veps-qk)\Bigg).
\end{align*}
So for $\veps$ and $r$ fixed, $\frac{\partial g_k}{\partial q}(\veps,r,q)=0$ has at most $2$ non-zero solutions $q_1, q_2$. Note that
\begin{align*}
    \left(\frac{k}{\veps}\right)^{r-2}\frac{\partial g_k}{\partial q}(\veps,r,\veps/k)=\frac{\veps}{k^2} \Bigg( &\binom{r+k-1}{r+1} \veps (r+1)
	+\binom{r+k-1}{r} (kr-\veps (r+1)(r+k-1))\nonumber\\
	&\quad\quad+\binom{r+k-1}{r-1}k(\veps(r+k-1)-k)\Bigg),
\end{align*}
which is easily seen to be non-positive whenever $r$ and $k$ are positive integers. Similarly,
\begin{align*}
    \left(\frac{k+1}{\veps}\right)^{r}\frac{\partial g_k}{\partial q}(\veps,r,\veps/(k+1))=\frac{1}{\veps}\Bigg(&\binom{r+k-1}{r+1}\veps(r+1)\nonumber\\
	&\quad+\binom{r+k-1}{r} (-\veps (r+1)(r+k-1)+r(k+1))\nonumber\\
	  &\quad+\binom{r+k-1}{r-1}\big(\veps(k^2+k-r^2+1)-(k+1)(k-r+1)\big)\Bigg).
\end{align*}
This expression is non-negative whenever $\frac{\veps}{k+1}<\frac{1-\veps}{r}$, which holds in the regime we are interested in. Thus, as $q$ ranges from $\frac{\veps}{k+1}$ to $\frac{\veps}{k}$, the function $g_k(\veps,r,q)$ starts with a non-negative slope, ends with non-positive slope, and has at most two points where the derivative is zero. It follows that the minimum value of $g_k(\veps,r,q)$ as $q$ ranges from $\frac{\veps}{k+1}$ to $\min\{\frac{\veps}{k},\frac{1-\veps}{r}\}$ occurs at one of the two endpoints. This proves that
\begin{align*}
   1- \beta(\veps)=\min\Bigg\{g_{{\floor{z}}}&\left(\veps, r,\frac{1-\veps}{r}\right), g_{\ceil{z}}\left(\veps,r,\frac{\veps}{\ceil{z}}\right),\\&g_{\ceil{z}+1}\left(\veps,r,\frac{\veps}{\ceil{z}+1}\right),
     \dots, g_{d-r}\left(\veps,r,\frac{\veps}{d-r}\right)\Bigg\}.
\end{align*}
From here, the equality
\begin{align*}
    1-\beta(\veps)=\min\Bigg\{&g_{{\floor{z}}}\left(\veps, r,\frac{1-\veps}{r}\right), g_{\ceil{z}}\left(\veps,r,\frac{\veps}{\ceil{z}}\right), g_{d-r}\left(\veps,r,\frac{\veps}{d-r}\right)\Bigg\}.
\end{align*}
follows from \Cref{claim:final}.
\end{proof}

\section*{Acknowledgments}
We thank Gero Friesecke for pointing us to Theorem~11.58 in~\cite{Hackbusch2019}, and Tim Seynnaeve and Claudia De Lazzari for valuable discussions. AL thanks Byron Chin, Norah Tan, Aram Harrow, and John Wright for helpful comments on an earlier draft of this work. AL is supported by the U.S. Department of Energy, Office of Science, National Quantum Information Science Research Centers Quantum Systems Accelerator. Additional support is acknowledged from the NSF for AL (grant PHY-1818914). BL acknowledges that this material is based upon work supported by the National Science Foundation under Award No. DMS-2202782.

\appendix
\section{Deferred proofs}
\subsection{Optimal measurements}\label{sec:optimal_measurement_proofs}
In this section we prove \Cref{thm:purification,thm:PC_optimal_sr_test}. Throughout this section, we make use of notation established in \Cref{sec:schur_weyl}. We first recall a basic fact from representation theory, which is essentially Schur's lemma. Consider a compact group $G$ with a finite-dimensional unitary representation $(V,r)$ and
let $\hat{G}$ be a set of labels for a complete set of inequivalent
irreducible representations of $\hat{G}$, denoted $\{(V^{(\lambda)},r^{(\lambda)}): \lambda\in\hat{G}\}$. As a group
representation, $V$ has a canonical (orthogonal) decomposition of the form
$V\cong \bigoplus_{\lambda\in \hat{G}}\mathbb{C}^{k_\lambda}\otimes V^{(\lambda)}$. Denote this isomorphism by $T$, so that $Tr(g)T^{-1}=\bigoplus_{\lambda\in \hat{G}} \mathds{1}_{k_{\lambda}}\otimes r^{(\lambda)}(g)$ for every $g\in G$.
\theoremstyle{plain}
\begin{fact}\label{fact:commuting_implies_blocks}
   Let $M\in \rmL(V)$ be a linear operator such that $r(g)$ commutes with $M$ for every
   $g\in G$. Then
   $T M T^{-1}=\bigoplus_{\lambda\in \hat{G}} M^{(\lambda)}\otimes \mathds{1}_{V^{(\lambda)}}$
   where for each $\lambda\in \hat{G}$ we have $M^{(\lambda)}\in  \rmL(\mathbb{C}^{m_{\lambda}})$.
\end{fact}
\begin{proof}
	Let
	$\widetilde{M}:=TMT^{-1}\in\rmL(\bigoplus_{\lambda\in\hat{G}}\mathbb{C}^{k_\lambda}\otimes V^{(\lambda)})$.
	From the hypothesis of the fact, we have that $\widetilde{M}$ commutes with
	$\bigoplus_{\lambda\in \hat{G}} \mathds{1}_{k_{\lambda}}\otimes r^{(\lambda)}(g)$ for every
	$g\in G$. If we let $\widetilde{M}^{\lambda \lambda^\prime}: V^{(\lambda^\prime)}\to V^{(\lambda)}$
	denote the action of $\widetilde{M}$ on $V^{(\lambda^\prime)}$, projected  onto $V^{(\lambda)}$, then
	we have the matrix equality
	\begin{align*}
		\widetilde{M}^{\lambda\lambda^\prime} (\mathds{1}_{k_{\lambda^\prime}}\otimes r^{(\lambda^\prime)}(g))
		&= (\mathds{1}_{k_\lambda}\otimes r^{(\lambda)}(g))\widetilde{M}^{\lambda\lambda^\prime}
	\end{align*}
	for every $g\in G$ and $\lambda,\lambda^\prime\in \hat{G}$. This in turn implies that
	for every $(i,j)\in [k_\lambda]\times [k_{\lambda^\prime}]$ and $g\in G$ we have
	\begin{align*}
		\widetilde{M}^{\lambda\lambda^\prime}_{ij}r^{(\lambda^\prime)}(g)
		&= r^{(\lambda)}(g)\widetilde{M}^{\lambda\lambda^\prime}_{ij}
	\end{align*}
	where we have defined
	$\widetilde{M}^{\lambda\lambda^\prime}_{ij}:=(\bra{i}\otimes \mathds{1}_{V^{(\lambda)}})\widetilde{M}^{\lambda\lambda^\prime}(\ket{j}\otimes \mathds{1}_{V^{(\lambda^\prime)}})$.
	By Schur's Lemma (see, e.g., Lemma~4.1.4 in \cite{goodman2009symmetry}), for every $(i,j)\in[k_{\lambda}]\times [k_{\lambda^\prime}]$ and $\lambda\in\hat{G}$ there exists a complex number $\alpha_{ij}^{(\lambda)}\in\mathbb{C}$ such that for any $\lambda^\prime$ it holds that ${\widetilde{M}^{\lambda\lambda^\prime}_{ij}=\delta_{\lambda\lambda^\prime}\alpha_{ij}^{(\lambda)}\mathds{1}_{V^{(\lambda)}}}$. It follows that $\widetilde{M}$ is of the desired form with $(M^{(\lambda)})_{ij}=\alpha^{(\lambda)}_{ij}$.
\end{proof}
In words, $M$ block-diagonalizes in an orthonormal basis partially labelled by the irreps, so that it acts within each isotypic space and non-trivially only on the multiplicity space in the tensor product. This fact allows one to derive optimal measurements (defined in \Cref{sec:property_testing_defs}) for certain symmetric properties. Below, we focus on the case of pure states, though the logic could be straightforwardly applied in the case of mixed state testing as well.
\begin{proposition}\label{prop:optimal_meas}
    Let $(\calH,r)$ be a unitary representation of a compact group $G$ and $\calP\subset \rmP(\calH)$ be a pure state property such that $r(g)\calP r(g)^\dag = \calP$ for all $g\in G$. Fix a positive integer $N$ and suppose
    \begin{align}\label{eq:symmetric_decomp}
        \quad \vee^N(\calH)&\cong \bigoplus_{\lambda} \calH_{\lambda}
    \end{align}
    is a multiplicity-free decomposition into $G$-irreps for the tensor power representation of $G$ on the symmetric subspace. An optimal measurement for $\calP$ is given by the projective measurement $\{\Pi_{\lambda}\}_{\lambda}$, where $\Pi_{\lambda}$ is the projector onto the $\lambda^{\text{th}}$ subspace in \cref{eq:symmetric_decomp}, for each $\lambda$.
\end{proposition}
\begin{proof}
    Let $\veps\in (0,1]$ and $\Qacc\in \rmL(\calH^{\otimes N})$ be any $N$-copy test for $\calP$ which has CS parameters $(a,b)$ given distance $\veps$. We will demonstrate the existence of a test $\Pacc=\sum_{\lambda}\alpha_\lambda\Pi_\lambda$ for some $\alpha_\lambda\in [0,1]$ which has CS parameters $a^\prime \geq a$ and $b^\prime \leq b$. Firstly, letting $\PSym\colon \calH^{\otimes N}\to\vee^N(\calH)$ denote the projector onto the symmetric subspace, we clearly have $\Tr(\Qacc \psi^{\otimes N})=\Tr(\PSym \Qacc \PSym \psi^{\otimes N})$ for any pure state $\psi\in \rmP(\calH)$. Hence, using the assumption that $r(g)\calP r(g)^\dag = \calP$ for every $g\in G$ we have
    \begin{align*}
    a
    &\leq \Tr(\PSym\Qacc\PSym (r(g)\psi r(g)^\dag)^{\otimes N}).
    \end{align*}
    for any $\psi\in\calP$ and $g\in G$. Let $\Qacc^\prime:=\PSym\Qacc \PSym$ and $\mu(\cdot)$ denote the Haar measure on the group $G$. The above implies
    \begin{align*}
        a &\leq \inf_{\psi\in \calP} \int_{g\in G}\Tr(r(g^{-1})^{\otimes N} \Qacc^\prime r(g)^{\otimes N}\psi^{\otimes N})\rmd \mu (g)\\
        &= \inf_{\psi\in\calP}\Tr(\Pacc \psi^{\otimes N})
    \end{align*}
    where $\Pacc := \int_{g\in G} (r(g^{-1}))^{\otimes N}\Qacc^\prime r(g)^{\otimes N}\rmd \mu(g)\in \rmL(\vee^N(\calH))$. By the invariance of the Haar measure under left- or right- multiplication by a group element, it follows $\Pacc$ commutes with the action $r(g)^{\otimes N}$ for any $g\in G$ and thus, by \Cref{fact:commuting_implies_blocks}, it holds that $\Pacc$ is of the desired form. (In this case, the operators acting on the multiplicity spaces are just one-dimensional, i.e., scalars.) Next, note that if $\psi\in\rmP(\calH)$ is $\veps$-far from $\calP$ then $r(g)\psi r(g)^\dag$ is $\veps$-far from $\calP$ for any $g\in G$, by the unitary invariance of the Schatten 1-norm. This implies that
    \begin{align*}
        b\geq \Tr(\Qacc^\prime (r(g)\psi r(g)^\dag)^{\otimes N})
    \end{align*}
    for any $g\in G$ and $\psi$ which is $\veps$-far from $\calP$. Then a similar reasoning to the above establishes that $b\geq \sup_{\psi\ \veps\textnormal{-far from}\ \calP}\Tr(\Pacc\psi^{\otimes N})$.
\end{proof}
The final ingredient in the proof of \Cref{thm:purification} is the following.
\begin{lemma}\label{lem:double_sw_projectors}
    Let $\calH_A$ and $\calH_B$ be finite-dimensional Hilbert spaces and $\Pi_{\textnormal{Sym}}\in\rmL(\calH_A^{\otimes N}\otimes \calH_B^{\otimes N})$ be the projector onto the symmetric subspace $\vee^N(\calH_A\otimes \calH_B)$. It holds that
    \ba\label{eq:sym_subspace_proj_result}
        (P^{(A)}_{\lambda_1}\otimes P^{(B)}_{\lambda_2})\Pi_{\textnormal{Sym}} &= \delta_{\lambda_1\lambda_2}\Pi_{\lambda_1}\quad\textnormal{for every}\ \lambda_1,\lambda_2\vdash N, \ell(\lambda_1)\leq d_A,\ell(\lambda_2)\leq d_B
    \ea
    where, for each $j\in \{A,B\}$, we let $P^{(j)}_{\lambda}$ denote the projector onto the $\lambda^\text{th}$ $(\SG_N\times \rmU(\calH_j))$-irrep in $\calH_j^{\otimes N}$ and $\Pi_\lambda\in\rmL(\calH_A^{\otimes N}\otimes\calH_B^{\otimes N})$ is a projector onto the $(\rmU(\calH_A)\times \rmU(\calH_B))$-irrep isomorphic to $\calQ_{\lambda}(\calH_A)\otimes\calQ_{\lambda}(\calH_B)$ within $\vee^N(\calH_A^{\otimes N}\otimes\calH_B^{\otimes N})$.
\end{lemma}
\begin{proof}
    The space $\calH_A^{\otimes N}\otimes \calH_B^{\otimes N}$ is a representation of $\SG_N\times\rmU(\calH_A)\times\rmU(\calH_B)$ via the action $\ket{v}\mapsto (W_{\pi}\otimes W_{\pi})(U_A^{\otimes N}\otimes U_B^{\otimes N})\ket{v}$ for any $U_A\in \rmU(\calH_A)$, $U_B\in\rmU(\calH_B)$, and $\pi\in\SG_N$. By Schur-Weyl duality,
    \begin{align}
        \calH_A^{\otimes N}\otimes \calH_B^{\otimes N}&\cong\bigoplus_{\lambda_1^\prime,\lambda_2^\prime} \calP_{\lambda_1^\prime}\otimes \calP_{\lambda_2^\prime}\otimes \calQ_{\lambda_1^\prime}(\calH_A)\otimes \calQ_{\lambda_2^\prime}(\calH_B)\nonumber\\
        &\cong \bigoplus_{\mu}\bigoplus_{\lambda_1^\prime,\lambda_2^\prime}\mathbb{C}^{k_{\mu}(\lambda_1^\prime,\lambda_2^\prime)} \otimes \calP_\mu \otimes \calQ_{\lambda_1^\prime}(\calH_A)\otimes \calQ_{\lambda_2^\prime}(\calH_B)\label{eq:this}
    \end{align}
    as representations of $\SG_N\times \rmU(\calH_A)\times\rmU(\calH_B)$, where in the second line we decomposed $\calP_{\lambda_1^\prime}\otimes \calP_{\lambda_2^\prime}$ into irreps of $\mathfrak{S}_N$, and $k_{\mu}(\lambda_1^\prime,\lambda_2^\prime)$ are the Kronecker coefficients. Under this isomorphism, the projection onto the symmetric subspace is the projection onto terms in this direct sum where $\mu=(N)$, while the projection $P^{(A)}_{\lambda_1}\otimes P^{(B)}_{\lambda_2}$ is onto the terms in the direct sum with indices $\lambda_1^\prime=\lambda_1$, and $\lambda_2^\prime=\lambda_2$. Since these two projections are both block diagonal with respect to the orthogonal direct sum in~\cref{eq:this}, their product is the projection onto the intersection of their images, given by the subspace of $\calH_A^{\otimes N}\otimes \calH_B^{\otimes N}$ isomorphic to
    \begin{align}\label{eq:surep}
        \mathbb{C}^{k_{(N)}(\lambda_1,\lambda_2)}\otimes \calP_{(N)}\otimes \calQ_{\lambda_1}(\calH_A)\otimes \calQ_{\lambda_2}(\calH_B).
    \end{align}
We first note that $k_{(N)}(\lambda_1,\lambda_2)=\delta_{\lambda_1,\lambda_2}$. Let $\vee^N(\calP_{\lambda_1}\otimes\calP_{\lambda_2})$ denote the trivial representation of $\SG_N$ within $\calP_{\lambda_1}\otimes \calP_{\lambda_2}$, and let $\rmL(\calP_{\lambda_1},\calP_{\lambda_2})^{\SG_N}$ denote the set of linear maps in $\rmL(\calP_{\lambda_1},\calP_{\lambda_2})$ which commute with the action of $\SG_N$. By Schur's Lemma, we have $\dim(\rmL(\calP_{\lambda_1},\calP_{\lambda_2})^{\SG_N})=\delta_{\lambda_1,\lambda_2}$. Since $\vee^N(\calP_{\lambda_1}\otimes \calP_{\lambda_2})\cong \rmL(\calP_{\lambda_1},\calP_{\lambda_2})^{\SG_N}$ as representations of $\SG_N$, this shows that $k_{(N)}(\lambda_1,\lambda_2)=\delta_{\lambda_1,\lambda_2}$. Since $\calP_{(N)}$ is one-dimensional, the subspace in~\cref{eq:surep} is isomorphic to $\delta_{\lambda_1,\lambda_2} \calQ_{\lambda_1}(\calH_A)\otimes \calQ_{\lambda_2}(\calH_B)$ as a representation of $\rmU(\calH_A)\times \rmU(\calH_B)$. This completes the proof.
\end{proof}
The above proof also shows
\begin{align*}
    \vee^N(\calH_A\otimes \calH_B)&\cong \bigoplus_\lambda \calQ_{\lambda}(\calH_A)\otimes \calQ_{\lambda}(\calH_B)\tag{\ref{eq:schur_weyl_bipartite_symmetric}}
\end{align*}
as representations of $\rmU(\calH_A)\times \rmU(\calH_B)$, as mentioned in \Cref{sec:schur_weyl}. We can now prove \Cref{thm:purification}, which we restate below for convenience.
\newtheorem*{Tpur}{\Cref{thm:purification}}
\begin{Tpur}
    Let $\calH_A$, $\calH_B$ be finite-dimensional Hilbert spaces and $\calP \subset \rmP(\calH_A\otimes \calH_B)$ be a property of pure states which is invariant under local unitary transformations, as in \cref{eq:statement_of_symmetry}. Weak Schur Sampling performed locally, either on $\calH_A^{\otimes N}$ or $\calH_B^{\otimes N}$, followed by classical post-processing, is an optimal measurement for $\calP$.
\end{Tpur}
\begin{proof}
    By \Cref{prop:optimal_meas} and \cref{eq:schur_weyl_bipartite_symmetric}, without loss of generality a given test is of the form $\Pacc = \sum_{\lambda}\alpha_\lambda \Pi_{\lambda}$ for some $\alpha_{\lambda}\in [0,1]$, where $\Pi_{\lambda}$ is the projection onto the $\lambda^{\text{th}}$ term in the direct sum in \cref{eq:schur_weyl_bipartite_symmetric}. Then setting $R^{(j)}_{\textnormal{accept}}:=\sum_\lambda \alpha_\lambda P^{(j)}_\lambda$ for either $j=A$ or $j=B$, where $P^{(j)}_\lambda$ is as in \Cref{lem:double_sw_projectors}, we have
    \begin{align*}
        (R^{(A)}_{\textnormal{accept}}\otimes \mathds{1}_B^{\otimes N})\Pi_{\textnormal{Sym}} &= (\sum_{\lambda\vdash N}\alpha_{\lambda}P^{(A)}_{\lambda}\otimes \mathds{1}_B^{\otimes N})\Pi_{\textnormal{Sym}}&\\
        &= \sum_{\lambda,\lambda^\prime} \alpha_{\lambda}(P^{(A)}_{\lambda}\otimes P^{(B)}_{\lambda^\prime})\Pi_{\textnormal{Sym}}&\\
        &=\sum_{\lambda}\alpha_{\lambda}\Pi_{\lambda}& (\textnormal{by \Cref{lem:double_sw_projectors}})\\
        &= P_{\textnormal{accept}}&
    \end{align*}
    and similarly for the measurement operator $\mathds{1}_A^{\otimes N}\otimes R^{(B)}_{\textnormal{accept}}$.
\end{proof}
Finally, we prove \Cref{thm:PC_optimal_sr_test}, restated below for convenience.
\newtheorem*{pcsr}{\Cref{thm:PC_optimal_sr_test}}
\begin{pcsr}[PC-optimal Schmidt-rank test]
    Let $\calH_A,\calH_B$ be finite-dimensional Hilbert spaces such that $\dim(\calH_j)=d_j$ for each $j\in \{A,B\}$, $N$ and $r\leq d_A\leq d_B$ be positive integers, and $\PSym\in\rmL(\calH_A^{\otimes N}\otimes \calH_B^{\otimes N})$ denote the projector onto the symmetric subspace $\vee^N(\calH_A\otimes \calH_B)$.
    \begin{enumerate}
        \item It holds that
            \begin{align*}
                (\mathds{1}_{A}^{\otimes N}\otimes \Pi_{\leq r}^{(B)})\PSym &= (\Pi_{\leq r}^{(A)}\otimes \mathds{1}_{B}^{\otimes N})\PSym = \Proj \Span \{\ket{\phi}^{\otimes N}: \SR(\ket{\phi})\leq r\}
            \end{align*}
        is a strongly PC-optimal test for $\mathsf{SR}(r)\subseteq\rmP(\calH_A\otimes\calH_B)$.
        \item The copy complexity of testing $\mathsf{SR}(r)$ with one-sided error is $\Theta(r^2/\veps^2)$.
    \end{enumerate}
\end{pcsr}
\begin{proof}
We begin with the first item. We claim either one of $\Pi_{\leq r}\otimes \mathds{1}_B^{\otimes N}$ or $\mathds{1}_A^{\otimes N}\otimes \Pi_{\leq r}$ is a PC-optimal test for $\mathsf{SR}(r)$, which implies
\begin{align*}
    \PSym(\Pi_{\leq r}\otimes \mathds{1}_B^{\otimes N})\PSym &= \Proj \Span \{\ket{\phi}^{\otimes N}: \SR(\ket{\phi})\leq r\}
\end{align*}
by \Cref{lem:pc_opt_pure_proj_lemma}, and similarly for $\PSym(\mathds{1}_A^{\otimes N}\otimes \Pi_{\leq r})\PSym$. To see this, suppose for contradiction that there is a test $E$ which has perfect completeness and accepts with probability strictly less than the acceptance probability for the test $\Pi_{\leq r}\otimes \mathds{1}_B^{\otimes N}$ on some state $\psi\notin \mathsf{SR}(r)$. By \Cref{thm:purification} we can take $E=E_A\otimes\mathds{1}_B^{\otimes N}$ for some $E_A\in\rmL(\calH_A^{\otimes N})$. Then, if $\rho_A[\psi]$ denotes the reduced state of $\psi$ on $A$ we have
\begin{align*}
    \Tr(E_A\rho_A[\psi]^{\otimes N}) &= \Tr\left((E_A\otimes \mathds{1}_B^{\otimes N}) \psi^{\otimes N}\right) < \Tr\left((\Pi_{\leq r}\otimes \mathds{1}_B^{\otimes N}) \psi^{\otimes N}\right)= \Tr(\Pi_{\leq r}\rho_A[\psi]^{\otimes N})
\end{align*}
But then since $E_A$ is a test for $\mathsf{Rank}(r)$ with perfect completeness, this contradicts the strong PC-optimality of $\Pi_{\leq r}$ for the property $\mathsf{Rank}(r)$, which is Proposition~6.1 in \cite{o2015quantum}. It remains to show that $\PSym$ commutes with both $\Pi_{\leq r}\otimes \mathds{1}_B^{\otimes N}$ and $\mathds{1}_A^{\otimes N}\otimes \Pi_{\leq r}$. Using \Cref{lem:double_sw_projectors} we have
\begin{align*}
        (\Pi_{\leq r}\otimes \mathds{1}_{B}^{\otimes N})\PSym &= \sum_{\lambda_1: \ell(\lambda_1)\leq r}\ \sum_{\lambda_2: \ell(\lambda_2)\leq d_B} (P_{\lambda_1}\otimes P_{\lambda_2})\PSym\\
        &= \sum_{\lambda:\ell(\lambda)\leq r}\Pi_{\lambda}
\end{align*}
where the first line uses the definition of the rank test operator in \cref{eq:rank_test_op} and in the second line the projections $\Pi_{\lambda}$, defined explicitly in \Cref{lem:double_sw_projectors} as projectors onto irreps, have orthogonal images. Since a product of orthogonal projection operators commutes if and only if their product is also an orthogonal projection operator, we have shown the desired commutation relation for the operator $\Pi_{\leq r}\otimes \mathds{1}_{d_B}^{\otimes N}$. A similar argument applies to the operator $\mathds{1}_{d_A}^{\otimes N}\otimes \Pi_{\leq r}$.

For the second item note that, for testing $\mathsf{Rank}(r)$ at distance $\veps$, the rank test $\Pi_{\leq r}$ of O'Donnell and Wright succeeds using $O(r^2/\veps)$ copies, and also requires $\Omega(r^2/\veps)$ copies to be successful~\cite[Lemma~6.2]{o2015quantum}. Then, following the discussion on trace distances in \Cref{sec:sr_rank_testing_equivalence}, if $\psi$ is $\veps$-far from $\mathsf{SR}(r)$ the reduced state $\rho_A[\psi]$ on $A$ is $\veps^2$-far from $\mathsf{Rank}(r)$ and
\begin{align*}
    \Tr\left((\Pi_{\leq r}\otimes \mathds{1}_B^{\otimes N})\PSym\psi^{\otimes N}\right) &= \Tr(\Pi_{\leq r}\rho_A[\psi]^{\otimes N})
\end{align*}
is at most $1/3$ by taking $N=O(r^2/\veps^2)$, by the upper bound for rank testing. On the other hand, assume for contradiction that $(\Pi_{\leq r}\otimes \mathds{1}_B^{\otimes N})\PSym$ succeeds in testing $\mathsf{SR}(r)$ at distance $\veps$ using $o(r^2/\veps^2)$ copies. Then, if $\rho$ is $\veps^2$-far from $\mathsf{Rank}(r)$ the canonical purification $\phi^{\rho}$ is $\veps$-far from $\mathsf{SR}(r)$ and we have
\begin{align*}
    \Tr(\Pi_{\leq r}\rho)&= \Tr\left((\Pi_{\leq r}\otimes \mathds{1}_B^{\otimes N})\PSym (\phi^{\rho})^{\otimes N}\right)
\end{align*}
is at most $1/3$ by taking $N=o(r^2/\veps^2)$. But this contradicts the lower bound for rank testing.
\end{proof}

\subsection{Low-rank approximation for bipartite pure states}\label{sec:eckart_young_mirsky}
In this section we provide an elementary proof of \Cref{prop:eckart_young_mirsky}. We will use the following version of the Eckart-Young-Mirsky Theorem~\cite{Eckart1936,mirsky1960symmetric} for low-rank approximation with respect to unitarily invariant norms.
\begin{theorem}[{\cite[{Theorem~1.1}]{Li2020elementary}}]\label{thm:low_rank_approx}
    Let $\calH_1$ and $\calH_2$ be finite-dimensional Euclidean vector spaces and let $\norm{\cdot}$ be a unitarily invariant norm on $\rmL(\calH_2,\calH_1)$. Suppose $A\in \rmL(\calH_2,\calH_1)$ has a singular value decomposition $A=\sum_{j=1}^{\min\{d_1,d_2\}} \lambda_j \outerprod{u_j}{v_j}$. For any positive integer $r\leq \min\{d_1,d_2\}$, the matrix $A_r:=\sum_{j=1}^r \lambda_j\outerprod{u_j}{v_j}$ satisfies
    \begin{align*}
        \norm{A-A_r}\leq \norm{A-B}
    \end{align*}
    for any $B\in\rmL(\calH_2,\calH_1)$ with rank at most $r$.
\end{theorem}
Taking $A:=\sum_{j=1}^{d_1}\lambda_j\outerprod{u_j}{v_j}$ with $\lambda_j$, $\ket{u_j}$, and $\ket{v_j}$ as defined in \cref{eq:psi_schmidt_decomp}, the maximization in the statement of the proposition is equivalently expressed as
\begin{align}\label{eq:new_optimization}
    \max\{|\Tr(A^\dag B)|^2: B\in \rmL(\calH_2,\calH_1),\ \rank(B)\leq r,\ \norm{B}_{\mathrm{F}}=1\}
\end{align}
with the optimal solution of the original maximization being related to the matrix which attains the maximum in \cref{eq:new_optimization} through a vectorization of this matrix. It thereby remains to show that $\widetilde{B}:=(\sum_{k=1}^r\lambda_j^2)^{-1/2}\sum_{j=1}^r\lambda_j\outerprod{u_j}{v_j}$ is an optimal solution. To this end, consider the optimization
\begin{align}\label{eq:other_opt}
    \min\{\norm{A-X}_{\mathrm{F}}:X\in\rmL(\calH_2,\calH_1), \ \rank(X)\leq r\}.
\end{align}
By \Cref{thm:low_rank_approx} we have that the maximum is attained at $X=\sum_{j=1}^r\lambda_j\outerprod{u_j}{v_j}$. Now, suppose by way of contradiction that there exists an operator $B^\prime$ satisfying $\norm{B^\prime}_{\mathrm{F}}=1$, $\rank(B^\prime)\leq r$, and $\Tr(A^\dag B^\prime) = r\ee^{\ii\theta}$ for some $\theta\in [0,2\pi)$ and $r\in\mathbb{R}_+$ such that $r>\gamma:= \sqrt{\sum_{j=1}^r\lambda_j^2}$. Then we must have
\begin{align*}
    \norm{A-\gamma \ee^{-\ii\theta}B^\prime}_{\mathrm{F}}^2 &= \norm{A}_{\mathrm{F}}^2 + \gamma^2 - 2\gamma r \\
    &< \norm{A}_{\mathrm{F}}^2 - \gamma^2\\
    &= \sum_{j=r+1}^{d_1}\lambda_j^2.
\end{align*}
But this contradicts $X$ being an optimal solution to the optimization in \cref{eq:other_opt}. Thus, we have that $|\Tr(A^\dag B^\prime)|\leq \gamma$ for any choice of $B^\prime$ satisfying the constraints in \cref{eq:new_optimization}. This concludes the proof, since $|\Tr(A^\dag \widetilde{B})|^2=\gamma^2$.

\subsection{Faithful TTNS approximations}\label{sec:ttns_approx_proof}
In this section we include a proof of \Cref{thm:ttns_approximation}. The proof uses the following lemma.
\begin{lemma}[{Lemma~4.123 in \cite{Hackbusch2019}}]\label{lem:projector_distance}
Let $\ket{\psi}\in \calH$ be an arbitrary vector in the finite-dimensional Hilbert space $\calH$ and $P_1,\dots,P_N$ be a sequence of orthogonal projection operators acting on $\calH$. It holds that
\begin{align*}
   {\left\lVert(\mathds{1}-\prod_{i=1}^N\nolimits
   P_i)\ket{\psi}\right\rVert}_2^2\leq
   \sum_{i=1}^N{\left\lVert(\mathds{1}-P_i)\ket{\psi}\right\rVert}_2^2.
\end{align*}
\end{lemma}
We restate the theorem for convenience. Here, $V$ and $E$ refer to the vertices and edges, respectively, of a tree graph $G$. 
\newtheorem*{faithful}{\Cref{thm:ttns_approximation}}
\begin{faithful}
   Let $\ket{\psi}\in \bigotimes_{v\in V} \calH_v$ be an arbitrary pure
   state on $n=|V|$ sites. For each $e\in E$, let $L^{(e)}, R^{(e)}\subseteq [n]$ be
   the bipartition of the vertices induced by removing the edge
   $e$. Suppose further that $\ket{\psi}$ has the Schmidt decomposition
   \begin{align*}
       \ket{\psi} = \sum_j {\lambda_j^{(e)}}\
       \ket{a^{(e)}_j}_{L^{(e)}}\otimes \ket{b^{(e)}_j}_{R^{(e)}}
   \end{align*}
   with Schmidt coefficients $\lambda_1^{(e)}\geq
   \lambda_2^{(e)}\geq \dots$ for every $e\in E$, and that
   \begin{align}
       0\leq \sum_{e\in E}\sum_{j\geq r+1}(\lambda_j^{(e)})^2\leq 1.
   \end{align}
   Then there exists a state
   $\ket{\phi}\in\mathsf{TTNS}(G,r)$ for which
   \begin{align*}
       |\langle \phi | \psi \rangle|\geq 1-\sum_{e\in E}\sum_{j\geq r+1}(\lambda_j^{(e)})^2.
   \end{align*}
\end{faithful}
\begin{proof}
We claim there exists a (potentially unnormalized) vector $\ket{\phi}\in\bigotimes_{v\in V}\calH_v$ such that the following hold:
   \begin{enumerate}
       \item $\mathrm{SR}_e(\ket{\phi})\leq r\quad \forall e\in E$
       \item ${\left\lVert\ket{\phi}-\ket{\psi}\right\rVert}_2^2\leq\sum_{e\in E}\sum_{j\geq r+1} (\lambda_j^{(e)})^2$.
   \end{enumerate}
   The theorem then follows since for any $0\leq \delta\leq 1$ and normalized vector $\ket{\psi}$,
   ${\left\lVert\ket{\phi}-\ket{\psi}\right\rVert}_2^2\leq \delta$
   if and only if
   $2\ \mathrm{Re}\{\langle \phi | \psi\rangle\}\geq 1 +{\lVert\ket{\phi}\rVert}_2^2 - \delta$
   which implies
   \begin{align*}
       \frac{|\langle \phi |\psi\rangle|}{\lVert\ket{\phi} \rVert_2}\geq \frac{1 + {\lVert\ket{\phi}\rVert}_2^2 - \delta}{2\ \lVert\ket{\phi} \rVert_2} \geq 1-\delta.
   \end{align*}
   It remains to prove the two properties above. But these are implied by~\cite[Theorem 11.58]{Hackbusch2019}. We supply that argument in the language of this paper here for the interested reader, making use of the definitions established in \Cref{sec:ttns}. Suppose the tree $G$ has depth $D$ and root node $v_{\text{root}}\in V$. For each $v\in V$ let $G(v)\subseteq G$ be the subgraph corresponding to the vertex $v$ and all descendants of $v$. Note that $G(v_{\text{root}})=G$, and for any $v\in V$ we have that $G(v)$ is a tree. Also, let $h(v)\in \{0,\dots, D-1\}$ be the height of the vertex $v\in V$, and for each $j\in \{0,\dots, D-1\}$ define
   \begin{align*}
       L_j :=\{v\in V: h(v)=j\}
   \end{align*}
   i.e., the subset of vertices at height $j$. (For example, $L_0$ is the set of all leaves.) Consider a fixed vertex $v\in L_j$ and edge $(u,v)\in E$ such that $u\in G(v)$. Then the removal of the edge $e=(u,v)$ bipartitions the vertices into disjoint subsets $G(u)$ and $V\backslash G(u)$. Define the orthogonal projection operator
   \begin{align*}
       P^{(v)}_u:= \sum_{j=1}^r {|a^{(u,v)}_j\rangle\langle a^{(u,v)}_j|}_{G(u)}\otimes \mathds{1}_{[n]\backslash G(u)}
   \end{align*}
   where $\ket{a_j^{(u,v)}}$ are the Schmidt vectors of $\ket{\psi}$ for the subsystem $G(u)$ in the Schmidt decomposition for the edge $(u,v)$. For any pair of edges $(u_1,v_1),(u_2,v_2)\in E$ such that $v_1,v_2\in L_j$ and $u_1\in G(v_1)$, $u_2\in G(v_2)$ it holds that $[P^{(v_1)}_{u_1},P^{(v_2)}_{u_2}]=0$, since these operators act on disjoint subsystems, unless $(u_1,v_1)=(u_2,v_2)$. Hence,
   \begin{align*}
       \Gamma_j:=\prod_{\substack{(u,v)\in E\\ u\in G(v),\ v\in L_j}}P^{(v)}_u
   \end{align*}
   is an orthogonal projection operator and the order in the product does not affect the definition. See \Cref{fig:proj_illustration} for an illustration of this point. We then define the (unnormalized) tree tensor network approximation
   \begin{align}\label{eq:ttns_approx_def}
       \ket{\phi}:= \Gamma_0\Gamma_1\dots\Gamma_{D-1} \ket{\psi}
   \end{align}
where the order in the product now does matter, and which we claim satisfies the two properties mentioned above.

\begin{figure}
    \centering
    \includegraphics[width=0.5\linewidth]{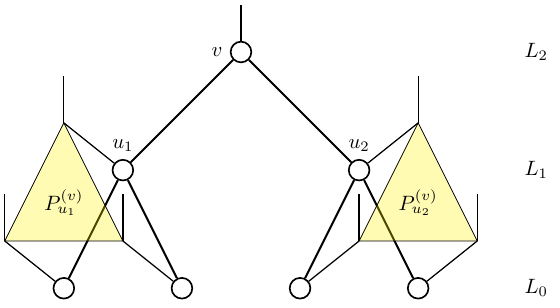}
    \caption{\label{fig:proj_illustration}Illustration of the orthogonal projection operator $\Gamma_2 = P^{(v)}_{u_1}P^{(v)}_{u_2}$ corresponding to the set of vertices at height $2$, which consists of only the root node $v$. The two projections in the product commute since they act non-trivially on disjoint subsystems.}
\end{figure}

For the first property we aim to show, that $\mathrm{SR}_e(\ket{\phi})\leq r$, consider a fixed edge $(u,v)\in E$ such that $v\in L_j$ for some $j\in \{0,\dots, D-1\}$, and such that removing this edge bipartitions the graph into disjoint subsets of vertices $G(u)$ and $G(u)\backslash V$. With respect to this bipartition, the state $\ket{\psi}$ has a Schmidt decomposition
\begin{align*}
    \ket{\psi}=\sum_{j}\lambda_j^{(u,v)} \ket{a_j^{(u,v)}}_{G(u)}\otimes \ket{b_j^{(u,v)}}_{G(u)\backslash V}
\end{align*}
We have
\begin{align*}
    \ket{\phi} &= \Gamma_{0}\dots\Gamma_{j-1} P^{(v)}_u \ket{\psi^\prime}
\end{align*}
for some (possibly unnormalized) vector $\ket{\psi^\prime}\in\bigotimes_{v\in V}\calH_v$. A moment of thought (perhaps employing the vectorization mapping) reveals that applying a rank-$r$ projection operator to the subsystem $G(u)$ yields a vector whose Schmidt-rank is at most $r$. Thus,
\begin{align*}
    \ket{\phi} = \Gamma_0\dots\Gamma_{j-1} \sum_{j=1}^r\beta_j \ket{r_j}_{G(u)}\otimes\ket{\ell_j}_{G(u)\backslash V}
\end{align*}
for some real coefficients $\beta_j$ and collections of orthogonal vectors $\{\ket{r_j}\}$ and $\{\ket{\ell_j}\}$. Now note that, by construction, for any $i < j$ it holds that $\Gamma_i=P\otimes Q$ for some orthogonal projection operators $P,Q$ acting solely within the subspaces corresponding to $G(u)$, $V\backslash G(u)$, respectively. Hence, $\mathrm{SR}_e(\ket{\phi})\leq r$ for the edge $(u,v)$, and this argument applies to any edge in the tree.

The second desired property follows from the fact that we have
\begin{align*}
    {\left\lVert\ket{\phi}-\ket{\psi}\right\rVert}_2^2&=
    {\left\lVert(\mathds{1}-\prod_{j=0}^{D-1}\nolimits\Gamma_j)\ket{\psi}\right\rVert}_2^2\\
    &\leq \sum_{(u,v)\in E}
    {\left\lVert(\mathds{1}-P^{(v)}_u)\ket{\psi}\right\rVert}_2^2&
    (\textnormal{\Cref{lem:projector_distance}})\\ &=\sum_{e\in
E}\sum_{j\geq r+1}(\lambda_j^{(e)})^2
\end{align*}
which concludes the proof.
\end{proof}

\subsection{\texorpdfstring{$\Omega(\sqrt{n})$}{Square root} lower bound at constant bond dimension}\label{sec:sqrt_lb}
In this section we prove a lower bound of $\Omega(\sqrt{n}/\veps^2)$ for testing $\mathsf{Prod}_2(n,r)$ at constant $r$ and distance $\veps\in (0,1/\sqrt{2}]$, even with two-sided error. The proof is essentially the same as that for the lower bound for testing MPS given in~\cite{soleimanifar2022testingmps}, but we include it for completeness. Throughout we assume $r\geq 2$ and $d-1\geq 2(r-1)$. Assume $\veps\leq 1/\sqrt{2}$, set $\theta=8\veps^2/n$, and for each $j\in [n/2]$ let $\ket{\phi_j^{(0)}}\in\mathbb{C}^d\otimes\mathbb{C}^d$ be a bipartite state with Schmidt coefficients $\sqrt{1-\theta}, \sqrt{\theta/(d-1)},\dots,\sqrt{\theta/(d-1)}$ and $\ket{\phi_j^{(1)}}\in\mathbb{C}^d\otimes\mathbb{C}^d$ be a bipartite state with Schmidt coefficients $\sqrt{1-\theta}, \sqrt{\theta/(r-1)},\dots,\sqrt{\theta/(r-1)},0,\dots,0$. We clearly have $\phi^{(1)}:=\bigotimes_{j\in [n/2]} \phi_j^{(1)}\in\mathsf{Prod}_2(n,r)$. On the other hand, $\phi^{(0)}:=\bigotimes_{j\in [n/2]}\phi_j^{(0)}$ is $\veps$-far from $\mathsf{Prod}_2(n,r)$. This is because the overlap between $\phi^{(0)}$ and the best approximation to this state in $\mathsf{Prod}_2(n,r)$ is, by \Cref{prop:eckart_young_mirsky} and the assumption that $d-1\geq 2(r-1)$, at most $(1-\theta/2)^{n/2}\leq 1-\veps^2$ where the inequality follows by bounding the remainder term in the Taylor Series of $(1-x)^q$ to deduce that $(1-x)^q\leq 1-qx/2$ for all $x\leq 1/q$. Hence, a successful test with two-sided error distinguishes between these two states with constant probability of success, given some number of copies $N$ of the unknown state prepared in either $\phi^{(0)}$ or $\phi^{(1)}$. Furthermore, this argument has no dependence on the choice of Schmidt vectors for the states $\ket{\phi_j^{(0)}}$ and $\ket{\phi_j^{(1)}}$. So the test also distinguishes between
\begin{align*}
    \bigotimes_{j\in [n/2]} (U_j\otimes V_j) \phi_1 (U_j^\dag \otimes V_j^\dag)\quad \textnormal{and}\quad  \bigotimes_{j\in [n/2]} (U_j\otimes V_j^\dag) \gamma_1 (U_j^\dag \otimes V_j^\dag)
\end{align*}
for any choice of $U_1,\dots, U_{n/2}\in \rmU(\mathbb{C}^d)$ and $V_1,\dots,V_{n/2}\in \rmU(\mathbb{C}^d)$, given $N$ copies of the unknown state. By the linearity of trace, the test succeeds in distinguishing between the two cases
\begin{align*}
    \sigma_x&=\expct\bigotimes_{j\in [n/2]}  \big[(\bm{U}_j \otimes \bm{V}_j)\phi_1^{(x)}(\bm{U}^\dag_j \otimes \bm{V}^\dag_j)\big]^{\otimes N}\quad \textnormal{for $x=0$ and $x=1$}
\end{align*}
where $\bm{U}_1,\dots,\bm{U}_{n/2}$ and $\bm{V}_1,\dots,\bm{V}_{n/2}$ are independent, Haar-random unitaries. We can rewrite these two mixed states as
\begin{align*}
    \sigma_x&= \rho_x^{\otimes n/2},\quad \rho_x := \expct_{\bm{U},\bm{V}\sim\textnormal{Haar}} \left((\bm{U} \otimes \bm{V})\phi_1^{(x)}(\bm{U}^\dag \otimes \bm{V}^\dag)\right)^{\otimes N}
\end{align*}
for each $x\in \{0,1\}$. By the well-known telescoping property of trace distance\footnote{The property follows from the inequality $\norm{\tau_1 \otimes \tau_2 - \omega_1\otimes\omega_2}_1\leq \norm{\tau_1\otimes \tau_2 - \omega_1\otimes \tau_2}_1 + \norm{\omega_1\otimes \tau_2-\omega_1\otimes\omega_2}_1$ and the identity $\norm{(\tau-\omega)\otimes \rho}_1=\norm{\tau-\omega}_1$.}, we have
\begin{align}\label{eq:eq_175}
    \dTr(\sigma_0,\sigma_1)\leq \frac{n}{2}\dTr(\rho_0,\rho_1)
\end{align}
where $\dTr(\rho,\sigma):=\frac{1}{2}\norm{\rho-\sigma}_1$.
By Theorem~35 in \cite{soleimanifar2022testingmps}, it holds that
\begin{align}\label{eq:eq_177}
    \dTr(\rho_0,\rho_1) &= \dTr\left(\expct (\bm{U}\tau_0 \bm{U}^\dag)^{\otimes N},\ \expct (\bm{U}\tau_1 \bm{U}^\dag)^{\otimes N}\right)
\end{align}
where $\bm{U}$ is Haar-random and $\tau_x$ is the reduced state of $\phi_1^{(x)}$ on the first subsystem for each $x\in \{0,1\}$. Expanding the tensor products and using the unitary invariance of the Haar measure, we have
\begin{align*}
    \expct (\bm{U}\tau_x\bm{U}^\dag)^{\otimes N} &= (1-\theta)^N \expct (\bm{U}\outerprod{1}{1}\bm{U}^\dag)^{\otimes N}\nonumber\\
    &\qquad + \theta (1-\theta)^{N-1} \sum_{i=0}^{N-1} \expct \bm{U}^{\otimes N}|\underbrace{11\dots 1}_{\times i}21\dots 1\rangle\langle 11\dots 1 21\dots 1 |(\bm{U}^\dag)^{\otimes N}\nonumber\\
    &\qquad\qquad + (1-p)\omega_x
\end{align*}
where $p:=(1-\theta)^N + N(1-\theta)^{N-1}\theta$ and $\omega_x\in\rmD((\mathbb{C}^d)^{\otimes N})$ is some density matrix, for each $x\in\{0,1\}$. Note that the first two terms have no dependence on $x$. From this, one sees the trace distance on the right-hand side of \cref{eq:eq_177} equals
\begin{align*}
    \frac{1}{2}\norm{p\tau + (1-p)\omega_0 - \left(p\tau +(1-p)\omega_1\right)}_1&= (1-p)\dTr(\omega_0,\omega_1)
\end{align*}
where $\tau$ is some density matrix. Since trace distances are at most one, the right-hand side is at most
\begin{align*}
    1-(1-\theta)^N-N(1-\theta)^{N-1}\theta&\leq 1- (1-N\theta)-N\theta(1-(N-1)\theta)\\
    &= N(N-1)\theta^2
\end{align*}
and so the right-hand side of \cref{eq:eq_175} is at most $\frac{n}{2} N(N-1) \left(\frac{8\veps^2}{n}\right)^2=O(N^2\veps^4/n)$, which is $o(1)$ unless $N=\Omega(\sqrt{n}/\veps^2)$.

\section{Miscellaneous facts}\label{sec:chernoff_bound}
\begin{lemma}[Chernoff bound]\label{lem:chernoff_bound}
    Let $\bm{X}_1,\dots,\bm{X}_n$ be independent Bernoulli random variables and define $\bm{X}=\sum_{i=1}^n\bm{X}_i$ and $\mu=\expct \bm{X}$. It holds that
    \begin{equation}
            \Pr[\bm{X}> \mu(1+t)]\leq \ee^{\frac{-\mu t^2}{2+t}}
    \end{equation}
    for any $t\geq 0$.
\end{lemma}
\begin{fact}\label{fact:sum_1}
    $\sum_{n=1}^\infty (n+1)^2 x^{n-1}=\frac{x^2-3x+4}{(1-x)^3}$ for $x\in\mathbb{R}$ such that $|x|<1$.
\end{fact}
\begin{fact}\label{fact:sum_2}
    $\sum_{n=3}^{\infty}(n+1)x^n = \frac{x^3(4-3x)}{(1-x)^2}$ for $x\in\mathbb{R}$ such that $|x|<1$.
\end{fact}
\begin{fact}\label{fact:sum_3}
    Let $N$ be a positive integer. For any integer $k\geq 1$ it holds that $\sum_{j=k}^N\binom{N}{j}x^j\leq \sqrt{e} N^k x^k$ for $x\in\mathbb{R}$ such that $0\leq x\leq 1/(2N)$.
\end{fact}
\begin{proof}
    Observe
    \begin{align*}
        \sum_{j=k}^N\binom{N}{j}x^j&\leq  N^kx^k\sum_{j=0}^{N-k} \frac{N^{j}x^{j}}{(j+k)!}
        \leq N^kx^k\sum_{j=0}^\infty\frac{N^jx^j}{j!}
        \leq N^kx^k\sum_{j=0}^\infty \frac{(1/2)^{j}}{j!}
        =\sqrt{e}N^kx^k
    \end{align*}
    as required.
\end{proof}

\printbibliography

\end{document}